\documentclass[reqno,a4paper]{amsart}
\usepackage{float}
\usepackage[dvipsnames]{xcolor}
\usepackage{graphicx}
\usepackage{tikz}
\usepackage{tkz-graph}
\usetikzlibrary{graphs}
\usetikzlibrary{graphs.standard}
\usepackage{pgf}
\usetikzlibrary{shapes,cd,arrows,automata,decorations.pathreplacing,angles,quotes,calc}
\usepackage{tkz-euclide}
\usepackage[matrix,arrow,frame,curve]{xy}
\usepackage{amsmath,amssymb}

\usepackage{enumitem}
\usepackage{todonotes}
\usepackage{url}
\usepackage{longtable}
\usepackage[useregional]{datetime2}
\usepackage{tikz-cd}
\newtheorem{theorem}{Theorem}[section]
\newtheorem{proposition}[theorem]{Proposition}
\newtheorem{corollary}[theorem]{Corollary}
\newtheorem{lemma}[theorem]{Lemma}
\theoremstyle{definition}
\newtheorem{remark}[theorem]{Remark}
\newtheorem{definition}[theorem]{Definition}
\newtheorem{example}[theorem]{Example}
\numberwithin{equation}{section}
\numberwithin{figure}{section}
\numberwithin{table}{section}
\newcommand\Pone{\textrm{P}_{\textrm{I}} }
\newcommand\Ptwo{\textrm{P}_{\textrm{II}} }
\newcommand\Ptwojm{\Ptwo^{\textrm{JM}}}
\newcommand\Ptwofn{\Ptwo^{\textrm{FN}}}

\newcommand\Pthree{\textrm{P}_{\textrm{III}}}
\newcommand\Pfour{\textrm{P}_{\textrm{IV}}}
\newcommand\Pfive{\textrm{P}_{\textrm{V}}}
\newcommand\Psix{\textrm{P}_{\textrm{VI}}}
\usepackage[hidelinks]{hyperref}
\usepackage{caption} 
\usepackage{subcaption}
\newcommand{\ns}{\scalebox{1.25}[1.0]{\( - \)}}
\newcommand{\tcr}{\textcolor{red}}
\newcommand{\tcb}{\textcolor{blue}}
\newcommand{\tcm}{\textcolor{black}}
\newcommand{\orcidauthorNJ}{0000-0001-7504-4444}
\newcommand{\orcidauthorMM}{0000-0001-9917-2547}
\newcommand{\nn}{\nonumber}
\allowdisplaybreaks
\begin{document}

\title{Segre surfaces and geometry of the Painlev\'e equations}
\author{Nalini Joshi}
\address{School of Mathematics and Statistics F07, The University of Sydney, New South Wales 2006, Australia.}
\email{nalini.joshi@sydney.edu.au}
\thanks{NJ's ORCID is \orcidauthorNJ, MM's is \orcidauthorMM. PR's is 0000-0002-0461-7580}
\author{Marta Mazzocco}
\address{School of Mathematics, University of Birmingham, UK.}
\email{m.mazzocco@bham.ac.uk}
\author{Pieter Roffelsen}
\address{School of Mathematics and Statistics F07, The University of Sydney, New South Wales 2006, Australia.}
\email{pieter.roffelsen@sydney.edu.au}
\begin{abstract}
In this paper, we consider a six parameter family of affine Segre surfaces embedded in $\mathbb C^6$. For generic values of the parameters, this family is associated to the $q$-difference sixth Painlev\'e equation. We show that different limiting forms of this family give Segre surfaces that are isomorphic as affine varieties to the monodromy manifolds of each  Painlev\'e differential equation.
\end{abstract}

\maketitle
\setcounter{tocdepth}{1}
\tableofcontents

\section{Introduction}
In this paper, we consider a family of affine Segre surfaces defined, for given $\mu_1,\ldots,\mu_6\in \mathbb C,\lambda_1,\lambda_2\in \mathbb C^*$, by
\begin{equation}\label{eq:segqp6}
\begin{cases}
&{z}_{1} + {z}_{2} + {z}_{3} + {z}_{4} + {z}_{5} + {z}_{6}=0,\\
&\mu_{1}{z}_{1} +\mu_{2} {z}_{2} + \mu_{3} {z}_{3} + \mu_{4} {z}_{4} + \mu_{5} {z}_{5} + 
  \mu_{6} {z}_{6} - 1=0,\\
&{z}_{3} {z}_{4} - {z}_{1} {z}_{2} \lambda_{1}=0,\\
&{z}_{5} {z}_{6} - {z}_{1} {z}_{2} \lambda_{2}=0.
\end{cases}
\end{equation}
We show that limiting forms of this family give Segre surfaces as monodromy manifolds of the Painlev\'e equations. This result is compelling on two levels. First, the well-known monodromy manifolds of the differential Painlev\'e equations are cubic surfaces.
Second, the geometry of each Segre surface turns out to have a deep connection to distinguished solutions of the corresponding Painlev\'e equation and provides a natural Poisson structure. 

The generic family of surfaces \eqref{eq:segqp6}, which we will denote in this paper by $\mathcal Z_q$, is associated with a \textit{difference} equation known as the $q$-difference sixth Painlev\'e equation $q\Psix$ \eqref{eq:qpvi}.  \tcm{Namely, for generic parameter values, the monodromy manifold of $q\Psix$ can be identified with \eqref{eq:segqp6}, see \cite[Theorem 2.20]{jr_qp6}, with explicit formulas for the coefficients given in Section \ref{subsec:parametrisation}.
This $q$-difference equation was introduced by Jimbo-Sakai \cite{jimbosakai} and they showed that the continuum limit of this difference equation becomes the celebrated sixth differential Painlev\'e equation $\Psix$ \eqref{eq:pvi}.}  

Like $\Psix$, $q\Psix$ has been associated with conformal field theory and algebraic geometry. \tcm{In particular, the latter 
encodes discrete dynamics of $q$-deformed conformal blocks \cite{JNS}.}

\tcm{ $q\Psix$ } has also been related to the classical study of linear $q$-difference equations, series expansions and connection problems; see \cite{jimbosakai}, \cite{manoqpvi:2010}, \cite{ORS20}, \cite{jr_qp6} and references therein.

Although the equations \eqref{eq:segqp6} defining $\mathcal Z_q$ appear to contain eight parameters, in Corollary \ref{cor:6dim} we show  that, up to affine equivalence, they define a six-parameter family of embedded affine Segre surfaces in 
$\mathbb C^6$. 
It is worth noting that $\mathcal Z_q$ arises from the Riemann-Hilbert problem associated with $q\Psix$. See \cite{jr_qp6} for  details of its explicit derivation. We prove that the members of the $\mathcal Z_q$ family are embedded affine Segre surfaces of generic type. This in particular fills a gap stated as an open problem in \cite{RSpreprint}, see Remark \ref{remark:mostgeneric}.

A large amount of work has been devoted to the differential Painlev\'e equations. For $(\omega_1,\omega_2,\omega_3,\omega_4)\in \mathbb C^4$, the monodromy manifold of $\Psix$ is well-known to be given by the  Jimbo-Fricke family of affine cubic surfaces  \cite{FrickeKlein1897}, \cite{J} defined by
\begin{equation}\label{eq:fksurf}
\begin{cases}
x_1x_2x_3
+x_1^2+x_2^2+x_3^2
+\omega_1 x_1+\omega_2x_2+\omega_3 x_3+\omega_4=0, \\(x_1,x_2,x_3)\in\mathbb C^3 .
\end{cases}
\end{equation}
Monodromy manifolds of the remaining Painlev\'e equations have been unified and collected in \cite{SvdP} and in \cite{CMR} it was shown how to obtain them as limits of \eqref{eq:fksurf}, corresponding to confluence limits of the Painlev\'e equations. All such monodromy manifolds are cubic surfaces whose projective completion has a triangle of lines at infinity.

Since the continuum limit of  $q\Psix$ leads to $\Psix$, and confluence limits of the latter give rise to the remaining Painlev\'e equations, a natural question is whether Segre surfaces also arise as monodromy manifolds of the Painlev\'e differential equations. Motivated by this question, we show how to obtain explicit Segre surfaces in $\mathbb C^6$  for each Painlev\'e equation, from appropriately \tcm{parametrised} coordinates and coefficients. They all have the same form:
 \begin{subequations}\label{gen-segre}
\begin{align}
&\epsilon_1 {z}_{1} +\epsilon_2  {z}_{2} + {z}_{3} + \epsilon_4{z}_{4} +{z}_{5} +\epsilon_6 {z}_{6}=0\\
& \rho_{3} {z}_{3} +  {z}_{4} + \rho_{5} {z}_{5} + 
  \rho_{6} {z}_{6} - 1=0\\
&{z}_{3} {z}_{4} - {z}_{1} {z}_{2} \lambda_{1}=0\\
&{z}_{3} {z}_{4} - {z}_{5} {z}_{6} \lambda_{2}=0,
\end{align}
\end{subequations}
with different choices of the parameters $\epsilon_i$ and $\rho_j$ according to the chosen Painlev\'e equation.  We call these surfaces $\mathcal Z$--Segre surfaces. They are listed in Table \ref{tb:Z-Segre}.

\begin{longtable}{|c || c |} 
\caption{Affine Segre surfaces for $q\Psix$ and all differential Painlev\'e equations.}\label{tb:Z-Segre}\\
 \hline
$\begin{array}{c}
    \hbox{Painlev\'e} \\
    \hbox{ equation}\\
\end{array} $  &  $\mathcal Z$--Segre surface \\[10pt]
 \hline
$q\Psix$ & $\begin{array}{l}
 {z}_{1} +  {z}_{2} + {z}_{3} + {z}_{4} + {z}_{5} +  {z}_{6}=0,\\[2pt]
\rho_{2} {z}_{2}+\rho_{3} {z}_{3} +  {z}_{4} + \rho_{5} {z}_{5} + 
  \rho_{6} {z}_{6} = 1,\\[2pt]
{z}_{3} {z}_{4} -\lambda_{1} {z}_{1} {z}_{2} =0,\quad
z_5 z_6 - \lambda_2  z_1  z_2=0.
\end{array} $\\[10pt]
 \hline
$\Psix$ & $\begin{array}{l}
 {z}_{1} +  {z}_{2} + {z}_{3} + {z}_{4} + {z}_{5} + {z}_{6}=0,\\[2pt]
\rho_{3} {z}_{3} +  {z}_{4} + \rho_{5} {z}_{5} + 
  \rho_{6} {z}_{6} - 1=0,\\[2pt]
{z}_{3} {z}_{4} -\lambda_{1} {z}_{1} {z}_{2}=0,\quad
z_5 z_6 - \frac{\rho_3\lambda_1}{\rho_5 \rho_6}  z_1  z_2 =0.
\end{array} $\\[10pt]
 \hline
 $\Pfive$ & $\begin{array}{l}
  {z}_{1} +   {z}_{2} + {z}_{3} + {z}_{4} + {z}_{5} +  {z}_{6}=0,\\
   {z}_{4} + \rho_{5} {z}_{5}  - 1=0,\\
{z}_{3} {z}_{4} -\lambda_{1} {z}_{1} {z}_{2} =0,\quad
z_5 z_6 - \lambda_2  z_1  z_2=0.
\end{array} $\\[10pt]
 \hline
$\Pfive^{\text{deg}}$   & $\begin{array}{l}
 {z}_{1} + {z}_{3} + {z}_{4} + {z}_{5} +{z}_{6}=0,\\
\rho_{3} {z}_{3} +  {z}_{4} + \rho_{5} {z}_{5} + 
  \frac{\rho_3}{\rho_5} {z}_{6} - 1=0,\\
{z}_{3} {z}_{4} - {z}_{1} {z}_{2} =0,\quad
z_5 z_6 -   z_1  z_2=0.
\end{array} $\\[10pt]
 \hline
 $\Pfour$ & $\begin{array}{l}
 {z}_{1} +  {z}_{2} + {z}_{3} + {z}_{4} + {z}_{5} +  {z}_{6}=0,\\
   {z}_{4}  - 1=0,\\
{z}_{3} {z}_{4} - \lambda_{1}{z}_{1} {z}_{2}=0,\quad
z_5 z_6 - \lambda_2  z_1  z_2=0.
\end{array} $\\[10pt]
 \hline
 $\Pthree^{D_6}$ & $\begin{array}{l}
 {z}_{1} +  {z}_{2} + {z}_{3} + {z}_{4} + {z}_{5} =0,\\
  {z}_{4} + \rho_{5} {z}_{5} - 1=0,\\
{z}_{3} {z}_{4} - \lambda_{1}{z}_{1} {z}_{2} =0,\quad
z_5 z_6 -   z_1  z_2=0.
\end{array} $\\[10pt]
 \hline
 $\Pthree^{D_7}$ & $\begin{array}{l}
 {z}_{1} +  {z}_{2} + {z}_{3} + {z}_{4} + {z}_{5}  =0,\\
   {z}_{4} + \rho_{5} {z}_{5}   - 1=0,\\
{z}_{3} {z}_{4} - {z}_{1} {z}_{2}  =0,\quad
z_5 z_6 - z_1  z_2=0.
\end{array} $\\
 \hline
$\Ptwojm,\, \Ptwofn$ & $\begin{array}{l}
 {z}_{1} +{z}_{2} + {z}_{3} + {z}_{4} + {z}_{5} +{z}_{6}=0,\\
  {z}_{4}  - 1=0,\\
{z}_{3} {z}_{4} - {z}_{1} {z}_{2} =0,\quad
z_5 z_6 - \lambda_2  z_1  z_2=0.
\end{array} $\\
 \hline
$\Pone$ & $\begin{array}{l}
  {z}_{3} + {z}_{4} + {z}_{5} + {z}_{6}=0,\\
  {z}_{4}  - 1=0,\\
{z}_{3} {z}_{4} - {z}_{1} {z}_{2}  =0,\quad
{z}_{3} {z}_{4} - {z}_{5} {z}_{6} =0.
\end{array} $\\
 \hline
\end{longtable} 

We also show that the resulting Segre surfaces are isomorphic to the corresponding cubic surfaces as affine varieties, for each  Painlev\'e equation. To achieve this, first we calculate the continuum limit $\mathcal Z_1$ of the Segre surface $\mathcal Z_q$ and show that it is given by \eqref{gen-segre} with generic parameters defined by equations \eqref{eq:mu-cl} and \eqref{eq:la-cl}. Our first main result is the following theorem.

\begin{theorem}\label{thm:isomorphism_intro}
    The Jimbo-Fricke cubic surface for \emph{$\Psix$}
    and the Segre surface $\mathcal Z_1$ are isomorphic as affine varieties, with parameters related by equations \eqref{eq:mu-cl}, \eqref{eq:la-cl}, \eqref{eq:pvi-om} and \eqref{eq:pvi-om1}.
\end{theorem}

In Remark \ref{remark:confluence} we show that this theorem answers certain questions asserted in \cite{RSpreprint}.
The proof of Theorem \ref{thm:isomorphism_intro} relies on two main ingredients. One is the comparison of the asymptotic behaviours of solutions of $\Psix$ and $q\Psix$ as $q\rightarrow 1$. The other is based on blow-downs of lines at infinity on the cubic surface.


Using an explicit form of this isomorphism (see Theorem \ref{thm:isomorphism}), we
show that, in all un-ramified cases, the confluence scheme of the Painlev\'e monodromy manifolds can be carried through to the affine transformation, therefore producing isomorphisms between each cubic surface 
and the $\mathcal Z$--Segre obtained by confluencing $\mathcal Z_1$. 

For the ramified cases, the confluence either produces a reducible Segre surface or a family which does not have the correct number of free parameters. This is not surprising because the confluences to ramified and to non-ramified Painlev\'e equations are deeply different in geometric as well as analytic terms. 

For this reason, we perform an in--depth study of the lines and singularities of all cubic surfaces with a
triangle of lines at infinity, as well as their blow downs,
which allows us to build
the isomorphic $\mathcal Z$--Segre for ramified cases as well.

\begin{theorem}
    The monodromy manifold of each differential  Painlev\'e equation (except $\Pthree^{D_8}$) is isomorphic to the corresponding $\mathcal Z$--Segre surface reported in Table \ref{tb:Z-Segre} as affine varieties. 
\end{theorem}

\begin{remark}
    It is not unusual for $\Pthree^{D_8}$ to be an  exceptional case here. It is often the exception in many studies of the Painlev\'e equations in the literature. 
    For example, in \cite{BCR}, Lax pairs were proposed for all the generalised multi-particle Painlev\'e equations, and $\Pthree^{D_8}$ is the only one for which the isomonodromic Hamiltonian is rational rather than polynomial. 
\end{remark}

The study of lines not only helps in building the above mentioned isomorphism, but it has its own interest and a long history in algebraic geometry. The famous Cayley-Salmon theorem \cite{cay1849,sal1849} shows that there are 27 lines on a general smooth cubic surface in $\mathbb C\mathbb P^3$. For the highly transcendental solutions of the Painlev\'e equations, these lines give vital information about certain distinguished behaviours. For example, for the first Painlev\'e equation ($\Pone$), lines on the cubic surface correspond to the tronqu\'ee solutions \cite[Theorem 3]{kapaev_kitaev} and in the case of $\Psix$, they correspond to truncations of the generic asymptotic expansions, see  \cite[Prop. 4.5]{klimes} and \cite[Table 1]{guzzetti_tabulation}. On the other hand, geometric properties of the Jimbo-Fricke surfaces and its confluence limits make them important objects in many areas of mathematics, including Cherednik algebras \cite{Obl,MM}, mirror symmetry \cite{GHKS}, Calabi-Yau algebras \cite{CMR1}, and moduli spaces \cite{SvdP}.

Another perspective arises in Okamoto's characterization of the initial-value space and symmetry group of each Painlev\'e equation \tcm{\cite{N04,O79,sakai}} in terms of divisors,  intersection theory and the corresponding interpretation of such divisors as generators of affine Weyl groups.  In this context, the corresponding limits of the Painlev\'e equations are realized as coalescences of base points. 

The lattices associated with extended affine Weyl or Coxeter groups \cite{H92,C48,G67} lead naturally to the construction of discrete Painlev\'e equations, with each step of iteration being a translation on the lattice, mapping coordinates in one tile (or polytope) in the periodic lattice to another such tile. Multiplicative or $q$-difference Painlev\'e equations arise from such operations on a hyperbolic lattice. 

In the case of the differential Painlev\'e equations, the monodromy manifold is ``dual'' to the initial value space under the Riemann--Hilbert correspondence. For $q$-discrete Painlev\'e  equations, a notion of monodromy manifold was missing until recently \tcm{\cite{jr_qp6,jr_qp4,ORS20}}. By constructing a Segre surface for each differential Painlev\'e equation, we  provide the start of a uniform setting that naturally extends to a description of the monodromy manifold of their $q$-difference analogues. 

\tcm{The coordinate ring of a Painlev\'e monodromy manifold is endowed with a natural Poisson bracket \cite{nambu}, that coincides with the one induced by the natural symplectic structure \cite{tak} on the surface.} By a suitable deformation quantization of this Poisson algebra, the spherical sub-algebras of certain confluent versions of the Cherednik algebra of type $\check C_1C_1$ were obtained in \cite{MM}. This led to a representation--theoretic approach to the theory of the Painlev\'e equations and to surprising links with the theory of basic hypergeometric polynomials. 
In this paper, we show that the coordinate rings of the Segre surfaces arising as blow downs of such monodromy manifolds also carry a natural Poisson structure \cite{tak,OdR}. Looking beyond the current paper, it would be interesting to quantise these Segre surfaces  and understand their role in the theory of basic hypergeometric polynomials as well as the relation between such quantizations and the mononodromy manifolds of the corresponding multiplicative discrete Painlev\'e  equations.



\subsection{Terminology}
In this paper, we always work over the field $\mathbb{C}$, and denote the $n$-dimensional complex projective space, $n\geq 1$, by $\mathbb{P}^n=\mathbb{CP}^n$. We also use the notation $\mathbb{C}^*:=\mathbb{C}\setminus\{0\}$.
Almost all algebraic varieties in this paper are affine varieties embedded in  $\mathbb{C}^n$, or projective varieties embedded in  $\mathbb{P}^n$, for a suitable positive integer $n$.

We use small roman letters for affine coordinates, e.g.
\begin{equation*}
    \mathbb{C}^n=\{(x_1,x_2,\ldots,x_n)\in\mathbb{C}^n\},
\end{equation*}
and capitalised roman letters for homogeneous coordinates, e.g.
\begin{equation*}
    \mathbb{P}^n=\{[X_0:X_1:X_2:\ldots:X_n]\in\mathbb{P}^n\}.
\end{equation*}
We consider $\mathbb{C}^n$ as a subset of $\mathbb{P}^n$ in the natural way, through the identification
\begin{equation}\label{eq:CPembedding}
    [X_0:X_1:X_2:\ldots:X_n]=[1:x_1:x_2:\ldots:x_n].
\end{equation}
Given an embedded affine variety $V\subseteq\mathbb{C}^n$, we define its canonical projective completion $\overline{V}\subseteq \mathbb{P}^n$  as the completion of the image of  $V$ in $\mathbb{P}^n$ through \eqref{eq:CPembedding}.

Given two embedded affine varieties $V\subseteq \mathbb{C}^m$ and $W\subseteq \mathbb{C}^n$, with $m\le n$,  we say that $V$ and $W$ are \textit{affinely equivalent} if there exists an affine linear mapping $\phi:\mathbb{C}^m\rightarrow \mathbb{C}^n$ of maximal rank $m$ that restricts to an isomorphism $V\rightarrow W$. In this case, we also call $W$ and $V$ affinely equivalent.
In particular, given an affine variety $V\subseteq \mathbb{C}^m\subseteq  \mathbb{C}^{m+n}$, with $m,n\geq 1$, where $\mathbb{C}^m$ lies linearly in $\mathbb{C}^{m+n}$,
the two views of $V$ as an embedded variety in $\mathbb{C}^m$ or $\mathbb{C}^{m+n}$ are affinely equivalent.
Note that affine equivalence is generally stronger than isomorphic as affine varieties.

Projective equivalence between embedded projective varieties is defined similarly.

\subsection{Acknowledgements}
This work was funded by the Leverhulme Trust visiting professorship grant VP2-2018-013.
Part of the work was carried out during NJ's and PR's residence at the Isaac Newton Institute for Mathematical Sciences in the program ``Applicable resurgent asymptotics: towards a universal theory'' during September-December 2022. We are grateful for the hospitality of the organizers of this program and the support of the Newton Institute. MM's research was supported by the Leverhulme Trust Research Project Grant RPG-2021-047. PR and NJ's research was supported by the Australian Research
Council Discovery Project \#DP210100129. The authors thank Davide Guzzetti, Oleg Lisovyy, Volodya Rubtsov for fruitful discussions \tcm{and the referees for carefully reading this manuscript and providing many helpful suggestions}.

\tcm{
\subsection{Data Availability Statement}
No datasets were analysed or generated during this study.}

\section{The affine Segre surface of the $q\Psix$ equation}\label{sec:qpsixsegre}

For simplicity, we refer to the Segre surface \eqref{eq:segqp6} as \begin{equation}\label{qp6:seg}
    \mathcal{Z}_q:=\{z\in\mathbb{C}^6:h_i=0\text{ for }1\leq i\leq 4\},
\end{equation}
where for given $\mu_1,\ldots,\mu_6,\lambda_1,\lambda_2\in \mathbb C$, the polynomials $h_i\in\mathbb C[z_1,\dots,z_6]$, $1\leq i\leq 4$, are defined by
\begin{subequations}\label{qp6:segh}
\begin{align}
&h_1={z}_{1} + {z}_{2} + {z}_{3} + {z}_{4} + {z}_{5} + {z}_{6},\\
&h_2=\mu_{1}{z}_{1} +\mu_{2} {z}_{2} + \mu_{3} {z}_{3} + \mu_{4} {z}_{4} + \mu_{5} {z}_{5} + 
  \mu_{6} {z}_{6} - 1,\\
&h_3={z}_{3} {z}_{4} - {z}_{1} {z}_{2} \lambda_{1},\\
&h_4={z}_{5} {z}_{6} - {z}_{1} {z}_{2} \lambda_{2},
\end{align}
\end{subequations}
where $\mu_k$, $1\leq k\leq 6$ and $\lambda_1,\lambda_2$, are some complex numbers.

We show that a generic member of the family $\mathcal{Z}_q$ is an affine Segre surface, with smooth canonical projective completion, and conversely, that a generic embedded affine Segre surface with smooth canonical projective completion, can be put into the form \eqref{qp6:seg} via an affine transformation.

The family of affine surfaces  $\mathcal{Z}_q$ first appeared in \cite{jr_qp6}
as the monodromy manifold of the linear problem for $q\Psix$. 
In that case, the coefficients  $\mu_1,\dots,\mu_6$ and $\lambda_1,\lambda_2$ are parametrised explicitly in terms of the six parameters of the $q\Psix$ equation. 
The fact that only six of the eight coefficients  the coefficients  $\mu_1,\dots,\mu_6$ and $\lambda_1,\lambda_2$ are independent is a consequence of the intrinsic scaling freedoms in the defining equations \eqref{qp6:segh}. 

\begin{lemma}\label{lm:parindep}
  The family of affine surfaces $\mathcal{Z}_q$ is equivalent to the six parameter family defined as the zero set in $\mathbb C^6$ of $h_1, h_3,h_4$ and $h_2''$ where
  \begin{equation}\label{eqp-h2-rho}
h_2''=\rho_2 {z}_{2} + \rho_{3} {z}_{3} + {z}_{4} + \rho_{5} {z}_{5} + 
  \rho_{6} {z}_{6} - 1.
\end{equation}
\end{lemma}

\begin{proof}
 Using the freedom of adding arbitrary multiples of $h_1$ to $h_2$, we may normalise $h_2$ so that it has no $z_1$ term,
\begin{equation*}
h_2'=(\mu_{2}-\mu_1){z}_{2} + (\mu_{3}-\mu_1) {z}_{3} + (\mu_{4}-\mu_1) {z}_{4} + (\mu_{5}-\mu_1) {z}_{5} + 
 (\mu_{6}-\mu_1) {z}_{6} - 1.
\end{equation*}
Moreover, using the freedom of scaling $z_k\mapsto c\, z_k$, $1\leq k\leq 6$, for a nonzero $c$,
which leaves $h_1$, $h_3$ and $h_4$ invariant, but scales each coefficient in $h_2$ by $c^{-1}$, except the constant term. 
Choose two coefficients $\mu_i\neq \mu_j$, then one can apply this scaling with the choice $c=\mu_i-\mu_j$. Without loss of generality, assume $\mu_4\neq \mu_1$ (if this is not the case, relabel accordingly), then we can normalise $h_2'$ such that the coefficient of $z_4$ equals one, resulting in \eqref{eqp-h2-rho},
where
\begin{equation}\label{eq:rhos}
    \rho_k:=\frac{\mu_k-\mu_1}{\mu_4-\mu_1},\qquad k=2,3,5,6.
\end{equation}
\end{proof}

Lemma \ref{lm:parindep} shows that the collection of algebraic surfaces \eqref{qp6:seg}, up to affine equivalence, constitutes at most a six-dimensional family.
In Section \ref{subsec:qpvi_characterisation}, we prove that it is in fact exactly six-dimensional. For this reason, and the fact that this family appears as the monodromy manifold of $q\Psix$, we give the following definition: 
\begin{definition}
   We refer to the embedded affine surface $\mathcal{Z}_q$ defined in \eqref{qp6:seg}, or equivalently the one defined as the zero set in $\mathbb C^6$ of $h_1, h_3,h_4$ and $h_2''$,  as the \textit{affine Segre surface of $q\Psix$.}
\end{definition}

Lemma \ref{lm:parindep} also provides a reduced form in which redundant parameters have been eliminated. This is important in order to speak about \textit{generic parameters.}

\begin{definition}
    We say that the parameters  $\mu_1,\dots,\mu_6$ and $\lambda_1,\lambda_2$ are generic if $\mu_4\neq \mu_1$ and $\lambda_1,\lambda_2,\rho_2,\rho_3,\rho_5,\rho_6$, defined in \eqref{eq:rhos}, do not satisfy any non-trivial polynomial relations with rational coefficients.
\end{definition}

This Section is organised as follows: in Section
 \ref{subsec:qpvi_characterisation} we show that generic members of the family $\mathcal Z_q$ are affine Segre surfaces and we study their projective completion. In 
 Section \ref{subsec:smooth_affine}, we show that the family $\mathcal Z_q$ is a six parameter family. In particular, 
 we derive a standard form for embedded affine Segre surfaces, with smooth canonical projective completion, and prove that a generic such surface can be transformed into \eqref{qp6:seg} via an affine transformation. This is followed by Section \ref{subsec:parametrisation}, where we study the parametrisation of the surface in terms of the parameters of $q\Psix$, and we further provide explicit formulas of the $16$ lines on the Segre surface in Section \ref{subsec:lines_qpviSegre}.

\subsection{Algebraic characterisation}\label{subsec:qpvi_characterisation}
Notice that \eqref{qp6:seg} indeed defines an affine Segre surface as, upon eliminating two of the variables using the linear equations $\{h_1=0,h_2=0\}$, one is left with two quadratic equations $\{h_3=0,h_4=0\}$, which define an irreducible variety given by the complete intersection of two quadrics in $\mathbb{C}^4$,  for generic coefficients.


In order give an algebraic characterisation of its generic members, we define
the canonical projective completion $\overline{\mathcal{Z}}_q$ in $\mathbb{P}^6$ of $\mathcal{Z}_q$, replacing $h_2$ with $h_2''$ defined in \eqref{eqp-h2-rho} and using projective coordinates,
\begin{equation*}
  [Z_0:Z_1:Z_2:Z_3:Z_4:Z_5:Z_6]=[1:z_1:z_2:z_3:z_4:z_5:z_6], 
\end{equation*}
by
\begin{subequations}\label{qp6:seg_comp}
\begin{align}
&{Z}_{1} + {Z}_{2} + {Z}_{3} + {Z}_{4} + {Z}_{5} + {Z}_{6}=0,\label{qp6:seg_comp_a}\\
& \rho_{2} {Z}_{2} + \rho_{3} {Z}_{3} +  {Z}_{4} + \rho_{5} {Z}_{5} + 
  \rho_{6} {Z}_{6} - Z_0=0,\label{qp6:seg_comp_b}\\
&{Z}_{3} {Z}_{4} - {Z}_{1} {Z}_{2} \lambda_{1}=0,\label{qp6:seg_comp_c}\\
&{Z}_{5} {Z}_{6} - {Z}_{1} {Z}_{2} \lambda_{2}=0.\label{qp6:seg_comp_d}
\end{align}
\end{subequations}

\begin{proposition}\label{prop:qpvialgebraic}
For generic coefficients, the Segre surface $\overline{\mathcal{Z}}_q$ is smooth and the hyperplane section at infinity, $\overline{\mathcal{Z}}_q\setminus \mathcal{Z}_q$, is an irreducible smooth quartic curve, isomorphic to the intersection of two quadric surfaces in $\mathbb{P}^3$, of genus $1$.
\end{proposition}
\begin{proof}
For $\overline{\mathcal{Z}}_q$ to have a singularity at some point $Z\in \mathbb{P}^6$, it is required that the the Jacobian, of the left-hand sides of equations \eqref{qp6:seg_comp},
\begin{equation*}
    J=\begin{bmatrix}
    0 & 1 & 1 & 1 & 1 & 1 & 1\\
    -1 &0 & \rho_2 & \rho_3 & 1 & \rho_5 & \rho_6 \\
    0 & -\lambda_1 Z_2 & -\lambda_1 Z_1 & Z_4 & Z_3 & 0 & 0\\
    0 & -\lambda_2 Z_2 & -\lambda_2 Z_1 & 0 & 0 & Z_6 & Z_5\\
    \end{bmatrix},
\end{equation*}
has rank less than four at this point. By substracting $Z_3$ times the first row from the third row, it follows that this is equivalent to the following two row vectors being linearly dependent,
\begin{align*}
    &v_1=[-\lambda_1 Z_2-Z_3 & &-\lambda_1 Z_1-Z_3 & &Z_4-Z_3 & -Z_3& & -Z_3]&,\\
    &v_2=[-\lambda_2 Z_2 & &-\lambda_2 Z_1 & &0 & Z_6& & Z_5]&.
\end{align*}
Following the proof of \cite[Proposition 2.6]{rof_qpvi}, we find that equations \eqref{qp6:seg_comp} and
\begin{equation*}
 r_0 v_1-r_1 v_2=0,   
\end{equation*}
have a common solution in $Z\in\mathbb{P}^6$ and $[r_0:r_1]\in\mathbb{P}^1$, only if 
 $\lambda_1=0$, $\lambda_2=0$ or
\begin{equation}\label{eq:smoothnesstest}
   (\lambda_1-\lambda_2)^2-2(\lambda_1+\lambda_2)+1=0. 
\end{equation}
As a consequence, $\overline{\mathcal{Z}}_q$ is smooth for generic coefficients.

The curve at infinity is obtained by setting $Z_0=0$ in equations \eqref{qp6:seg_comp}. Upon eliminating $Z_3$ and $Z_6$ using the first two resulting linear equations, we are left with two quadratic equations, 
\begin{align*}
 Q_1:=\lambda_1 Z_1Z_2-Z_4\frac{-\rho_6 Z_1+(\rho_2-\rho_6)Z_2+(1-\rho_6)Z_4+(\rho_5-\rho_6)Z_5}{(\rho_6-\rho_3)}=0,\\
Q_2:=\lambda_2 Z_1Z_2+Z_5\frac{-\rho_3Z_1+(\rho_2-\rho_3)Z_2+(1-\rho_3)Z_4+(\rho_5-\rho_3)Z_5}{(\rho_6-\rho_3)}=0.
\end{align*}

This shows that the curve at infinity is isomorphic to the intersection of the two quadrics in $\{[Z_1:Z_2:Z_4:Z_5]\in\mathbb{P}^3\}$, defined by  the above two equations. 

If the curve has a singularity at some point, then the gradients of $Q_1$ and $Q_2$, with respect to $Z_1,Z_2,Z_4,Z_5$, must be multiples of each other at this point, that is,
\begin{equation}\label{eq:lin_dependence}
r_0\nabla Q_1-r_1\nabla Q_2=0,
\end{equation}
for some $[r_0:r_1]\in\mathbb{P}^1$. Equation \eqref{eq:lin_dependence} constitutes four multi-linear homogeneous equations among the variables $r_0,r_1,Z_1,Z_2,Z_4,Z_5$. 
A direct calculation shows that these four equations,
in addition to the two quadratic equations $Q_{1,2}=0$, admit a common solution only if certain algebraic relations among the parameters $\lambda_1,\lambda_2,\rho_2,\rho_3,\rho_5,\rho_6$  are satisfied. As a consequence, the curve at infinity is smooth for generic coefficients.

Finally, we recall that a smooth intersection of two quadrics in $\mathbb{P}^3$ with a rational point, is isomorphic to an elliptic curve \cite[\S 8 (iv)]{cassels_elliptic}.
The curve at infinity is therefore irreducible and has genus $1$. The proposition follows.
\end{proof}

\begin{remark}\label{remark:curve_at_infinity}
Recall that a smooth intersection of two quadrics in $\mathbb{P}^3$ has genus $1$, 
  so the hyperplane section at infinity, $\overline{\mathcal{Z}}_q\setminus \mathcal{Z}_q$, in Proposition \ref{prop:qpvialgebraic}, with rational point $[1:0:0:0]$,  is generically an elliptic curve. To put it into cubic form, note that both $Q_1$ and $Q_2$ are linear in $Z_1$, so their resultant with respect to $Z_1$ is a homogeneous cubic in $Z_2,Z_4,Z_5$, defining the following cubic curve in $\mathbb{P}^2$,
\begin{align*}
    Z_2 Z_5\lambda_1(Z_2(\rho_3-\rho_2)+Z_4(\rho_3-1)+ Z_5(\rho_3-\rho_5))-Z_4 Z_5(\rho_2 Z_2+ Z_4+ \rho_5 Z_5)+\\
     Z_2 Z_4\lambda_2(Z_2(\rho_6-\rho_2)+Z_4(\rho_6-1)+ Z_5(\rho_6-\rho_5))=0.
  \end{align*}
  The projection of the curve at infinity onto this cubic, regularised near the point $[1:0:0:0]$ by sending
 \begin{equation*}
    [1:0:0:0]\mapsto [1:\lambda_1(\rho_3/\rho_6-1):\lambda_2(\rho_6/\rho_3-1)],
  \end{equation*}
  is an
  isomorphism \cite[\S 8 (iv)]{cassels_elliptic}.
\end{remark}

\subsection{Smooth Segre surfaces up to affine equivalence}\label{subsec:smooth_affine}
In this section, we derive a standard form for embedded affine Segre surfaces with smooth canonical projective completion. 
Because as observed at the beginning of Subsection \ref{subsec:qpvi_characterisation}, each generic element of $\mathcal Z_q$ defines
an irreducible variety given by the complete intersection of two quadrics in  $\mathbb{C}^4$, in this subsection we work in $\mathbb{P}^4$.

Firstly, let us recall that any isomorphism between Segre surfaces, embedded in $\mathbb{P}^4$, comes from a projective equivalence. We formulate this as the following lemma. 
 \begin{lemma}\label{lem:segre_isomorpism}
Any isomorphism $\phi:\mathcal S\rightarrow S'$
between embedded smooth Segre surfaces $\mathcal S,\mathcal S'\subseteq \mathbb{P}^4$,  
extends to a unique linear projective transformation of $\mathbb{P}^4$.
\end{lemma}
\begin{proof}
Any smooth Segre surface $\mathcal S\subseteq \mathbb{P}^4$ is the image under the anti-canonical mapping of the projective plane blown up at five points in general position \cite{Dolgachev}. Explicitly, choose five skew lines $L_k$, $1\leq k\leq 5$, in $\mathcal S$. Then there exist five points $u_k\in\mathbb{P}^2$, $1\leq k\leq 5$, in general position, and a bi-rational map
\begin{equation*}
     \pi: \mathcal S\rightarrow \mathbb{P}^2,
 \end{equation*}
 which is the simultaneous blow down of $L_k$ to $u_k$, $1\leq k\leq 5$.
 
 For an explicit description of $\pi^{-1}$, consider the vector space
\begin{equation*}
    \mathcal C=\{\text{cubic forms in $\{U_0,U_1,U_2\}$ that vanish at $u_k$, $1\leq k\leq 5$}\}, 
 \end{equation*}
 where we use projective coordinates $[U_0:U_1:U_2]$ on $\mathbb{P}^2$. This vector space has dimension $5$ and there exists a basis $\{C_0,\ldots,C_4\}$ of $\mathcal C$ such that $\pi^{-1}$ is given by
 \begin{equation*}
 \begin{array}{lccc}
     \pi^{-1}:&\mathbb{P}^2&\dashrightarrow &
     \mathcal S\\
     &[U_0:U_1:U_2]&\mapsto& [C_0:C_1:C_2:C_3:C_4].
      \end{array}
 \end{equation*}
Let $\mathcal S'\subseteq \mathbb{P}^4$ be another smooth Segre surface and denote correspondingly the objects introduced above by $\pi'$, $u_k'$,$1\leq k\leq 5$ and $\mathcal C'$ for this Segre surface. Then, an isomorphism $\phi:\mathcal S\rightarrow S'$, induces a corresponding automorphism $\widetilde{\phi}$ of $\mathbb{P}^2$, making the diagram
\begin{center}\begin{tikzcd}
\mathcal S \arrow{r}{\phi} \arrow{d}[swap]{\pi} &S' \arrow{d}{\pi'} \\   
\mathbb{P}^2\arrow[swap]{r}{\widetilde{\phi}} & \mathbb{P}^2
 \end{tikzcd}
\end{center}
commutative. So $\widetilde{\phi}$ is a projective linear transformation. Upon picking an $M\in GL_3(\mathbb{C})$ such that $\widetilde{\phi}(U)=U\cdot M$, $U\in\mathbb{P}^2$, we obtain a corresponding isomorphism
\begin{equation*}
\mathcal{L}:\mathcal C'\rightarrow \mathcal C,C'\mapsto C,\qquad C(U)=C'(U\cdot M).
\end{equation*}
Denote by $N\in GL_5(\mathbb{C})$ the unique change of basis matrix
\begin{equation}\label{eq:change_basis}
 (\mathcal{L}C_0', \mathcal{L}C_1',\mathcal{L}C_2', \mathcal{L}C_3',\mathcal{L}C_4')=(C_0,C_1,C_2,C_3,C_4)\cdot N,
\end{equation}
then $\phi(S)=S\cdot N$ for $S\in \mathcal S$ and the lemma follows.
\end{proof}
\begin{remark}\label{remark:lines_segre}
    From the blow-up model of smooth Segre surfaces in the proof of Lemma \ref{lem:segre_isomorpism}, it follows that a Segre surface has exactly $16$ lines: $5$ are the exceptional lines above the base points, $10$ are the total transforms of lines in $\mathbb{P}^2$ going through two of the base points and finally one further line is the total transform of the unique conic going through all five base points. In Section \ref{subsec:lines_qpviSegre}, we give explicit expressions for the $16$ lines on the affine Segre surface of $q\Psix$, as well as their intersection graph.
\end{remark}

It follows from the lemma above that, to classify smooth Segre surfaces up to isomorphism, it suffices to classify them up to projective equivalence. In turn, see \cite{mabuchi} or \cite[Theorem 8.5.1]{Dolgachev}, any smooth Segre surface $\mathcal S$ can, by a linear projective transformation, be put into diagonal form
\begin{equation*}
    S_0^2+S_1^2+S_2^2+S_3^3+S_4^2=0,\quad \gamma_0 S_0^2+\gamma_1 S_1^2+\gamma_2 S_2^2+\gamma_3 S_3^3+\gamma_4S_4^2=0,
\end{equation*}
for some mutually distinct $\gamma_k\in\mathbb{C}$, $0\leq k\leq 4$.
Since we are interested in embedded affine Segre surfaces (with smooth canonical projective completion), we derive a corresponding standard form for them, up to affine transformations. 

Given any embedded affine Segre surface $\mathcal S\subseteq \mathbb{C}^4$, we may choose symmetric $5\times 5$ matrices $M,N$ such that $\mathcal S$ is described by
\begin{equation}\label{eq:segre_form}
    s^T M s=0,\qquad s^T N s=0,\qquad s=[1,s_1,s_2,s_3,s_4]^T.
\end{equation} 
The canonical projective completion 
$\overline{\mathcal S}\subseteq \mathbb{P}^4$ of $\mathcal S$ is obtained through projective coordinates $$[S_0:S_1:S_2:S_3:S_4]=[1:s_1:s_2:s_3:s_4],$$ by replacing $s\mapsto S$ in equations \eqref{eq:segre_form}. The Segre surface $\overline{\mathcal S}$ is smooth
 if and only if the equation
\begin{equation}\label{eq:vanishdet}
   |r_0 M-r_1 N|=0, 
\end{equation}
has five distinct solutions $r=[r_0:r_1]\in\mathbb{P}^1$, see \cite[\S 8.5]{Dolgachev}.

We have the freedom of applying automorphisms of $\mathbb{C}^4$ to equations \eqref{eq:segre_form}, i.e. affine transformations
\begin{equation}\label{eq:affine_transformations}
s\mapsto s'=G^{-1} s,\quad M\mapsto M' =G^T M G,\quad N\mapsto N' =G^T N G,
\end{equation}
where $G\in GL_5(\mathbb{C})$ takes the form
\begin{equation*}
    G=\begin{bmatrix}
        1 & \underline{0}^T\\
        b & G_4
    \end{bmatrix},
\end{equation*}
for some $G_4\in GL_4(\mathbb{C})$ and vector $b\in\mathbb{C}^4$, where $\underline{0}^T$ denotes the transposed zero vector. Using such transformations, we arrive at a standard form for embedded affine  Segre surfaces with smooth canonical projective completion, as detailed in the following lemma.

\begin{lemma}\label{lemma:standard_form}
Given an embedded affine Segre surface $\mathcal S\subseteq \mathbb{C}^4$, whose canonical projective completion $\overline{\mathcal S}\subseteq \mathbb{P}^4$ is smooth, there exists an affine transformation that puts it in the standard form
\begin{subequations}\label{eq:segre_standard}
\begin{align}
&s_1^2+s_2^2+s_3^2+s_4^2=0,\label{eq:segre_standard_1}\\
    &\alpha_1 s_1^2+\alpha_2 s_2^2+\alpha_3 s_3^2+\alpha_4 s_4^2+\beta_1s_1+\beta_2 s_2+\beta_3 s_3+\beta_4 s_4+1=0,\label{eq:segre_standard_2}
\end{align}
\end{subequations}
for some $\alpha_k,\beta_k\in\mathbb{C}$, $1\leq k\leq 4$. 
\end{lemma}
\begin{remark}\label{remark:scaling_freedom_standard}
In the standard form \eqref{eq:segre_standard}, there is still the freedom of adding arbitrary multiples of the first equation to the second, so that we may for example normalise the $\alpha$'s such that $\alpha_1+\alpha_2+\alpha_3+\alpha_4=0$. Similarly, there is the freedom of scaling $s_k\mapsto g^{-1} s_k$, $1\leq k\leq 4$, so that
\begin{equation*}
    \alpha_k\mapsto g^2\alpha_k,\quad \beta_k\mapsto g\beta_k,\quad (1\leq k\leq 4),
\end{equation*}
 In the generic setting, when the $\alpha$'s are distinct, the only other freedom comes from permuting the variables $\{s_1,s_2,s_3,s_4\}$. In particular, the family of embedded affine Segre surfaces, with smooth canonical projective completion, up to affine equivalence, is six-dimensional.
\end{remark}
\begin{remark}\label{remark:smoothness_standard_form} We note that the canonical projective completion of the Segre surface, defined by equations \eqref{eq:segre_standard}, is only smooth for generic coefficients: it is smooth if and only if the discriminant $\operatorname{Disc}_r(|N-rM|)$ is nonzero, where $M$ and $N$ are the symmetric matrices corresponding to equations \eqref{eq:segre_standard}, as in equation \eqref{eq:segre_form}.
\end{remark}

\begin{proof}[Proof of Lemma \ref{lemma:standard_form}]
Choose symmetric $5\times 5$ matrices $M,N$ such that $S$ is described by equations \eqref{eq:segre_form}.
Our goal is to bring these equations into the standard form \eqref{eq:segre_standard}, using affine transformations \eqref{eq:affine_transformations}, as well as by taking linear combinations of the two equations,
\begin{equation*}
    M\mapsto a_{11} M+a_{12} N,\quad N\mapsto a_{21} M+a_{22} N,\quad (a_{ij})_{1\leq i,j\leq 2}\in GL_2(\mathbb{C}).
\end{equation*}
We are going to give an algorithmic method to accomplish this. We will be indexing $5\times 5$ matrices from $0$ to $4$, e.g. $P=(p_{ij})_{0\leq i,j\leq 4}$, and we further introduce the notation $P_4=(p_{ij})_{1\leq i,j\leq 4}$.

Let us for the moment assume that
\begin{equation}\label{eq:nondegenerate}
    |M_4-r N_4|\not\equiv 0.
\end{equation}
This  assumption is in fact moot, as it is always satisfied when $\overline{S}$ is smooth, but we will provide an argument for this a bit later.

Due to \eqref{eq:nondegenerate}, we know that
\begin{equation}\label{eq:vanishdet4}
   |r_0 M_4-r_1 N_4|=0, 
\end{equation}
has at most four solutions $r=[r_0:r_1]\in\mathbb{P}^1$. On the other hand, since $\overline{S}$ is smooth, equation \eqref{eq:vanishdet} 
has five distinct solutions $r=[r_0:r_1]\in\mathbb{P}^1$.
 Pick a solution $r\in\mathbb{P}^1$ to \eqref{eq:vanishdet}, which does not satisfy \eqref{eq:vanishdet4}. If $r_0\neq 0$, replace $M$ by $r_0 M_4-r_1 N_4$, otherwise simply swap $M$ and $N$.
Then $M$ is not invertible, but $M_4$ is. Since $M_4$ is symmetric and invertible, we can construct an orthogonal matrix $P$ such that $P^TM_4 P=D$, where $D$ is a diagonal matrix with nonzero entries. By applying the affine transformation \eqref{eq:affine_transformations}, with
\begin{equation*}
    G=\begin{bmatrix}
        1 & \underline{0}^T\\
        \underline{0} &  P D^{-\tfrac{1}{2}}
    \end{bmatrix},
\end{equation*}
 we normalise $M$ such that $M_4=I_4$, that is,
\begin{equation*}
    M=\begin{bmatrix}
        m_0 & m_1 & m_2 & m_3 & m_4\\
        m_1 & 1 & 0 & 0 & 0\\
        m_2 & 0 & 1 & 0 & 0\\
        m_3 & 0 & 0 & 1 & 0\\
        m_4 & 0 & 0 & 0 & 1
    \end{bmatrix},
\end{equation*}
with $m_0=m_1^2+m_2^2+m_3^2+m_4^2$, since $|M|=0$. Next, we apply the affine transformation \eqref{eq:affine_transformations}, with
\begin{equation*}
G=\begin{bmatrix}
        1 & 0 & 0 & 0 & 0\\
        -m_1 & 1 & 0 & 0 & 0\\
        -m_2 & 0 & 1 & 0 & 0\\
        -m_3 & 0 & 0 & 1 & 0\\
        -m_4 & 0 & 0 & 0 & 1
    \end{bmatrix},
\end{equation*}
which reduces the matrix $M$ to
\begin{equation}\label{eq:Mstandard}
    M=\begin{bmatrix}
        0 & \underline{0}^T\\
        \underline{0} &  I_4
    \end{bmatrix}.
\end{equation}
So, the first equation in \eqref{eq:segre_form} is now given by \eqref{eq:segre_standard_1}.

Next, we diagonalise $N_4$. Let $P$ be an orthonormal matrix such that $P^TN_4 P=D$, where $D$ is a diagonal matrix. Applying the affine transformation \eqref{eq:affine_transformations}, with
\begin{equation}\label{eq:algorithm_step_diagonalN}
    G=\begin{bmatrix}
        1 & \underline{0}^T\\
        \underline{0} &  P\hspace{0.1cm}
    \end{bmatrix},
\end{equation}
the matrix $M$ is left invariant, whereas $N$ takes the form
\begin{equation}\label{eq:Nstandard}
    N=\begin{bmatrix}
        \beta_0 & \tfrac{1}{2}\beta_1 & \tfrac{1}{2}\beta_2 & \tfrac{1}{2}\beta_3 & \tfrac{1}{2}\beta_4\\
        \tfrac{1}{2}\beta_1 & \alpha_1 & 0 & 0 & 0\\
        \tfrac{1}{2}\beta_2 & 0 & \alpha_2 & 0 & 0\\
        \tfrac{1}{2}\beta_3 & 0 & 0 & \alpha_3 & 0\\
        \tfrac{1}{2}\beta_4 & 0 & 0 & 0 & \alpha_4
    \end{bmatrix},
\end{equation}
for some $\alpha_k,\beta_k\in\mathbb{C}$, $1\leq k\leq 4$ and $\beta_0\in\mathbb{C}$. In fact, $\beta_0$ must be nonzero since $|M-rN|$ cannot have a double root at $r=0$. By dividing $N$ by $\beta_0$, equations \eqref{eq:segre_form} have been brought into the standard form \eqref{eq:segre_standard}.



To finish the proof, we come back to the assumption \eqref{eq:nondegenerate}, that we made in the beginning. Suppose that it does not hold true. 
Since $|M-rN|\not\equiv 0$, we may replace $M$ by $M-rN$, for a generic $r$, so that $M$ is invertible. In particular, for such a choice of $M$ the rank of $M_4$ must be at least $3$. On the other hand, as \eqref{eq:nondegenerate} does not hold true, $|M_4|=0$ and so $M_4$ must have rank exactly $3$. Analogous to how we brought $M$ and $N$ into the forms \eqref{eq:Mstandard} and \eqref{eq:Nstandard} above, and by additionally permuting the variables $\{s_1,s_2,s_3,s_4\}$ if necessary, we bring $M$ and $N$ into the form
\begin{equation*}
    M=\begin{bmatrix}
        1 & 0 & 0 & 0 & 0\\
        0 & 0 & 0 & 0 & 0\\
        0 & 0 & 1 & 0 & 0\\
        0 & 0 & 0 & 1 & 0\\
        0 & 0 & 0 & 0 & 1\\
    \end{bmatrix},\quad
      N=\begin{bmatrix}
        \beta_0 & \tfrac{1}{2}\beta_1 & \tfrac{1}{2}\beta_2 & \tfrac{1}{2}\beta_3 & \tfrac{1}{2}\beta_4\\
        \tfrac{1}{2}\beta_1 & \alpha_1 & \gamma_2 & \gamma_3 & \gamma_4\\
        \tfrac{1}{2}\beta_2 & \gamma_2 & \alpha_2 & 0 & 0\\
        \tfrac{1}{2}\beta_3 & \gamma_3 & 0 & \alpha_3 & 0\\
        \tfrac{1}{2}\beta_4 & \gamma_4 & 0 & 0 & \alpha_4
    \end{bmatrix}.
\end{equation*}
From the coefficient of $r^1$ in $|M_4-rN_4|\equiv 0$, we immediately read off that $\alpha_1=0$. By similarly looking at the coefficients of $r^2, r^3$ and $r^4$, we obtain the linear system
\begin{equation}\label{eq:technical_proof_linearsystem}
    \begin{bmatrix}
    1 & 1 & 1\\
    \alpha_3+\alpha_4 & \alpha_2+\alpha_4 & \alpha_2+\alpha_3\\
        \alpha_3\alpha_4 & \alpha_2 \alpha_4 & \alpha_2\alpha_3
    \end{bmatrix}\cdot
    \begin{bmatrix}
        \gamma_2^2\\
        \gamma_3^2\\
        \gamma_4^2
    \end{bmatrix}=0.
\end{equation}
If $\gamma_2=\gamma_3=\gamma_4=0$, then $|M-rN|$ has a double root at $r=0$, so $\overline{S}$ is not smooth and we have arrived at a contradiction. Else $\alpha_2,\alpha_3$ and $\alpha_4$ are not all distinct, and by solving \eqref{eq:technical_proof_linearsystem} for $\{\gamma_2,\gamma_3,\gamma_4\}$, one similarly sees that $|M-rN|$ always has a root with multiplicity, again yielding a contradiction. We conclude that assumption \eqref{eq:nondegenerate} always holds and the lemma follows.
\end{proof}

In the following proposition, it is shown that a generic embedded affine Segre surface, with smooth projective completion, can be put in the form \eqref{qp6:seg}, by an affine transformation.
\begin{proposition}\label{prop:segre_equivalence}
    For generic $\alpha_k,\beta_k\in\mathbb{C}$, $1\leq k\leq 4$, the Segre surface \eqref{eq:segre_standard} is affinely equivalent to the Segre surface \eqref{qp6:seg}, for some values of the coefficients in \eqref{qp6:seg}.
\end{proposition}
\begin{proof}
The proof of Lemma \ref{lemma:standard_form} gives an algorithmic procedure to put an embedded affine Segre surface, with smooth projective completion, into the form \eqref{eq:segre_standard}. We are going to apply this procedure to the family of Segre surfaces $\mathcal Z_q$ \eqref{qp6:seg}. This will provide a set of algebraic equations among the parameters of both families of Segre surfaces. By application of the implicit function theorem and standard algebraic arguments, we obtain that this set of algebraic equations is solvable for generic parameter values
$\alpha_k,\beta_k\in\mathbb{C}$, $1\leq k\leq 4$, from which the proposition will follow.


As a first step, we eliminate two of the variables in \eqref{qp6:seg}.
Assuming $\rho_2\neq 0$, we use the linear equations $\{h_1=0,h_2''=0\}$ to eliminate $z_1$ and $z_2$ from the  equations $\{h_3=0,h_4=0\}$, and, by renaming the remaining variables as $z_{k+2}=s_k$, $1\leq k\leq 4$, we arrive at
\begin{align*}
    s_1s_2-\frac{\lambda_1}{\lambda_2}&  s_3s_4=0,\\
    s_1s_2-\hspace{0.5mm}\frac{\lambda_1}{\rho_2^2}\hspace{0.5mm}&  (1+(\rho_2-\rho_3)s_1+(\rho_2-1)s_2+(\rho_2-\rho_5)s_3+(\rho_2-\rho_6)s_4)\\
    &(\rho_3 s_1+ s_2+ \rho_5 s_3+ \rho_6 s_4-1)=0,
\end{align*}
where we used the $\rho_k$, $1\leq k\leq 4$, defined in \eqref{eq:rhos}.

Now, construct the corresponding symmetric $5\, 5$ matrices $M,N$, as in \eqref{eq:segre_form}. Applying the affine transformation \eqref{eq:affine_transformations}, with
\begin{equation*}
    G=\frac{1}{\sqrt{2}}\begin{bmatrix}
        1 & 0 & 0 & 0 & 0\\
        0 & -i & i & 0 & 0\\
        0 & 1 & 1 & 0 & 0\\
        0 & 0 & 0 & -\sqrt{\frac{\lambda_2}{\lambda_1}} & \sqrt{\frac{\lambda_2}{\lambda_1}}\\
        0 & 0 & 0 & i\sqrt{\frac{\lambda_2}{\lambda_1}} & i\sqrt{\frac{\lambda_2}{\lambda_1}}\\
    \end{bmatrix},
\end{equation*}
readily puts $M$ into the form \eqref{eq:Mstandard}. 

What is left, is to construct an orthonormal matrix $P$ that diagonalizes $N_4$, i.e. $P^TN_4 P=D$, where $D$ is a diagonal matrix. Upon introducing such a $P$ formally, applying \eqref{eq:algorithm_step_diagonalN}, comparing the result with \eqref{eq:Nstandard} (with $\beta_0=1$) and finally eliminating the entries of $P$ by taking resultants, one is left with $7$ polynomial equations, among the variables $\lambda_1,\lambda_2,\rho_2,\rho_3,\rho_5,\rho_6$ and $\alpha_k,\beta_k$, $1\leq k\leq 4$. By remark \eqref{remark:scaling_freedom_standard}, we may normalise $\alpha_4=-\alpha_1-\alpha_2-\alpha_3$, and by further introducing the scaling by nonzero $g$ as specified in the remark, we end up with seven polynomial equations, with rational coefficients, among two sets of seven variables,
\begin{equation*}
F_k=F_k(g,\lambda,\rho;\alpha,\beta)=0\quad (1\leq k\leq 7),  
\end{equation*}
where $\lambda=(\lambda_1,\lambda_2)$, $\rho=(\rho_2,\rho_3,\rho_5,\rho_6)$, $\alpha=(\alpha_k)_{1\leq k\leq 3}$ and $\beta=(\beta_k)_{1\leq k\leq 4}$.

Explicit formulas for these polynomials are lengthy and best obtained using mathematical software, e.g. Mathematica. One of the simplest solutions that we could find to this polynomial system, is given by

\begin{align*}
\lambda_1&=-3, &
\lambda_2&=-3, &
\rho_2&=2,\\
\rho_3&=\tfrac{2}{3}, &
\rho_5&=-1, &
\rho_6&=\tfrac{2}{3},\\
\alpha_1&=\tfrac{1}{6}(2+\sqrt{6})i, &
\alpha_2&=\tfrac{1}{6}(-2+\sqrt{6})i, &
\alpha_3&=\tfrac{1}{6}(2-\sqrt{6})i,
\end{align*}
together with $g=1$ and
\begin{align*}
\beta_1&=\frac{1}{6}\sqrt{(-24+39 i)-(16-15i)\sqrt{6}}, &
\beta_2&=\frac{1}{6}\sqrt{(-24-39 i)+(16+15i)\sqrt{6}},\\
\beta_3&=\frac{1}{6}\sqrt{(-24+39 i)+(16-15i)\sqrt{6}}, &
\beta_4&=\frac{1}{6}\sqrt{(-24-39 i)-(16+15i)\sqrt{6}}.
\end{align*}

In particular, the system $F=(F_k)_{1\leq k\leq 7}$ is consistent and defines a non-empty algebraic set $V\subseteq \mathbb{C}^{14}$.
Furthermore, the Jacobian determinant $|J_{\lambda,\rho,g}(F)|$ of $F=(F_k)_{1\leq k\leq 7}$ with respect to the seven variables $\{g,\lambda,\rho\}$, evaluates to a nonzero rational number at this point. Though we will not need it, we remark that the Jacobian of $F$ with respect to the other set of seven variables, $\{\alpha,\beta\}$, is also nonzero at this point.

By the implicit function theorem, the set
\begin{equation*}
V=\{(\alpha,\beta)\in\mathbb{C}^7:\exists_{(g,\lambda,\rho)\in\mathbb{C}^7}:F=0,|J_{\lambda,\rho,g}(F)|\neq 0\text{ and }g,\lambda_1,\lambda_2\neq 0\},    
\end{equation*}
is open and non-empty. But $V$ is a constructible set with respect to the Zariski topology,
so either $V$ or its complement is dense in $\mathbb{C}^7$ with respect to the Euclidean topology. It follows that $V$ is dense in $\mathbb{C}^7$.
 That is, for generic $(\alpha,\beta)\in\mathbb{C}^7$, we can find $\lambda_1,\lambda_2$,  $\rho_2,\rho_3,\rho_5,\rho_6$, and an affine transformation (which generally depends on $g$), that transforms the Segre surface, defined by equations \eqref{eq:segre_standard}, into \eqref{qp6:seg}. The proposition follows. 
\end{proof}

From Lemma \ref{lemma:standard_form}, Remark \ref{remark:scaling_freedom_standard} and Proposition \ref{prop:segre_equivalence}, we obtain the following result.
\begin{corollary}\label{cor:6dim}
The embedded family of affine Segre surfaces $\mathcal Z_q$ defined in \eqref{qp6:seg},  constitute a six-dimensional family, up to affine equivalence.
\end{corollary}

\subsection{Parametrisation} \label{subsec:parametrisation}
In this section, we introduce and study the parametrisation of the affine Segre surface \eqref{qp6:seg} in terms of the parameter of $q\Psix$, which is defined below. This parametrisation will be helpful for two separate purposes. The first is to give explicit expressions for the lines on the Segre surface. The second is to compute a continuum limit of the Segre surface.

Given $q\in\mathbb{C}$, $0<|q|<1$, and $\kappa=(\kappa_0,\kappa_t,\kappa_1,\kappa_\infty)\in (\mathbb C^*)^4$, the $q$-difference sixth Painlev\'e equation \cite{jimbosakai} is 
\begin{align}\label{eq:qpvi}
q\Psix:\ \begin{cases}\  f\overline{f}&=\dfrac{(g-\kappa_0\,t)(g-\kappa_0^{-1}t)}{(g-\kappa_\infty)(g-q^{-1}\kappa_\infty^{-1})},\\
  \   g\overline{g}&=\dfrac{(f-\kappa_t\,t)(f-\kappa_t^{-1}t)}{q(f-\kappa_1)(f-\kappa_1^{-1})},
 \\
  f,g &:T\rightarrow \mathbb{P}^1
     \end{cases} 
\end{align}
where the domain is given by a $q$-spiral, $T=q^{\mathbb{Z}}t_0$,
 and  $f=f(t)$, $g=g(t)$, $\overline{f}=f(q\,t)$, $\overline{g}=g(q\,t)$, for $t\in T$. Here $q^\alpha:=e^{\alpha\log q}$ and the parameters of the equation are given by  $\kappa_j=q^{\vartheta_j}$ for $j=0,t,1,\infty$, with
\begin{equation}\label{eq:qpvi_parameters}
\vartheta_0,\vartheta_t,\vartheta_1,\vartheta_\infty\in\mathbb{C},\quad t_0\in\mathbb{C}^*,\quad \log q\in \{x\in \mathbb C:\Re x<0\}.
\end{equation}

In order to define the parametrisation of the $q\Psix$ Segre surface, we require some elementary $q$-special functions. Firstly, the $q$-Pochhammer symbol is defined by the product
\begin{equation*}
(z;q)_\infty=\prod_{k=0}^{\infty}{(1-q^kz)},
\end{equation*}
which is locally uniformly convergent in $(z,q)\in\mathbb{C}\times \mathbb{D}$, where $\mathbb{D}$ denotes the open unit disc $\mathbb{D}=\{x\in\mathbb{C}: |x|<1\}$.
Secondly, the $\mathit{q}$-theta function, defined by
\begin{equation*}
\theta_q(z)=(q;q)_\infty(z;q)_\infty(q/z;q)_\infty,
\end{equation*}
is analytic in $(z,q)\in\mathbb{C}^*\times \mathbb{D}$, and admits the following convergent expansion on its domain,
\begin{equation}\label{eq:jacobitripleproduct}
 \theta_q(z)=\sum_{n=-\infty}^\infty(-1)^n q^{\frac{1}{2}n(n-1)}z^n.
\end{equation}
This is known as the Jacobi triple product formula, which shows that
$\theta_q(z)$ has essential singularities at $z=0$ and $z=\infty$, unless $q=0$. 
For $n\in\mathbb{N}^*$, we use the common abbreviation for repeated products of this function
\begin{equation*}
\theta_q(z_1,\ldots,z_n)=\theta_q(z_1)\cdot \ldots\cdot \theta_q(z_n).
\end{equation*}

Now, the coefficients in \eqref{qp6:segh} are explicitly parametrised by the parameters in \eqref{eq:qpvi_parameters} as follows \cite{jr_qp6}, 
\begin{equation}\label{eq:qpvimu}
\begin{aligned}
    &\mu_{1}=\prod_{\epsilon=\pm 1}\frac{\theta_q\left(q^{+\vartheta_\infty}t_0\right)}{\theta_q\left(q^{\epsilon \hspace{0.1mm}\vartheta_0+\vartheta_\infty}t_0\right)}, &
    &\mu_{2}=\prod_{\epsilon=\pm 1}    \frac{\theta_q\left(q^{-\vartheta_\infty}t_0\right)}{\theta_q\left(q^{\epsilon \hspace{0.1mm}\vartheta_0-\vartheta_\infty}t_0\right)},\\
    &\mu_{3}=\prod_{\epsilon=\pm 1} \frac{\theta_q(q^{\vartheta_t+\vartheta_1+\vartheta_\infty})}{\theta_q(q^{\epsilon \,\vartheta_0+\vartheta_t+\vartheta_1+\vartheta_\infty})}, &
    &\mu_{4}=\prod_{\epsilon=\pm 1}\frac{\theta_q(q^{-\vartheta_t-\vartheta_1+\vartheta_\infty})}{\theta_q(q^{\epsilon \,\vartheta_0-\vartheta_t-\vartheta_1+\vartheta_\infty})},\\
    &\mu_{5}=\prod_{\epsilon=\pm 1}\frac{\theta_q(q^{\vartheta_t-\vartheta_1+\vartheta_\infty})}{\theta_q(q^{\epsilon \,\vartheta_0+\vartheta_t-\vartheta_1+\vartheta_\infty})}, &
    &\mu_{6}=\prod_{\epsilon=\pm 1}\frac{\theta_q(q^{-\vartheta_t+\vartheta_1+\vartheta_\infty})}{\theta_q(q^{\epsilon \,\vartheta_0-\vartheta_t+\vartheta_1+\vartheta_\infty})},
\end{aligned}
\end{equation}
and
\begin{equation}\label{eq:lambda12}
    \lambda_1=\frac{\eta_{3}^{(q)}\eta_{4}^{(q)}}{\eta_{1}^{(q)}\eta_{2}^{(q)}},\qquad \lambda_2=\frac{\eta_{5}^{(q)}\eta_{6}^{(q)}}{\eta_{1}^{(q)}\eta_{2}^{(q)}},
\end{equation}
where the $\eta_k^{(q)}$ are given by
\begin{align*}
\eta_{1}^{(q)}&=+\theta_q\left(q^{\vartheta_0+\vartheta_\infty}t_0,q^{-\vartheta_0+\vartheta_\infty}t_0\right)\theta_q(q^{2\vartheta_t},q^{2\vartheta_1}),\\
\eta_{2}^{(q)}&=+\theta_q\left(q^{\vartheta_0-\vartheta_\infty}t_0,q^{-\vartheta_0-\vartheta_\infty}t_0\right)\theta_q(q^{2\vartheta_t},q^{2\vartheta_1})q^{2\vartheta_\infty},\\
\eta_{3}^{(q)}&=-\theta_q\left(q^{\vartheta_t-\vartheta_1}t_0,q^{-\vartheta_t+\vartheta_1}t_0\right)\theta_q(q^{\vartheta_0+\vartheta_t+\vartheta_1+\vartheta_\infty},q^{-\vartheta_0+\vartheta_t+\vartheta_1+\vartheta_\infty}),\\
\eta_{4}^{(q)}&=-\theta_q\left(q^{\vartheta_t-\vartheta_1}t_0,q^{-\vartheta_t+\vartheta_1}t_0\right)\theta_q(q^{\vartheta_0-\vartheta_t-\vartheta_1+\vartheta_\infty},q^{-\vartheta_0-\vartheta_t-\vartheta_1+\vartheta_\infty})q^{2\vartheta_t+2\vartheta_1},\\
\eta_{5}^{(q)}&=+\theta_q\left(q^{\vartheta_t+\vartheta_1}t_0,q^{-\vartheta_t-\vartheta_1}t_0\right)\theta_q(q^{\vartheta_0+\vartheta_t-\vartheta_1+\vartheta_\infty},q^{-\vartheta_0+\vartheta_t-\vartheta_1+\vartheta_\infty})q^{2\vartheta_1},\\
\eta_{6}^{(q)}&=+\theta_q\left(q^{\vartheta_t+\vartheta_1}t_0,q^{-\vartheta_t-\vartheta_1}t_0\right)\theta_q(q^{\vartheta_0-\vartheta_t+\vartheta_1+\vartheta_\infty},q^{-\vartheta_0-\vartheta_t+\vartheta_1+\vartheta_\infty})q^{2\vartheta_t}.
\end{align*}
We further note that
\begin{equation}\label{eq:Tk_identities}
    \mu_k=\frac{\widehat{\eta}_k^{(q)}}{\eta_k^{(q)}},\quad\widehat{\eta}_k^{(q)}:=\eta_k^{(q)}|_{\vartheta_0=0}\qquad (1\leq k \leq 6).
\end{equation}

Under the parametrisation above, the six invariant quantities, $\lambda_1,\lambda_2,\rho_2,\rho_3,\rho_5,\rho_6$, see equation \eqref{eq:rhos}, depend on six free parameters \eqref{eq:qpvi_parameters}. A priori, it is conceivable that there might exist an (algebraic) relation among the coefficients in \eqref{qp6:seg} under this parametrisation, but it follows from the following lemma that this is not the case.

\begin{lemma}\label{lem:genericcoef}
    The meromorphic mapping,
    \begin{equation}\label{eq:map_param_to_coef}
(\vartheta_0,\vartheta_t,\vartheta_1,\vartheta_\infty,t_0,\log q)\mapsto (\lambda_1,\lambda_2,\rho_2,\rho_3,\rho_5,\rho_6),
    \end{equation}
is locally biholomorphic near almost any point in its domain \eqref{eq:qpvi_parameters}.  
\end{lemma}
\begin{proof}
It is helpful to write \eqref{eq:map_param_to_coef} as the composition of the mappings
\begin{align}\label{eq:composition1}
(\vartheta_0,\vartheta_t,\vartheta_1,\vartheta_\infty,t_0,\log q)&\mapsto 
(\kappa_0,\kappa_t,\kappa_1,\kappa_\infty,t_0,q),\\
(\kappa_0,\kappa_t,\kappa_1,\kappa_\infty,t_0,q)&\mapsto 
(\lambda_1,\lambda_2,\rho_2,\rho_3,\rho_5,\rho_6), \label{eq:composition2}
\end{align}
where 
\begin{equation}\label{eq:theta-kappa}
    (\kappa_0,\kappa_t,\kappa_1,\kappa_\infty)=(q^{\vartheta_0},q^{\vartheta_t},q^{\vartheta_1},q^{-\vartheta_\infty}),
\end{equation}
and the domain of \eqref{eq:composition2} is defined by
\begin{equation}\label{eq:composition_2_domain}
\{(\kappa_0,\kappa_t,\kappa_1,\kappa_\infty,t_0,q)\in (\mathbb{C}^*)^5\times \mathbb{D}\}.
\end{equation}

Clearly, the mapping \eqref{eq:composition1} is a local biholomorphism. To prove the lemma, it thus suffices to show that the mapping \eqref{eq:composition2} is locally biholomorphic near almost any point in its domain. Since the domain \eqref{eq:composition_2_domain} is connected, and the mapping \eqref{eq:composition2} is meromorphic, it thus suffices to show that the  Jacobian determinant of \eqref{eq:composition2} is not identically zero. To prove the latter, we are going to compute an expansion of the Jacobian determinant  around $q=0$. 

Recall that $\theta_q(z)$ is an analytic function with respect to $(z,q)\in\mathbb{C}^*\times\mathbb{D}$. From the Jacobi triple product formula \eqref{eq:jacobitripleproduct}, we obtain the following convergent expansion around $q=0$,
\begin{equation*}
    \theta_q(z)=\sum_{n=0}^\infty (-1)^n (z^{-n}-z^{n+1})q^{\frac{1}{2}n(n+1)},
\end{equation*}
which allows us to compute expansions around $q=0$ for $\lambda_1,\lambda_2$ and $\rho_k$, $k=2,3,5,6$. For example,
\begin{align*}
\lambda_1&=\frac{\theta_q\big(t_0\kappa_t\kappa_1^{-1}\big)^2\theta_q\big(t_0\kappa_t^{-1}\kappa_1\big)^2}
{\theta_q(\kappa_t^2)^2\theta_q(\kappa_1^2)^2}
\prod_{\epsilon_{1,2}\in\{\pm 1\}}{\frac{\theta_q(\kappa_0^{\epsilon_1}\kappa_\infty^{\epsilon_2}\kappa_t\kappa_1)}{\theta_q(\kappa_0^{\epsilon_1}\kappa_\infty^{\epsilon_2} t_0)}}\\
&=\lambda_1^{(0)}+q \lambda_1^{(1)}+\mathcal{O}(q^2),
\end{align*}
as $q\rightarrow 0$, where
\begin{align*}
  \lambda_1^{(0)}=&  \frac{(t_0\kappa_t\kappa_1^{-1}-1)(t_0\kappa_t^{-1}\kappa_1-1)}{(\kappa_t^2-1)^2(\kappa_1^2-1)^2}
\prod_{\epsilon_{1,2}\in\{\pm 1\}}{\frac{(\kappa_0^{\epsilon_1}\kappa_\infty^{\epsilon_2}\kappa_t\kappa_1-1)}{(\kappa_0^{\epsilon_1}\kappa_\infty^{\epsilon_2} t_0-1)}},\\
\lambda_1^{(1)}=& 
t_0^{-1}\lambda_1^{(0)}(t_0\kappa_t\kappa_1-1)(t_0\kappa_t^{-1}\kappa_1^{-1}-1) \\
& \cdot
(\kappa_0\kappa_\infty+\kappa_0^{-1}\kappa_\infty+\kappa_0\kappa_\infty^{-1}+\kappa_0^{-1}\kappa_\infty^{-1}-2 \kappa_t\kappa_1^{-1}-2\kappa_t^{-1}\kappa_1).
\end{align*}
Similarly, we obtain expansions,
\begin{equation*}
    \rho_k=\rho_k^{(0)}+q \rho_k^{(1)}+\mathcal{O}(q^2)\qquad (q\rightarrow 0),
\end{equation*}
for $k=2,3,5,6$, with coefficients that are easily computed explicitly.

Let $J$ denote the Jacobian of the mapping \eqref{eq:composition2},
\begin{equation*}
  J=\left[\begin{array}{cccccc}
       \frac{\partial \lambda_1}{\partial \kappa_0}&
  \frac{\partial \lambda_1}{\partial \kappa_t}&
  \frac{\partial \lambda_1}{\partial \kappa_1}&
  \frac{\partial \lambda_1}{\partial \kappa_\infty}&
  \frac{\partial \lambda_1}{\partial t_0}&
  \frac{\partial \lambda_1}{\partial q}\\
   \frac{\partial \lambda_2}{\partial \kappa_0}&
  \frac{\partial \lambda_2}{\partial \kappa_t}&
  \frac{\partial \lambda_2}{\partial \kappa_1}&
  \frac{\partial \lambda_2}{\partial \kappa_\infty}&
  \frac{\partial \lambda_2}{\partial t_0}&
  \frac{\partial \lambda_2}{\partial q}\\
 \frac{\partial \rho_2}{\partial \kappa_0}&
  \frac{\partial \rho_2}{\partial \kappa_t}&
  \frac{\partial \rho_2}{\partial \kappa_1}&
  \frac{\partial \rho_2}{\partial \kappa_\infty}&
  \frac{\partial \rho_2}{\partial t_0}&
  \frac{\partial \rho_2}{\partial q}\\
\vdots&\vdots&\vdots&\vdots&\vdots&\vdots\\
  \frac{\partial \rho_6}{\partial \kappa_0}&
  \frac{\partial \rho_6}{\partial \kappa_t}&
  \frac{\partial \rho_6}{\partial \kappa_1}&
  \frac{\partial \rho_6}{\partial \kappa_\infty}&
  \frac{\partial \rho_6}{\partial t_0}&
  \frac{\partial \rho_6}{\partial q}\\
  \end{array}
  \right].
\end{equation*}
Then $J$ admits an expansion around $q=0$,
\begin{equation*}
  J=\sum_{k=0}^\infty J_k(\kappa_0,\kappa_t,\kappa_1,\kappa_\infty,t_0) q^k,  
\end{equation*}
which is locally uniformly convergent with respect to the remaining variables away from singularities of $J$. Now, a direct computation yields that the determinant of the constant term  is given by
\begin{align*}
|J_0|=&\left|
\begin{array}{cccccc}
       \frac{\partial \lambda_1^{(0)}}{\partial \kappa_0}&
  \frac{\partial \lambda_1^{(0)}}{\partial \kappa_t}&
  \frac{\partial \lambda_1^{(0)}}{\partial \kappa_1}&
  \frac{\partial \lambda_1^{(0)}}{\partial \kappa_\infty}&
  \frac{\partial \lambda_1^{(0)}}{\partial t_0}&
  \frac{\partial \lambda_1^{(0)}}{\partial q}\\
   \frac{\partial \lambda_2^{(0)}}{\partial \kappa_0}&
  \frac{\partial \lambda_2^{(0)}}{\partial \kappa_t}&
  \frac{\partial \lambda_2^{(0)}}{\partial \kappa_1}&
  \frac{\partial \lambda_2^{(0)}}{\partial \kappa_\infty}&
  \frac{\partial \lambda_2^{(0)}}{\partial t_0}&
  \frac{\partial \lambda_2^{(0)}}{\partial q}\\
 \frac{\partial \rho_2^{(0)}}{\partial \kappa_0}&
  \frac{\partial \rho_2^{(0)}}{\partial \kappa_t}&
  \frac{\partial \rho_2^{(0)}}{\partial \kappa_1}&
  \frac{\partial \rho_2^{(0)}}{\partial \kappa_\infty}&
  \frac{\partial \rho_2^{(0)}}{\partial t_0}&
  \frac{\partial \rho_2^{(0)}}{\partial q}\\
\vdots&\vdots&\vdots&\vdots&\vdots&\vdots\\
  \frac{\partial \rho_6^{(0)}}{\partial \kappa_0}&
  \frac{\partial \rho_6^{(0)}}{\partial \kappa_t}&
  \frac{\partial \rho_6^{(0)}}{\partial \kappa_1}&
  \frac{\partial \rho_6^{(0)}}{\partial \kappa_\infty}&
  \frac{\partial \rho_6^{(0)}}{\partial t_0}&
  \frac{\partial \rho_6^{(0)}}{\partial q}\\
  \end{array}\right|\\
=&\,8\frac{(\kappa_t^{-1}\kappa_1^{-1}t_0-1)^4(\kappa_t\kappa_1^{-1}t_0-1)^4(\kappa_t^{-1}\kappa_1t_0-1)^4}{\kappa_0^2\kappa_t^2\kappa_1^2\kappa_\infty^{-2}(\kappa_t\kappa_1t_0-1)(\kappa_t^{-1}\kappa_1^{-1}\kappa_\infty^{-2}t_0-1)^5}\cdot
\frac{(\kappa_0^2-1)(\kappa_\infty^{-2}-1)^6}{(\kappa_t^2-1)^2(\kappa_1^2-1)^2}
\\
&\cdot\prod_{\epsilon\in \{\pm 1\}}
\frac{(\kappa_0^{\epsilon}\kappa_t\kappa_1\kappa_\infty-1)^4(\kappa_0^{\epsilon}\kappa_\infty^{-1}t_0-1)}
{(\kappa_0^{\epsilon}\kappa_\infty t_0-1)^4(\kappa_0^{\epsilon}\kappa_t^{-1}\kappa_1\kappa_\infty-1)
(\kappa_0^{\epsilon}\kappa_t\kappa_1^{-1}\kappa_\infty-1)
(\kappa_0^{\epsilon}\kappa_t\kappa_1\kappa_\infty^{-1}-1)
}.
\end{align*}


As a consequence, $|J|\not \equiv 0$ and the lemma follows.
\end{proof}

Roughly speaking, Lemma \ref{lemma:standard_form} and Remarks \ref{remark:scaling_freedom_standard} and \ref{remark:smoothness_standard_form} allow us to identify generic embedded
 affine Segre surfaces, up to affine equivalence, with generic tuples $\{(\alpha,\beta)\in\mathbb{C}^8\}$, quotiented by the actions in Remark \ref{remark:scaling_freedom_standard}.
 As a consequence of Proposition \ref{prop:segre_equivalence} and  Lemma \ref{lem:genericcoef}, we find that the parametrisation with respect to the parameters of $q\Psix$ in \eqref{eq:qpvi_parameters}, maps out a subset with non-empty interior in this quotient space. This yields the following remark. \begin{remark}\label{remark:mostgeneric}
The embedded affine Segre surface for $q\Psix$ is of the most generic type. In particular, its projective completion is smooth and the curve at infinity is a smooth irreducible quartic curve of genus $1$. 

As a consequence, we obtain negations of certain assertions or conjectures made in \cite{RSpreprint}. In particular, item 7.2.1 in the list of problems in \cite[\S 7.2]{RSpreprint} contains an open problem and three conjectures. Our result resolves the open problem and negates these conjectures.
\end{remark}

\subsection{Lines}\label{subsec:lines_qpviSegre}
It is a classical fact, essentially due to Segre himself \cite{segre1884}, that a smooth Segre surface contains $16$ lines. Half of the lines in $\overline{\mathcal{Z}}_q$ admit a most simple description. Namely, take any $i\in \{1,2\}$, $j\in\{3,4\}$ and $k\in \{5,6\}$, then
\begin{equation}\label{eq:linesqpviI}
L_{i,j,k}^{(q)}:=\{Z\in\overline{\mathcal{Z}}_q:Z_i=Z_j=Z_k=0\},
\end{equation}
defines a line in $\overline{\mathcal{Z}}_q$.

	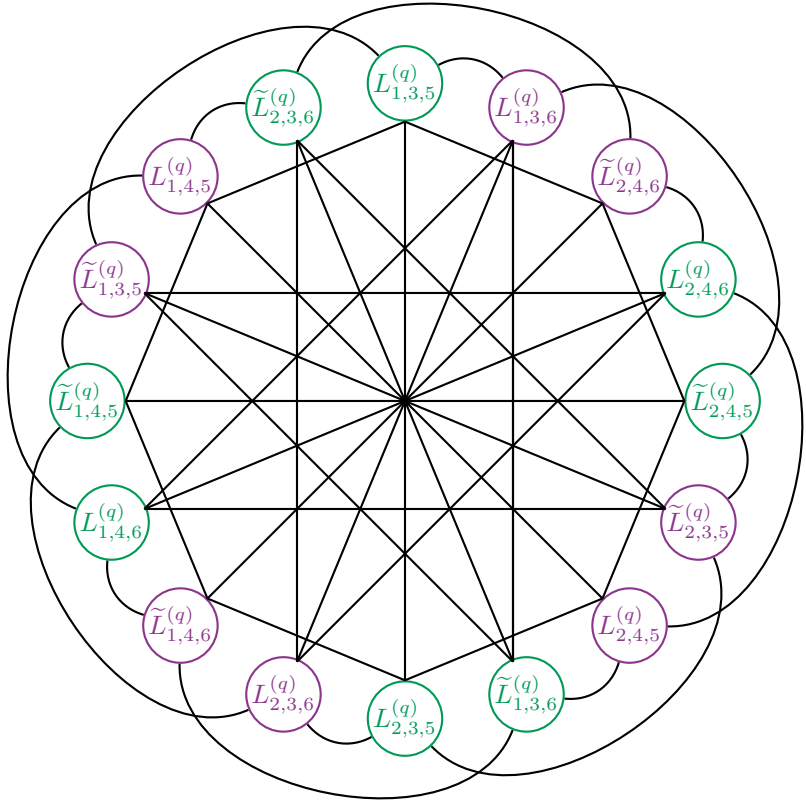
\begin{figure}[h!t]
		\centering
\begin{tikzpicture}
\tikzset{every state/.style={inner sep=-1pt,minimum size=2.8em}}
\def\ra{4.2}

\node[state,Fuchsia,thick] (Li4) at ( {\ra*cos(90-1*22.5)} , {\ra*sin(90-1*22.5)} ) {$L_{1,3,6}^{(q)}$};
\node[state,Fuchsia,thick] (Lhi1) at ( {\ra*cos(90-2*22.5)} , {\ra*sin(90-2*22.5)} ) {$\widetilde{L}_{2,4,6}^{(q)}$};

\node[state,ForestGreen,thick] (Li1) at ( {\ra*cos(90-3*22.5)} , {\ra*sin(90-3*22.5)} ) {$L_{2,4,6}^{(q)}$};
\node[state,ForestGreen,thick] (Lh04) at ( {\ra*cos(90-4*22.5)} , {\ra*sin(90-4*22.5)} ) {$\widetilde{L}_{2,4,5}^{(q)}$};

\node[state,Fuchsia,thick] (Lhi2) at ( {\ra*cos(90-5*22.5)} , {\ra*sin(90-5*22.5)} ) {$\widetilde{L}_{2,3,5}^{(q)}$};
\node[state,Fuchsia,thick] (L04) at ( {\ra*cos(90-6*22.5)} , {\ra*sin(90-6*22.5)} ) {$L_{2,4,5}^{(q)}$};

\node[state,ForestGreen,thick] (Lhi4) at ( {\ra*cos(90-7*22.5)} , {\ra*sin(90-7*22.5)} ) {$\widetilde{L}_{1,3,6}^{(q)}$};
\node[state,ForestGreen,thick] (Li2) at ( {\ra*cos(90-8*22.5)} , {\ra*sin(90-8*22.5)} ) {$L_{2,3,5}^{(q)}$};

\node[state,Fuchsia,thick] (L03) at ( {\ra*cos(90-9*22.5)} , {\ra*sin(90-9*22.5)} ) {$L_{2,3,6}^{(q)}$};
\node[state,Fuchsia,thick] (Lh02) at ( {\ra*cos(90-10*22.5)} , {\ra*sin(90-10*22.5)} ) {$\widetilde{L}_{1,4,6}^{(q)}$};

\node[state,ForestGreen,thick] (L02) at ( {\ra*cos(90-11*22.5)} , {\ra*sin(90-11*22.5)} ) {$L_{1,4,6}^{(q)}$};
\node[state,ForestGreen,thick] (Lhi3) at ( {\ra*cos(90-12*22.5)} , {\ra*sin(90-12*22.5)} ) {$\widetilde{L}_{1,4,5}^{(q)}$};

\node[state,Fuchsia,thick] (Lh01) at ( {\ra*cos(90-13*22.5)} , {\ra*sin(90-13*22.5)} ) {$\widetilde{L}_{1,3,5}^{(q)}$};
\node[state,Fuchsia,thick] (Li3) at ( {\ra*cos(90-14*22.5)} , {\ra*sin(90-14*22.5)} ) {$L_{1,4,5}^{(q)}$};

\node[state,ForestGreen,thick] (Lh03) at ( {\ra*cos(90-15*22.5)} , {\ra*sin(90-15*22.5)} ) {$\widetilde{L}_{2,3,6}^{(q)}$};
\node[state,ForestGreen,thick] (L01) at ( {\ra*cos(90-0*22.5)} , {\ra*sin(90-0*22.5)} ) {$L_{1,3,5}^{(q)}$};

\draw[thick] (L01) -- (Li2);
\draw[thick] (Lh03) -- (Lhi4);
\draw[thick] (Lhi3) -- (Lh04);
\draw[thick] (L02) -- (Li1);

\draw[thick] (Li3) -- (L04);
\draw[thick] (Lh01) -- (Lhi2);
\draw[thick] (Li4) -- (L03);
\draw[thick] (Lhi1) -- (Lh02);

\draw[black,thick] (Li3.south east) -- (L01.south);
\draw[black,thick] (L01.south) -- (Lhi1.south west);
\draw[black,thick] (Lhi1.south west) -- (Lh04.west);
\draw[black,thick] (Lh04.west) -- (L04.north west);
\draw[black,thick] (L04.north west) -- (Li2.north);
\draw[black,thick] (Li2.north) -- (Lh02.north east);
\draw[black,thick] (Lh02.north east) -- (Lhi3.east);
\draw[black,thick] (Lhi3.east) -- (Li3.south east);

\draw[black,thick] ($(Lh03.south)!0.5!(Lh03.south east)$) -- ($(Lhi2.west)!0.5!(Lhi2.north west)$);
\draw[black,thick] ($(Lhi2.west)!0.5!(Lhi2.north west)$) -- ($(L02.east)!0.5!(L02.north east)$);
\draw[black,thick] ($(L02.east)!0.5!(L02.north east)$) -- ($(Li4.south)!0.5!(Li4.south west)$) ;
\draw[black,thick] ($(Li4.south)!0.5!(Li4.south west)$) -- ($(Lhi4.north west)!0.5!(Lhi4.north)$);
\draw[black,thick] ($(Li4.south)!0.5!(Li4.south west)$) -- ($(Lhi4.north west)!0.5!(Lhi4.north)$);
\draw[black,thick] ($(Lhi4.north west)!0.5!(Lhi4.north)$) -- ($(Lh01.south east)!0.5!(Lh01.east)$);
\draw[black,thick] ($(Lh01.south east)!0.5!(Lh01.east)$) -- ($(Li1.south west)!0.5!(Li1.west)$);
\draw[black,thick] ($(Li1.south west)!0.5!(Li1.west)$) -- ($(L03.north east)!0.5!(L03.north)$);
\draw[black,thick] ($(L03.north east)!0.5!(L03.north)$) -- ($(Lh03.south)!0.5!(Lh03.south east)$);

\path[-]     (Lh03) edge   [bend left=80,thick]   (Lhi1);
\path[-]     (Lhi1) edge   [bend left=40,thick]   (Li1);
\path[-]     (Li1) edge   [bend left=80,thick]   (L04);
\path[-]     (L04) edge   [bend left=40,thick]   (Lhi4);
\path[-]     (Lhi4) edge   [bend left=80,thick]   (Lh02);
\path[-]     (Lh02) edge   [bend left=40,thick]   (L02);
\path[-]     (L02) edge   [bend left=80,thick]   (Li3);
\path[-]     (Li3) edge   [bend left=40,thick]   (Lh03);

\path[-]     (Lh01) edge   [bend left=80,thick]   (L01);
\path[-]     (L01) edge   [bend left=40,thick]   (Li4);
\path[-]     (Li4) edge   [bend left=80,thick]   (Lh04);
\path[-]     (Lh04) edge   [bend left=40,thick]   (Lhi2);
\path[-]     (Lhi2) edge   [bend left=80,thick]   (Li2);
\path[-]     (Li2) edge   [bend left=40,thick]   (L03);
\path[-]     (L03) edge   [bend left=80,thick]   (Lhi3);
\path[-]     (Lhi3) edge   [bend left=40,thick]   (Lh01);
\end{tikzpicture}
\caption{Clebsch graph encoding the configuration of lines and their points of intersection on the Segre surface $\overline{\mathcal{Z}}_q$.
The lines are explicitly described in Section \ref{subsec:lines_qpviSegre} for $0<|q|<1$ and in Section \ref{sec:contlines} when $q=1$.
Each vertex represents a line and each edge represents an intersection point of the two lines corresponding to its endpoints. The lines are colour-coded ForestGreen or Fuchsia dependent on whether, when $q=1$, they intersect the curve at infinity of $\overline{\mathcal{Z}}_1$ respectively in conic \eqref{eq:conicIcontinuum} or conic \eqref{eq:conicIIcontinuum}.}
\label{fig:lines_intersection}
\end{figure}

For the remaining eight lines, we make use of the parametrisation in Section \ref{subsec:parametrisation}.
To describe them, certain rational functions on $\overline{\mathcal{Z}}_q$, called Tyurin ratios, are helpful. To define these rational functions, note that equations \eqref{qp6:seg_comp_c} and \eqref{qp6:seg_comp_d} imply
\begin{equation*}
\frac{Z_1}{\eta_1^{(q)}}\cdot \frac{Z_2}{\eta_2^{(q)}}=\frac{Z_3}{\eta_3^{(q)}}\cdot \frac{Z_4}{\eta_4^{(q)}}=\frac{Z_5}{\eta_5^{(q)}}\cdot \frac{Z_6}{\eta_6^{(q)}}.
\end{equation*}
We now consider the following rational function,
\begin{equation*}
 {T}_{13}^{(q)}:=\frac{Z_1}{\eta_1^{(q)}}\Big/\frac{Z_3}{\eta_3^{(q)}}=\frac{Z_4}{\eta_4^{(q)}}\Big/\frac{Z_2}{\eta_2^{(q)}}. 
\end{equation*}
This is a mermorphic function on the Segre surface,
\begin{equation*}
   {T}_{13}^{(q)}: \overline{\mathcal{Z}}_q\rightarrow \mathbb{P}^1,
\end{equation*}
as there is no point on the Segre surface with $Z_1=Z_2=Z_3=Z_4=0$.

Let $\alpha\in S_6$ denote the permutation,
\begin{equation}\label{eq:alpha}
    \alpha=(1\;2)\;(3\;4)\;(5\;6),
\end{equation}
then, for any $1\leq i,j\leq 6$ not in a same cycle of $\alpha$, we similarly have a meromorphic rational function
\begin{equation}\label{def:Tyurinratios}
 {T}_{ij}^{(q)}:=\frac{Z_i}{\eta_i^{(q)}}\Big/\frac{Z_j}{\eta_j^{(q)}}=\frac{Z_{\alpha(j)}}{\eta_{\alpha(j)}^{(q)}}\Big/\frac{Z_{\alpha(i)}}{\eta_{\alpha(i)}^{(q)}}. 
\end{equation}
We call these functions Tyurin ratios, as they can be interpreted as ratios of Tyurin parameters of some elliptic matrix functions, see \cite[\S 2.4]{jr_qp6} or \cite[\S 2.2]{rof_qpvi}.
We have the following obvious symmetries
\begin{equation}\label{eq:tyurin_sym}
     {T}_{ij}^{(q)}= (T_{ji}^{(q)})^{-1}=T_{\alpha(j)\alpha(i)}^{(q)}=
     (T_{\alpha(i)\alpha(j)}^{(q)})^{-1},
\end{equation}
so that, the  $6\times 4=24$ choices of indices, only yield $6$ Tyurin ratios which are not trivially equivalent. We further have the following multiplicative formula,
\begin{equation}\label{eq:tyurin_mult}
    {T}_{ij}^{(q)}T_{jk}^{(q)}=T_{ik}^{(q)},
\end{equation}
for any choice of $1\leq i,j,k\leq 6$ in mutually disjoint cycles of $\alpha$.

The remaining eight lines can now be described as follows. For any $i\in \{1,2\}$, $j\in\{3,4\}$ and $k\in \{5,6\}$, the following set is a line in $\overline{\mathcal{Z}}_q$,
\begin{equation} \label{eq:linesqpviII}
    \widetilde{L}_{i,j,k}^{(q)}=\{Z\in \overline{\mathcal{Z}}_q:{T}_{ij}^{(q)}(Z)=\tau_{i}^{(q)}/\tau_{j}^{(q)},{T}_{jk}^{(q)}(Z)=\tau_{j}^{(q)}/\tau_{k}^{(q)}\},
\end{equation}
where
\begin{align*}
\tau_{1}^{(q)}&=\frac{\theta_q\left(t_0^{+1} q^{-\vartheta_\infty-\vartheta_0}\right)}{\theta_q\left(t_0^{+1} q^{+\vartheta_\infty-\vartheta_0} \right)}, & 
\tau_{3}^{(q)}&=\frac{\theta_q\left(q^{+\vartheta_t+\vartheta_1-\vartheta_\infty-\vartheta_0}\right)}{\theta_q\left(q^{+\vartheta_t+\vartheta_1+\vartheta_\infty-\vartheta_0} \right)}, &
\tau_{5}^{(q)}&=\frac{\theta_q\left(q^{+\vartheta_t-\vartheta_1-\vartheta_\infty-\vartheta_0}\right)}{\theta_q\left(q^{+\vartheta_t-\vartheta_1+\vartheta_\infty-\vartheta_0} \right)},\\
\tau_{2}^{(q)}&=\frac{\theta_q\left(t_0^{-1}q^{-\vartheta_\infty-\vartheta_0}\right)}{\theta_q\left(t_0^{-1} q^{+\vartheta_\infty-\vartheta_0} \right)}, &
\tau_{4}^{(q)}&=\frac{\theta_q\left(q^{-\vartheta_t-\vartheta_1-\vartheta_\infty-\vartheta_0}\right)}{\theta_q\left(q^{-\vartheta_t-\vartheta_1+\vartheta_\infty-\vartheta_0} \right)}, &
\tau_{6}^{(q)}&=\frac{\theta_q\left(q^{-\vartheta_t+\vartheta_1-\vartheta_\infty-\vartheta_0}\right)}{\theta_q\left(q^{-\vartheta_t+\vartheta_1+\vartheta_\infty-\vartheta_0}\right)}.
\end{align*}
These descriptions of the remaining eight lines follow from the explicit expressions for the lines in \cite[\S 3.3]{rof_qpvi}, where we note the following correspondence with the notation used in that paper,
\begin{align*}
\widetilde{L}_{1,3,5}^{(q)}&=\widetilde{\mathcal{L}}_1^0, &  \widetilde{L}_{1,4,6}^{(q)}&=\widetilde{\mathcal{L}}_2^0, &  \widetilde{L}_{2,3,6}^{(q)}&=\widetilde{\mathcal{L}}_3^0, &  \widetilde{L}_{2,4,5}^{(q)}&=\widetilde{\mathcal{L}}_4^0, \\
\widetilde{L}_{2,4,6}^{(q)}&=\widetilde{\mathcal{L}}_1^\infty, &  \widetilde{L}_{2,3,5}^{(q)}&=\widetilde{\mathcal{L}}_2^\infty, &  \widetilde{L}_{1,4,5}^{(q)}&=\widetilde{\mathcal{L}}_3^\infty, &  \widetilde{L}_{1,3,6}^{(q)}&=\widetilde{\mathcal{L}}_4^\infty.
\end{align*}
The correspondence between the two notations for the other eight lines is the same, with all the tildes removed. We further recall from \cite[Theorem 2.11]{rof_qpvi} the intersection graph of lines in Figure \ref{fig:lines_intersection}.

\section{The limit of $\mathcal Z_q$ as $q\uparrow 1$ and Jimbo-Fricke cubic for $\Psix$}\label{sec:pvi}
In the singular limit $q\uparrow 1$, known as the continuum limit, $q\Psix$ formally reduces to the classical sixth Painlev\'e equation. To see this explicitly, one substitutes formal Taylor series around $q=1$,
\begin{equation*}
   \begin{aligned}
       f(t)&=f_0(t)+(q-1)f_1(t)+(q-1)^2f_2(t)+\ldots,\\
       g(t)&=g_0(t)+(q-1)g_1(t)+(q-1)^2g_2(t)+\ldots,
   \end{aligned} 
\end{equation*}
and compares terms, which leads to
\begin{equation*}
    f_0=u,\quad g_0=\frac{u-t}{u-1},
\end{equation*}
where $u$ satisfies the sixth Painlev\'e equation
\begin{equation}\label{eq:pvi}
\begin{aligned}
\Psix: \quad u_{tt}&=\left(\frac{1}{u}+\frac{1}{u-1}+\frac{1}{u-t}\right)\frac{u_t^2}{2}-\left(\frac{1}{t}+\frac{1}{t-1}+\frac{1}{u-t}\right)u_t\\
&\qquad +\frac{u(u-1)(u-t)}{2t^2(t-1)^2}\biggl((2\vartheta_\infty-1)^2-\frac{4\vartheta_0^2 t}{u^2}+\frac{4\vartheta_1^2(t-1)}{(u-1)^2}\\
&\qquad\qquad\qquad\qquad +\frac{(1-4\vartheta_t^2) t(t-1)}{(u-t)^2}\biggr).
\end{aligned}
\end{equation}

In this section, we correspondingly compute the continuum limit of the Segre surface $\mathcal Z_q$ and compare it to the Jimbo-Fricke cubic for $\Psix$. The main result of the section, Theorem \ref{thm:isomorphism}, shows that the two are isomorphic as affine varieties. More precisely, the theorem states that the two are affinely equivalent
 after blowing down one of the lines at infinity of the cubic.

In Section \ref{subsec:semiclassicallimit}, the continuum limit is worked out, leading to an affine Segre surface $\mathcal{Z}_1$. We further give an algebraic characterisation of this surface and show that the descriptions of the lines on $\mathcal{Z}_q$ remain intact as $q\uparrow 1$.

In Section \ref{subsec:cubicblowdown}, we blow down the Jimbo-Fricke cubic surface to a Segre surface and describe the lines on both surfaces. In Section \ref{suse:iso} we explain how to find an isomorphism between the resulting Segre surface and the Turyn ratios.
Then, in Section \ref{subsec:isomorphism}, it is shown that the Jimbo-Friecke cubic is affinelyisomorphic to $\mathcal{Z}_1$.

\subsection{The  limit of $\mathcal Z_q$ as $q\uparrow 1$.}\label{subsec:semiclassicallimit}
With regards to the Segre surface $\mathcal{Z}_q$, note that all its coefficients, see equations \eqref{eq:qpvimu} and \eqref{eq:lambda12}, are rational in terms of simple factors of the form $\theta_q(q^\alpha)$ and $\theta_q(q^\alpha t_0)$, $\alpha\in\mathbb{C}$. To compute the limits of the coefficients as $q\uparrow 1$, it is thus sufficient to understand the limiting behaviours of these simple factors. Correspondingly, we have the following lemma.
\begin{lemma}
For any $\alpha\in\mathbb{C}$ and $t_0\in\mathbb{C}^*$, with $\arg(-t_0)<\pi$, we have the following converging limits
\begin{equation*}
    \lim_{q\uparrow 1}(1-q)^{-1}\frac{\theta_q(q^\alpha)}{(q;q)_\infty^3}=\frac{\sin{\pi \alpha}}{\pi},\qquad
    \lim_{q\uparrow 1}\frac{\theta_q(q^\alpha t_0)}{\theta_q(t_0)}=(-t_0)^{-\alpha}.
\end{equation*}
\end{lemma}
\begin{proof}
To derive the first limit, we recall that the $q$-gamma function, defined by
\begin{equation*}
    \Gamma_q(\alpha)=(1-q)^{1-\alpha}\frac{(q;q)_\infty}{(q^\alpha;q)_\infty},
\end{equation*}
converges to the usual gamma function as $q\uparrow 1$. Therefore,
\begin{equation*}
    \theta_q(q^\alpha)=\frac{(1-q)(q;q)_\infty^3}{\Gamma_q(\alpha)\Gamma_q(1-\alpha)}\sim
    \frac{(1-q)(q;q)_\infty^3}{\Gamma(\alpha)\Gamma(1-\alpha)}
    =(1-q)(q;q)_\infty^3  \frac{\sin{\pi \alpha}}{\pi},
\end{equation*}
as $q\uparrow 1$, from which the first limit in the lemma follows. The second is a direct consequence of
\begin{equation*}
    \lim_{q\uparrow 1}\frac{(q^\alpha z;q)_\infty}{(z;q)_\infty}=(1-z)^{-\alpha}\qquad (z\in \mathbb{C}\setminus [1,\infty)),
\end{equation*}
which we take from \cite[Eq. (1.3.19)]{gasperrahman}.
\end{proof}

As a consequence of the above lemma, we have the following limits for ratios of $q$-theta functions,
\begin{equation}\label{eq:limitlaws}
  \frac{\theta_q(q^\alpha)}{\theta_q(q^\beta)}\rightarrow \frac{\sin \pi \alpha}{\sin \pi \beta}\quad (\beta\notin\mathbb{Z}),\qquad \frac{\theta_q(q^\alpha t_0)}{\theta_q(q^\beta t_0)}\rightarrow (-t_0)^{\beta-\alpha}\qquad (\arg(-t_0)<\pi).
\end{equation}
From these identities, we obtain the continuum limit of the Segre surface $\mathcal Z_q$ for $q\Psix$, which is 
 a family of algebraic varieties depending on 
four generic parameters,
\begin{equation}\label{eq:pvi_parameters}
  \vartheta_0,\vartheta_t,\vartheta_1,\vartheta_\infty\in\mathbb{C},
\end{equation}
defined as the zero set of four polynomials $h_i\in\mathbb C[z_1,\dots,z_6]$, $i=1,\dots,4$
given by \eqref{qp6:seg} with the choice of parameters
 defined by 
  \begin{align}
    &\mu_1=1, & &\mu_2=1,\label{eq:mu-cl}\\
    &\mu_{3}=\prod_{\epsilon=\pm 1} \frac{\sin(\pi\left(\vartheta_t+\vartheta_1+\vartheta_\infty\right))}{\sin(\pi\left(\epsilon \,\vartheta_0+\vartheta_t+\vartheta_1+\vartheta_\infty\right))}, &
    &\mu_{4}=\prod_{\epsilon=\pm 1}\frac{\sin(\pi\left(-\vartheta_t-\vartheta_1+\vartheta_\infty\right))}{\sin(\pi\left(\epsilon \,\vartheta_0-\vartheta_t-\vartheta_1+\vartheta_\infty\right))},\nonumber\\
    &\mu_{5}=\prod_{\epsilon=\pm 1}\frac{\sin(\pi\left(\vartheta_t-\vartheta_1+\vartheta_\infty\right))}{\sin(\pi\left(\epsilon \,\vartheta_0+\vartheta_t-\vartheta_1+\vartheta_\infty\right))}, &
    &\mu_{6}=\prod_{\epsilon=\pm 1}\frac{\sin(\pi\left(-\vartheta_t+\vartheta_1+\vartheta_\infty\right))}{\sin(\pi\left(\epsilon \,\vartheta_0-\vartheta_t+\vartheta_1+\vartheta_\infty\right))}\nonumber,
\end{align}
and
\begin{align}
     \lambda_1&=\prod_{\epsilon=\pm 1} \frac{\sin(\pi\left(\epsilon\, \vartheta_0-\vartheta_t-\vartheta_1+\vartheta_\infty\right))\sin(\pi\left(\epsilon\, \vartheta_0+\vartheta_t+\vartheta_1+\vartheta_\infty\right))}{\sin(2\pi\vartheta_t)\sin(2\pi\vartheta_1)},\label{eq:la-cl}\\
    \lambda_2&=\prod_{\epsilon=\pm 1} \frac{\sin(\pi\left(\epsilon \, \vartheta_0-\vartheta_t+\vartheta_1+\vartheta_\infty\right))\sin(\pi\left(\epsilon \, \vartheta_0+\vartheta_t-\vartheta_1+\vartheta_\infty\right))}{\sin(2\pi\vartheta_t)\sin(2\pi\vartheta_1)}. \nonumber
\end{align}

Note that the $t_0$ dependence of the coefficients drops out in the continuum limit.

\begin{lemma}\label{lm:pvi_segre}
Equations \eqref{qp6:seg}, with coefficients given by \eqref{eq:mu-cl} and \eqref{eq:la-cl}, define a four-parameter family of affine Segre surfaces. Among the coefficients, there are the following non-trivial relations,
    \begin{equation}\label{p6:paramcond}
    \frac{\lambda_1}{\lambda_2}=\frac{(\mu_5-1)(\mu_6-1)}{(\mu_3-1)(\mu_4-1)},\qquad \mu_1=\mu_2=1,
\end{equation}
and these are effectively all, that is, any other relation among the coefficients is generated by them.
\end{lemma}

  Similarly to the case of $q\Psix$, we can introduce the parameters $\rho$ defined in \eqref{eq:rhos} which in the case of $q\uparrow 1$ become:
\begin{equation}\label{eq:rhos1}
   \rho_2=0,\quad  \rho_3=\frac{\mu_3-1}{\mu_4-1},\quad  \rho_5=\frac{\mu_5-1}{\mu_4-1},\quad\rho_6=\frac{\mu_6-1}{\mu_4-1},
\end{equation}  
and replace $h_2$ by $h_2''$ defined in \eqref{eqp-h2-rho}.

\begin{definition}\label{defi:pvi_segre}
    We denote the family of affine surfaces defined by equations \eqref{qp6:seg}, with coefficients given by \eqref{eq:mu-cl} and \eqref{eq:la-cl} by $\mathcal Z_1$. 
    The family $\mathcal Z_1$ is called \textit{$\Psix$ Segre}.
\end{definition}

\begin{proof}
Equations \eqref{p6:paramcond} follow by direct computation. Checking them by hand is easiest via equations \eqref{eq:Tk_identitiescontinuum} and \eqref{eq:mu_qpvilim} in Section \ref{sec:contlines}.

In terms of the parameters $\lambda_1,\lambda_2$ and  $\rho_2,\rho_3,\rho_5,\rho_6$ introduced in \eqref{eq:rhos1}, relations \eqref{p6:paramcond} yield
\begin{equation}\label{eq:rho_relations_pvi}
 \frac{\lambda_1}{\lambda_2}=\frac{\rho_5\rho_6}{\rho_3},\qquad \rho_2=0.
\end{equation}
To prove that these are effectively all the relations among the coefficients, we consider the mapping defined by \eqref{eq:mu-cl}, \eqref{eq:la-cl} and \eqref{eq:rhos1},
\begin{equation*}
   (\theta_0,\theta_t,\theta_1,\theta_\infty)\mapsto (\lambda_1,\lambda_2,\rho_5,\rho_6).
\end{equation*}
The Jacobian determinant of this mapping equals
\begin{equation*}
    -128\pi^4 \frac{(\upsilon_0-\upsilon_0^{-1})(\upsilon_\infty-\upsilon_\infty^{-1})^3}{(\upsilon_t-\upsilon_t^{-1})^3(\upsilon_1-\upsilon_1^{-1})^3}\prod_{\epsilon=\pm 1}\frac{(\upsilon_0^{\epsilon} \upsilon_\infty-\upsilon_t\upsilon_1)^2}{(\upsilon_0^{\epsilon} \upsilon_t-\upsilon_1\upsilon_\infty)(\upsilon_0^{\epsilon} \upsilon_1-\upsilon_t\upsilon_\infty)},
\end{equation*}
where we used the notation $\upsilon_k=e^{2\pi i \theta_k}$ for $k=0,t,1,\infty$. In particular, this mapping is locally biholomorphic near generic points in its domain, so that there are no relations among the coefficients, other than those generated by
\eqref{eq:rho_relations_pvi}.
\end{proof}

\begin{proposition}\label{prop:F1_smooth}
For generic parameter values $(\vartheta_0,\vartheta_t,\vartheta_1,\vartheta_\infty)\in\mathbb{C}^4$, the canonical projective completion $\overline{\mathcal{Z}}_1\subseteq \mathbb{P}^6$, of the affine Segre surface $\mathcal{Z}_1$, is smooth and the curve at infinity is reducible; it consists of two conics which intersect at two points.
\end{proposition}
\begin{proof}
We use projective coordinates,
\begin{equation*}
    [Z_0:Z_1:Z_3:Z_4:Z_5:Z_6]=[1:z_1:z_2:z_3:z_4:z_5:z_6],
\end{equation*}
so that $\overline{\mathcal{Z}}_1$ is described by equations \eqref{qp6:seg_comp}, with coefficients satisfying \eqref{p6:paramcond}.

Following the first part of the proof of Proposition \ref{prop:qpvialgebraic}, for $\overline{\mathcal{Z}}_1$
to be singular, it is necessary that $\lambda_1=0$, $\lambda_2=0$ or
\begin{equation}\label{eq:resultant}
 (\lambda_1-\lambda_2)^2-2(\lambda_1+\lambda_2)+1=0.
\end{equation}
But the left-hand side reads
\begin{equation*}
    (\lambda_1-\lambda_2)^2-2(\lambda_1+\lambda_2)+1=\frac{\sin(2\pi \vartheta_0)\sin(2\pi \vartheta_\infty\hspace{-0.2mm})}{\sin(2\pi \vartheta_t)\sin(2\pi \vartheta_1)},
\end{equation*}
which is clearly nonzero for generic parameter values. In fact, for any choice of parameters $(\vartheta_0,\vartheta_t,\vartheta_1,\vartheta_\infty)\in\mathbb{C}^4$, such that
\begin{equation}\label{eq:smoothness_conditions}
2\vartheta_0,2\vartheta_t,2\vartheta_1,2\vartheta_\infty\notin \mathbb{Z},\qquad
\epsilon_0 \vartheta_0+\epsilon_t \vartheta_t+\epsilon_1 \vartheta_1+\epsilon_\infty \vartheta_\infty\notin \mathbb{Z},
\end{equation}
for any $\epsilon_j\in\{\pm 1\}$, $j=0,t,1,\infty$, the Segre surface
$\overline{\mathcal{Z}}_1$ is well-defined and smooth. 

Next, we consider the curve at infinity, described by
\begin{subequations}\label{eq:curve_at_infinity}
\begin{align}
&{Z}_{1} + {Z}_{2} + {Z}_{3} + {Z}_{4} + {Z}_{5} + {Z}_{6}=0, \label{eq:curve_at_infinity_a}\\
&\rho_{3} {Z}_{3} +{Z}_{4} + \rho_{5} {Z}_{5} + 
  \rho_{6} {Z}_{6}=0, \label{eq:curve_at_infinity_b}\\
&{Z}_{3} {Z}_{4} - {Z}_{1} {Z}_{2} \lambda_{1}=0,\label{eq:curve_at_infinity_c}\\
&{Z}_{5} {Z}_{6} - {Z}_{1} {Z}_{2} \lambda_{2}=0.\label{eq:curve_at_infinity_d}
\end{align}
\end{subequations}
Due to the special conditions  \eqref{eq:rho_relations_pvi} on the coefficients, this curve is reducible. To see this, note that equations \eqref{eq:curve_at_infinity_c}, \eqref{eq:curve_at_infinity_d} 
imply
\begin{equation}\label{eq;inftyalt}
    Z_{3}Z_4-Z_5Z_6\frac{\lambda_1}{\lambda_2}=0.
\end{equation}
Now, recall from \eqref{eq:rho_relations_pvi} that  $\lambda_1/\lambda_2=\rho_5\rho_6/\rho_3$, and by solving equation \eqref{eq:curve_at_infinity_c} for $Z_3$ and substituting the result into \eqref{eq;inftyalt}, we obtain the following factorisation
\begin{equation*}
    (Z_4+\rho_6 Z_6)(Z_4+\rho_5 Z_5)=0.
\end{equation*}
It follows that the curve at infinity  consists of two components. The first, is described by
\begin{equation}\label{eq:conicIproof}
   \rho_3 Z_3=-\rho_5 Z_5,\qquad  
   Z_4=-\rho_6 Z_6,
\end{equation}
and
\begin{align*}
 &Z_1+Z_2=\big(\frac{\rho_5}{\rho_3}-1\big)Z_5+(\rho_6-1)Z_6,\\
 &Z_1Z_2=\lambda_2^{-1}Z_5Z_6.
\end{align*}
The second, is described by
\begin{equation}\label{eq:conicIIproof}
  \rho_3 Z_3=-\rho_6 Z_6,\qquad  
   Z_4=-\rho_5Z_5,
\end{equation}
and
\begin{align*}
 &Z_1+Z_2=(\rho_5-1)Z_5+\big(\frac{\rho_6}{\rho_3}-1\big)Z_6,\\
 &Z_1Z_2=\lambda_2^{-1}Z_5Z_6.
\end{align*}
For generic parameter values, both of these components are smooth, and thus irreducible, conics.

Next, we consider where the conics intersect. At an intersection point, we have
\begin{equation*}
    Z_0=0,\qquad \rho_3 Z_3=Z_4=-\rho_5Z_5=-\rho_6 Z_6.
\end{equation*}
It follows that the intersection points are given by
\begin{equation*}
Z_0=0, \quad Z_3=\rho_3^{-1},\quad Z_4=1,\quad Z_5=-\rho_5^{-1},\quad Z_6=-\rho_6^{-1},
\end{equation*}
with $\{Z_1,Z_2\}$ solutions of
\begin{align*}
&Z_1+Z_2=-\rho_3^{-1}-1+\rho_5^{-1}+\rho_6^{-1},\\
    &Z_1Z_2=\lambda_2^{-1}Z_5Z_6.
\end{align*}
The latter system clearly has two distinct solutions for generic parameter values,
so that the two conics intersect at two points.
The proposition follows.
\end{proof}

For an explicit description of the intersection points of the two conics in the proof of the proposition above, it is helpful to compute the continuum limits of 
the $\eta_k^{(q)}$ $1,\leq k\leq 6$, in equations \eqref{eq:Tk_identities}. These will also be helpful when describing the continuum limits of the lines. We thus compute,
\begin{equation}\label{eq:eta_limits}
\lim_{q\uparrow 1}\frac{\eta_k^{(q)}}{\theta_q(q^{\frac{1}{2}})\theta_q(t_0)^2}=\begin{cases}
    (-t_0)^{-2\theta_\infty}\eta_1^{(1)} &\text{if } k=1,\\
    (-t_0)^{+2\theta_\infty} \eta_2^{(1)}&\text{if }  k=2,\\
    \eta_k^{(1)}                     &\text{if }  3\leq k\leq 6,
    
\end{cases}    
\end{equation}
where
 \begin{align}
\eta_{1}^{(1)}&=\eta_{2}^{(1)}=+\sin(2\pi \vartheta_t)\sin(2\pi \vartheta_1),\nonumber\\
\eta_{3}^{(1)}&=-\sin(\pi(\vartheta_0+\vartheta_t+\vartheta_1+\vartheta_\infty))\sin(\pi(-\vartheta_0+\vartheta_t+\vartheta_1+\vartheta_\infty)),\nonumber\\
\eta_{4}^{(1)}&=-\sin(\pi(\vartheta_0-\vartheta_t-\vartheta_1+\vartheta_\infty))\sin(\pi(-\vartheta_0-\vartheta_t-\vartheta_1+\vartheta_\infty)),\label{eq:continuumquadratic_coef}\\
\eta_{5}^{(1)}&=+\sin(\pi(\vartheta_0+\vartheta_t-\vartheta_1+\vartheta_\infty))\sin(\pi(-\vartheta_0+\vartheta_t-\vartheta_1+\vartheta_\infty)),\nonumber\\
\eta_{6}^{(1)}&=+\sin(\pi(\vartheta_0-\vartheta_t+\vartheta_1+\vartheta_\infty))\sin(\pi(-\vartheta_0-\vartheta_t+\vartheta_1+\vartheta_\infty)).\nonumber
\end{align}

Writing $\widehat{\eta}_k^{(1)}=\eta_k^{(1)}|_{\vartheta_0=0}$, $1\leq k\leq 6$, we obtain the following expressions for the coefficients  of $\mathcal Z_1$,
\begin{equation}\label{eq:Tk_identitiescontinuum}
    \lambda_1=\frac{\eta_{3}^{(1)}\eta_{4}^{(1)}}{\eta_{1}^{(1)}\eta_{2}^{(1)}},\qquad \lambda_2=\frac{\eta_{5}^{(1)}\eta_{6}^{(1)}}{\eta_{1}^{(1)}\eta_{2}^{(1)}},\quad \mu_k=\frac{{\widehat{\eta}_k^{(1)}}}{\eta_k^{(1)}}\quad (1\leq k \leq 6).
\end{equation}
Moreover, we have the following special relations,
 \begin{equation}\label{eq:Tk_continuumhat}
    \eta_3^{(1)}-\widehat{\eta}_3^{(1)}= \eta_4^{(1)}-\widehat{\eta}_4^{(1)}=\sin(\pi \theta_0)^2=\widehat{\eta}_5^{(1)}-\eta_5^{(1)}=\widehat{\eta}_6^{(1)}-\eta_6^{(1)}.
 \end{equation}
from which we obtain the following expressions
 \begin{equation}\label{eq:mu_qpvilim}
     \mu_k=\begin{cases}1-\displaystyle \frac{\sin (\pi \vartheta_0)^2}{\eta_k^{(1)}} & \text{if $k=3,4$,}\\
     1+\displaystyle\frac{\sin (\pi \vartheta_0)^2}{\eta_k^{(1)}} & \text{if $k=5,6$.}
     \end{cases}
 \end{equation}

\begin{remark}\label{remark:conics_at_infinity}
Using equations \eqref{eq:Tk_identitiescontinuum} and \eqref{eq:mu_qpvilim}, we can rewrite the equations for the two conics at infinity in the proof of Proposition \ref{prop:F1_smooth} as follows. The first conic, see equation \eqref{eq:conicIproof}, is given by
\begin{gather}
    Z_0=0, \quad \frac{Z_3}{\eta_3^{(1)}}=\frac{Z_5}{\eta_5^{(1)}}, \quad \frac{Z_4}{\eta_4^{(1)}}=\frac{Z_6}{\eta_6^{(1)}},\label{eq:conicIcontinuum}\\
    Z_1+Z_2+(\eta_3^{(1)}+\eta_5^{(1)})\frac{Z_5}{\eta_5^{(1)}}+(\eta_4^{(1)}+\eta_6^{(1)})\frac{Z_6}{\eta_6^{(1)}}=0,\quad  \frac{Z_1Z_2}{\eta_{1}^{(1)} \eta_{2}^{(1)}}=\frac{Z_5Z_6}{\eta_{5}^{(1)} \eta_{6}^{(1)}}. \nonumber
\end{gather}
The second,  see equation \eqref{eq:conicIIproof}, is given by
\begin{gather}
    Z_0=0, \quad \frac{Z_3}{\eta_3^{(1)}}=\frac{Z_6}{\eta_6^{(1)}}, \quad \frac{Z_4}{\eta_4^{(1)}}=\frac{Z_5}{\eta_5^{(1)}},\label{eq:conicIIcontinuum}\\
    Z_1+Z_2+(\eta_4^{(1)}+\eta_5^{(1)})\frac{Z_5}{\eta_5^{(1)}}+(\eta_3^{(1)}+\eta_6^{(1)})\frac{Z_6}{\eta_6^{(1)}}=0,\quad  \frac{Z_1Z_2}{\eta_{1}^{(1)} \eta_{2}^{(1)}}=\frac{Z_5Z_6}{\eta_{5}^{(1)} \eta_{6}^{(1)}}. \nonumber
\end{gather}
Their two intersection points are given by
\begin{equation*}
Z_0=0,\quad
Z_1=\eta_{1}^{(1)}e^{+2\pi i \vartheta_\infty \epsilon},\quad
Z_2=\eta_{2}^{(1)}e^{-2\pi i \vartheta_\infty \epsilon},\quad
Z_k=\eta_{k}^{(1)}\quad (3\leq k\leq 6),
\end{equation*}
with $\epsilon=\pm 1$. 
\end{remark}

\subsubsection{Continuum limit of the lines}\label{sec:contlines}
As $q\uparrow 1$, the descriptions of the lines on $\overline{\mathcal{Z}}_q$, in Section \ref{subsec:lines_qpviSegre}, remain intact and define lines on $\overline{\mathcal{Z}}_1$.

Firstly, regarding the eight lines defined in equation \eqref{eq:linesqpviI}, the corresponding lines as $q\uparrow 1$ are given as follows.
For $i\in \{1,2\}$, $j\in\{3,4\}$ and $k\in \{5,6\}$, the line $L_{i,j,k}^{(q)}$ becomes 
\begin{equation*}
L_{i,j,k}^{(1)}:=\{Z\in\overline{\mathcal{Z}}_1:Z_i=Z_j=Z_k=0\},
\end{equation*}
as $q\uparrow 1$, which defines a line in $\overline{\mathcal{Z}}_1$.

To describe the limit of the remaining eight lines, we extend the definition of Tyurin ratios to $\overline{\mathcal{Z}}_1$, by setting
\begin{equation*}
 {T}_{ij}^{(1)}:=\frac{Z_i}{\eta_i^{(1)}}\Big/\frac{Z_j}{\eta_j^{(1)}}=\frac{Z_{\alpha(j)}}{\eta_{\alpha(j)}^{(1)}}\Big/\frac{Z_{\alpha(i)}}{\eta_{\alpha(i)}^{(1)}},
\end{equation*}
when $q=1$, for any $i,j$ not in the same cycle of the permutation $\alpha$, defined in equation \eqref{eq:alpha}.

The eight lines defined in equation \eqref{eq:linesqpviII} then converge to the following in the continuum limit.
For any $i\in \{1,2\}$, $j\in\{3,4\}$ and $k\in \{5,6\}$, the line $\widetilde{L}_{i,j,k}^{(q)}$ becomes
\begin{equation*}
    \widetilde{L}_{i,j,k}^{(1)}=\{Z\in \overline{\mathcal{Z}}_1:{T}_{ij}^{(1)}(Z)=\tau_{i}^{(1)}/\tau_{j}^{(1)},{T}_{jk}^{(1)}(Z)=\tau_{j}^{(1)}/\tau_{k}^{(1)}\},
\end{equation*}
as $q\uparrow 1$, 
where 
\begin{align*}
\tau_{1}^{(1)}&=1, & 
\tau_{3}^{(1)}&=\frac{\sin \pi\left(+\vartheta_t+\vartheta_1-\vartheta_\infty-\vartheta_0\right)}{\sin \pi\left(+\vartheta_t+\vartheta_1+\vartheta_\infty-\vartheta_0\right)}, &
\tau_{5}^{(1)}&=\frac{\sin \pi\left(+\vartheta_t-\vartheta_1-\vartheta_\infty-\vartheta_0\right)}{\sin \pi\left(+\vartheta_t-\vartheta_1+\vartheta_\infty-\vartheta_0 \right)},\\
\tau_{2}^{(1)}&=1, &
\tau_{4}^{(1)}&=\frac{\sin \pi\left(-\vartheta_t-\vartheta_1-\vartheta_\infty-\vartheta_0\right)}{\sin \pi\left(-\vartheta_t-\vartheta_1+\vartheta_\infty-\vartheta_0 \right)}, &
\tau_{6}^{(1)}&=\frac{\sin \pi\left(-\vartheta_t+\vartheta_1-\vartheta_\infty-\vartheta_0\right)}{\sin \pi\left(-\vartheta_t+\vartheta_1+\vartheta_\infty-\vartheta_0\right)}.
\end{align*}
We further note that the intersection graph of lines in Figure \ref{fig:intersectiongraph} remains valid when $q=1$. Furthermore, we have colour-coded the lines in the intersection graph in accordance to which conic at infinity in Remark \ref{remark:conics_at_infinity} they intersect.

\subsection{The Segre surface of the $\Psix$ equation}\label{subsec:cubicblowdown}
In this section, we blow down the monodromy manifold of the differential sixth Painlev\'e equation along a line to obtain a corresponding Segre surface. There are generally a number of ways to do this, but we are going to show that it is possible to choose a way that produces a surface affinely equivalent to $\mathcal Z_1$.

The monodromy manifold of the sixth Painlev\'e equation is given by the character variety
\begin{equation*}
\{(M_0,M_t,M_1)\in SL_2(\mathbb{C})^3\}/\!/SL_2(\mathbb{C}),
\end{equation*}
which, through trace coordinates \tcm{\cite{vogt1889}}
\begin{align*}
   x_1&=\operatorname{Tr} M_tM_1, &
   x_2&=\operatorname{Tr} M_0M_1, &
   x_3&=\operatorname{Tr} M_0M_t, &
   &\\
   \nu_0&=\operatorname{Tr} M_0, & 
   \nu_t&=\operatorname{Tr} M_t, &
   \nu_1&=\operatorname{Tr} M_1, &
   \nu_\infty&=\operatorname{Tr} M_1M_tM_0,
\end{align*}
is identified with the hypersurface in $\mathbb{C}^7$ defined by
\begin{equation}
    x_1x_2x_3
+x_1^2+x_2^2+x_3^2
+\omega_1 x_1+\omega_2x_2+\omega_3 x_3+\omega_4=0,\label{eq:pvicubic}
\end{equation}
where
\begin{gather}
     \omega_1:=-(\nu_0\nu_\infty+\nu_t \nu_1),\quad
    \omega_2:=-(\nu_0\nu_1+\nu_t\nu_\infty),\quad
    \omega_3:=-(\nu_0\nu_t+\nu_1\nu_\infty),\nonumber\\  \omega_4:=\nu_0^2+\nu_t^2+\nu_1^2+\nu_\infty^2+\nu_0\nu_t\nu_1\nu_\infty-4.  \label{eq:pvi-om}
\end{gather}
In the theory of Painlev\'e VI, the traces $\nu_k$ are typically considered fixed, and related to the parameters $\vartheta_k$ of the nonlinear ODE, by 
\begin{equation}
 \nu_k=\upsilon_k+\upsilon_k^{-1},\quad \upsilon_k:=e^{2\pi i \vartheta_k}\quad (k=0,t,1,\infty).\label{eq:pvi-om1}
\end{equation}
The algebraic equation \eqref{eq:pvicubic} then defines an embedded affine cubic surface with respect to the remaining variables, known as the Jimbo-Fricke cubic surface. We remark that the variables $\upsilon_k$, $k=0,t,1,\infty$, are very useful for our purposes, since every computation that follows can essentially be done over the field $\mathbb{Q}(\upsilon_0,\upsilon_t,\upsilon_1,\upsilon_\infty)$, and thus be carried out in a computer algebra system.

Let $\mathcal{X}\subseteq \mathbb{C}^3$ denote the embedded family of affine cubic surfaces defined by equations \eqref{eq:pvicubic}. Using projective coordinates,
\begin{equation*}
    [X_0:X_1:X_2:X_3]=[1:x_1:x_2:x_3],
\end{equation*}
its canonical projective completion $\overline{\mathcal{X}}\subseteq \mathbb{P}^3$, is described by
\begin{equation}\label{eq:pvicubic_comp}
\quad X_1X_2X_3
+(X_1^2+X_2^2+X_3^2)X_0
+(\omega_1 X_1+\omega_2X_2+\omega_3 X_3)X_0^2+\omega_4X_0^3=0.
\end{equation}
The hypersurface at infinity is a tritangent plane of the cubic, so that the hyperplane section at infinity is a triangle of lines,
\begin{equation*}
\overline{\mathcal{X}}\setminus\mathcal{X}=L_1^\infty\cup L_2^\infty\cup L_3^\infty,
\end{equation*}
where
\begin{equation*}
    L_k^\infty:=\{[\;0:X_1:X_2:X_3]\in\mathbb{P}^3:X_k=0\}\qquad (k=1,2,3).
\end{equation*}
Moreover, for any choice of parameters, there is no singularity at infinity. Conversely, any embedded affine cubic surface, with a hyperplane section at infinity given by a triangle of lines, consisting of only smooth points, can be brought into the form \eqref{eq:pvicubic} by an affine transformation \cite[Proposition 3.2]{Obl}.

\subsubsection{Lines}\label{sec:pvicubiclines}
For generic values of the parameters, the cubic surface $\overline{\mathcal{X}}$ is smooth. In fact, the precise conditions are given by \eqref{eq:smoothness_conditions}, see \cite{iwasaki}. According to the Cayley-Salmon theorem, in addition to the three lines at infinity, there must be a further $24$  distinct lines on the cubic. They were first written down explicitly in \cite{klimes}. Their explicit descriptions are necessary for the proof of Theorem \ref{thm:isomorphism} and so we provide them here. We give the $24$ lines in three sets of eight, such that the the lines in the $k$-th set intersect with the line $L_k^\infty$ at infinity, for $1\leq k\leq 3$. The first eight lines are given by
\begingroup
\allowdisplaybreaks
\begin{align*}
L_1:& & x_1&=\upsilon_0\upsilon_\infty+\frac{1}{\upsilon_0\upsilon_\infty}, & \upsilon_0\upsilon_\infty x_2+x_3&=\upsilon_0\,\nu_t+\upsilon_\infty \,\nu_1,\\
L_2:& & x_1&=\upsilon_0\upsilon_\infty+\frac{1}{\upsilon_0\upsilon_\infty}, & \upsilon_0\upsilon_\infty x_3+x_2&=\upsilon_0\,\nu_1+\upsilon_\infty \,\nu_t,\\
L_3:& & x_1&=\frac{\upsilon_0}{\upsilon_\infty}+\frac{\upsilon_\infty}{\upsilon_0}, & \upsilon_0\upsilon_\infty\,\nu_t\,+\nu_1&=\upsilon_0x_2+\upsilon_\infty x_3,\\
L_4:& & x_1&=\frac{\upsilon_0}{\upsilon_\infty}+\frac{\upsilon_\infty}{\upsilon_0}, & \upsilon_0\upsilon_\infty \,\nu_1\,+\nu_t&=\upsilon_0x_3+\upsilon_\infty x_2,\\
L_5:& & x_1&=\upsilon_t\upsilon_1+\frac{1}{\upsilon_t\upsilon_1}, & \upsilon_t\upsilon_1 x_2+x_3&=\upsilon_t\,\nu_0+\upsilon_1 \nu_\infty,\\
L_6:& & x_1&=\upsilon_t\upsilon_1+\frac{1}{\upsilon_t\upsilon_1}, & \upsilon_t\upsilon_1 x_3+x_2&=\upsilon_t \nu_\infty+\upsilon_1 \nu_0,\\
L_7:& & x_1&=\frac{\upsilon_t}{\upsilon_1}+\frac{\upsilon_1}{\upsilon_t}, & \upsilon_t\upsilon_1  \nu_0+\nu_\infty&=\upsilon_tx_2+\upsilon_1 x_3,\\
L_8:& & x_1&=\frac{\upsilon_t}{\upsilon_1}+\frac{\upsilon_1}{\upsilon_t}, & \upsilon_t\upsilon_1 \nu_\infty+\nu_0&=\upsilon_tx_3+\upsilon_1 x_2.
\end{align*}
\endgroup
Each of these intersects with $L_1^\infty$ at infinity. The next eight lines are given by
\begingroup
\allowdisplaybreaks
\begin{align*}
L_9:&  & x_2&=\upsilon_0\upsilon_1+\frac{1}{\upsilon_0\upsilon_1}, & \upsilon_0\upsilon_1 x_1+x_3&=\upsilon_0\,\nu_t+\upsilon_1 \nu_\infty,\\
L_{10}:& & x_2&=\upsilon_0\upsilon_1+\frac{1}{\upsilon_0\upsilon_1}, & \upsilon_0\upsilon_1 x_3+x_1&=\upsilon_0 \nu_\infty+\upsilon_1 \nu_t,\\
L_{11}:& & x_2&=\frac{\upsilon_0}{\upsilon_1}+\frac{\upsilon_1}{\upsilon_0}, & \upsilon_0\upsilon_1 \nu_t+\nu_\infty&=\upsilon_0 x_1+\upsilon_1 x_3,\\
L_{12}:& & x_2&=\frac{\upsilon_0}{\upsilon_1}+\frac{\upsilon_1}{\upsilon_0}, & \upsilon_0\upsilon_1 \nu_\infty+\nu_t&=\upsilon_0 x_3+\upsilon_1 x_1,\\
L_{13}:& & x_2&=\upsilon_t\upsilon_\infty+\frac{1}{\upsilon_t\upsilon_\infty}, & \upsilon_t\upsilon_\infty x_1+x_3&=\upsilon_t\,\nu_0+\upsilon_\infty \nu_1,\\
L_{14}:& & x_2&=\upsilon_t\upsilon_\infty+\frac{1}{\upsilon_t\upsilon_\infty}, & \upsilon_t\upsilon_\infty x_3+x_1&=\upsilon_t\,\nu_1+\upsilon_\infty \nu_0,\\
L_{15}:& & x_2&=\frac{\upsilon_t}{\upsilon_\infty}+\frac{\upsilon_\infty}{\upsilon_t}, & \upsilon_t\upsilon_\infty \nu_0+\nu_1&=\upsilon_t x_1+\upsilon_\infty x_3,\\
L_{16}:& & x_2&=\frac{\upsilon_t}{\upsilon_\infty}+\frac{\upsilon_\infty}{\upsilon_t}, & \upsilon_t\upsilon_\infty \nu_1+\nu_0&=\upsilon_t x_3+\upsilon_\infty x_1.
\end{align*}
\endgroup
Each of these intersects with $L_2^\infty$ at infinity. The final eight lines are given by
\begingroup
\allowdisplaybreaks
\begin{align*}
L_{17}:& & x_3&=\upsilon_0\upsilon_t+\frac{1}{\upsilon_0\upsilon_t}, & \upsilon_0\upsilon_t x_1+x_2&=\upsilon_0\,\nu_1+\upsilon_t \nu_\infty,\\
L_{18}:& & x_3&=\upsilon_0\upsilon_t+\frac{1}{\upsilon_0\upsilon_t}, & \upsilon_0\upsilon_t x_2+x_1&=\upsilon_0 \nu_\infty+\upsilon_t \nu_1,\\
L_{19}:& & x_3&=\frac{\upsilon_0}{\upsilon_t}+\frac{\upsilon_t}{\upsilon_0}, & \upsilon_0\upsilon_t \nu_1+\nu_\infty&=\upsilon_0 x_1+\upsilon_t x_2,\\
L_{20}:& & x_3&=\frac{\upsilon_0}{\upsilon_t}+\frac{\upsilon_t}{\upsilon_0}, & \upsilon_0\upsilon_t \nu_\infty+\nu_1&=\upsilon_0 x_2+\upsilon_t x_1,\\
L_{21}:& & x_3&=\upsilon_1\upsilon_\infty+\frac{1}{\upsilon_1\upsilon_\infty}, & \upsilon_1\upsilon_\infty x_1+x_2&=\upsilon_1\,\nu_0+\upsilon_\infty \nu_t,\\
L_{22}:& & x_3&=\upsilon_1\upsilon_\infty+\frac{1}{\upsilon_1\upsilon_\infty}, & \upsilon_1\upsilon_\infty x_2+x_1&=\upsilon_1\,\nu_t+\upsilon_\infty \nu_0,\\
L_{23}:& & x_3&=\frac{\upsilon_1}{\upsilon_\infty}+\frac{\upsilon_\infty}{\upsilon_1}, & \upsilon_1\upsilon_\infty \nu_0+\nu_t&=\upsilon_1 x_1+\upsilon_\infty x_2,\\
L_{24}:& & x_3&=\frac{\upsilon_1}{\upsilon_\infty}+\frac{\upsilon_\infty}{\upsilon_1}, & \upsilon_1\upsilon_\infty \nu_t+\nu_0&=\upsilon_1 x_2+\upsilon_\infty x_1.
\end{align*}
\endgroup
Each of these intersects with $L_3^\infty$ at infinity.

In Figure \ref{fig:intersectiongraph}, the intersection graph of the above $24$ lines is given. Each line $L_k$ intersects with one line at infinity and $9$ further lines. We proceed to describe these intersections in more detail. Consider one of the first eight lines, $L_i$, $1\leq i \leq 8$. This line intersects with precisely one other line among the first eight lines, given by $L_{i-1}$ if $i$ is even and $L_{i+1}$ if $i$ is odd. Furthermore, for any $5\leq j\leq 12$, the line $L_i$ intersects either  $L_{2j}$ or $L_{2j-1}$. This yields a total of $1+8=9$ affine intersection points of $L_i$ with other lines. Analogous accounting holds for any line in the remaining two sets of lines. In particular, there are a total of $\tfrac{1}{2}\times 24\times 9=108$ affine intersection points, corresponding to the $108$ edges in Figure \ref{fig:intersectiongraph}. Adding to this the further $27$ intersection points at infinity, yields a total of $135$ intersection points among lines in $\overline{\mathcal{X}}$.

\begin{figure}[H]
\centering
\includegraphics[width=0.9\textwidth]{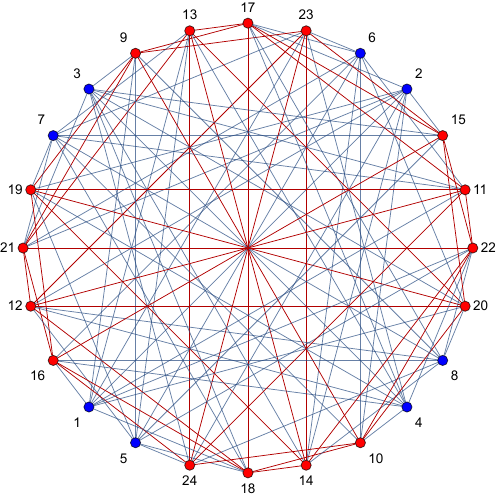}
\caption{Intersection graph of the $24$ affine lines on the Jimbo-Fricke surface. The vertex with index $k$ represents the line $L_k$, for $1\leq k\leq 24$, as defined in Section \ref{sec:pvicubiclines}. An edge represents an intersection point of the lines corresponding to its endpoints. The subgraph in red, with vertices $k$, $9\leq k\leq 24$, is the Clebsch graph which similarly encodes intersections among the lines $L_k^{\mathcal{Y}}$ on the Segre surface $\mathcal{Y}$, see Section \ref{sec:blowdownsegre}.}
   \label{fig:intersectiongraph}
\end{figure}

\subsubsection{Blowing down the monodromy manifold}\label{sec:blowdownsegre}
Blowing down any of the lines on the cubic $\overline{\mathcal{X}}$ yields a corresponding Segre (quartic) surface. Generally, different choices of lines lead to non-isomorphic Segre surfaces, since the automorphism group of a generic (projective) cubic is trivial. As we will see, blowing down the line $L_1^\infty\subseteq \overline{\mathcal{X}}$, leads to a Segre surface isomorphic to $\overline{\mathcal{Z}}_1$ coming from $q\Psix$ as $q\uparrow 1$. To realise this blow down, we introduce coordinates
\begin{equation*}
  y_1=x_1,\quad y_2=x_2,\quad y_3=x_3,\quad y_4=x_2 x_3,  
\end{equation*}
which satisfy
\begin{subequations}\label{eq:segre_blow_down}
\begin{align}
   &y_2y_3-y_4=0,\\
   &y_1y_4
+y_1^2+y_2^2+y_3^2
+\omega_1 y_2+\omega_2y_2+\omega_3 y_3+\omega_4=0.
\end{align}
\end{subequations}
These two equations define an affine Segre surface, which we denote by $\mathcal Y$. Using projective coordinates \begin{equation*}
    [Y_0:Y_1:Y_2:Y_3:Y_4]=[1:y_1:y_2:y_3:y_4],
\end{equation*}
its canonical projective completion $\overline{\mathcal{Y}}$ in $\mathbb{P}^4$, is given by
\begin{subequations}\label{eq:segre_blow_down_comp}
\begin{align}
   &Y_2Y_3-Y_0Y_4=0,\\
   &Y_1Y_4
+Y_1^2+Y_2^2+Y_3^2
+(\omega_1 Y_1+\omega_2Y_2+\omega_3 Y_3)Y_0+\omega_4Y_0^2=0.
\end{align}
\end{subequations}
The mapping $x\mapsto y$ is bi-rational, it maps $\mathcal{X}$ biholomorphically onto $\mathcal{Y}$, projects $L_1^\infty$ onto a point, and maps $L_2^\infty$ and $L_3^\infty$ to two conics which make up $\overline{\mathcal{Y}}\setminus \mathcal{Y}$, explicitly,
\begin{align}
\pi([0:0:X_2:X_3])&=[0:0:0:0:1],\label{eq:intpoint}\\
\pi([0:X_1:0:X_3])&=[0:X_1^2:0:X_1X_3:-(X_1^2+ X_3^2)],\label{eq:conic1im}\\
\pi([0:X_1:X_2:0])&=[0:X_1^2:X_1X_2:0:-(X_1^2+X_2^2)].\label{eq:conic2im}
\end{align}
It follows that
\begin{equation}\label{eq:blowdownpi}
   \begin{array}{cccc}  \pi:&\overline{\mathcal{X}}&\rightarrow &\overline{\mathcal{Y}}\\
    & x&\mapsto& y\\
     \end{array}
\end{equation} 
is indeed a blow-down of the line $L_1^\infty\subseteq \overline{\mathcal{X}}$.

Note that the two conics, traced out by \eqref{eq:conic1im} and \eqref{eq:conic2im}, are described by the following equations respectively,
\begin{subequations}\label{eq:conics}
\begin{align}
    &Y_0=0,\quad Y_2=0,\quad Y_1Y_4+Y_1^2+Y_3^2=0, \text{ and}\label{eq:conic1}\\
    &Y_0=0,\quad Y_3=0,\quad Y_1Y_4+Y_1^2+Y_2^2=0. \label{eq:conic2}
\end{align}
\end{subequations}
In particular, they are irreducible and intersect only at the point defined in equation \eqref{eq:intpoint} and the point $Y=[0:-1:0:0:1]$.

Under the blow-down, each of the lines $L_k$, $1\leq k\leq 8$, is also mapped to a conic. For example, $\pi(L_1)$ is given by
\begin{equation*}
y_1=\upsilon_0\upsilon_\infty+\frac{1}{\upsilon_0\upsilon_\infty}, \quad \upsilon_0\upsilon_\infty y_2+y_3=\upsilon_0\,\nu_t+\upsilon_\infty \,\nu_1,\quad y_4=y_2y_3.
\end{equation*}
On the other hand, each of the lines $L_k$, $9\leq k\leq 24$, is mapped to a line on the Segre surface. For instance, $\pi(L_9)$ reads
\begin{equation*}
y_2=\upsilon_0\upsilon_1+\frac{1}{\upsilon_0\upsilon_1}, \quad \upsilon_0\upsilon_1 y_1+y_3=\upsilon_0\,\nu_t+\upsilon_1 \nu_\infty,\quad y_4=\Big(\upsilon_0\upsilon_1+\frac{1}{\upsilon_0\upsilon_1}\Big)y_3.
\end{equation*}
In other words, those lines that intersect with $L_1^\infty$ become conics, and, those that do not, remain lines under the blow-down. Moreover, these are all the lines on the Segre surface. To see this, note that any line $L$ on $\overline{\mathcal{Y}}$ cannot lie inside the curve at infinity, $\overline{\mathcal{Y}}\setminus\mathcal Y$, as both conics at infinity are irreducible. Its affine part thus admits a parametrisation,
\begin{equation*}
    L:\quad y=t\,a+b\quad (t\in\mathbb{C}),
\end{equation*}
for some $a,b\in\mathbb{C}^4$, with at least one $a_k\neq 0$, $1\leq k\leq 3$, since $y_4=y_2y_3$. It follows that
\begin{equation*}
\pi^{-1}(L)=\{[t_0:a_1 t_1+b_1 t_0:a_2t_1+b_2 t_0:a_3t_1+b_3t_0]:[t_0:t_1]\in\mathbb{P}^1\},
\end{equation*}
is a line on the cubic.

So, the images 
\begin{equation}\label{eq:linesS}
  L_k^\mathcal{Y}:=\pi(L_k)\quad (9\leq k\leq 24),
\end{equation}
 are all the lines on $\overline{\mathcal{Y}}$, and their intersection graph is given by the red subgraph in Figure \ref{fig:intersectiongraph}, which is the Clebsch graph. In particular, this showcases the classical fact that a smooth Segre surface contains exactly $16$ lines \cite{segre1884}.

\subsection{Finding an isomorphism}\label{suse:iso}
We have constructed two affine Segre surfaces naturally associated with Painlev\'e VI. The first, denoted by $\mathcal{Z}_1$, is the result of the continuum limit $q\uparrow 1$ of the monodromy manifold $\mathcal{Z}_q$ for $q\Psix$. The second, denoted by $\mathcal{Y}$, was constructed by blowing down one of the lines at infinity of the Jimbo-Fricke cubic. 

It is natural to ask whether these two Segre surfaces are affinely equivalent. We answer this question in the positive in Theorem \ref{thm:isomorphism}. 

A brute force way to construct such an affine equivalence explicitly, would be to use the intersections of the $16$ lines in $\mathcal{Z}_1$ and in $\mathcal{Y}$. Any affine equivalence induces an isomorphism between the Clebsch graph in Figure \ref{fig:lines_intersection} and the Clebsch subgraph in Figure \ref{fig:intersectiongraph}. So, one could loop through all such isomorphisms and, for each, check whether it is compatible with an affine mapping from $\mathbb{C}^6$ to $\mathbb{C}^4$, using explicit expressions for some of the intersection points. If so, then one can finally check whether it maps $\mathcal{Z}_1$ onto $\mathcal{Y}$. However, this approach would be computationally heavy, devoid of leveraging any existing knowledge within the theory of the Painlev\'e equations.

Our strategy instead, is to observe that the Tyurin ratios defined in \eqref{def:Tyurinratios} appear in the asymptotic expansions for solutions of $q\Psix$, see \cite{rof_qpvi}. This is very interesting because, on one side the Tyurin ratios are natural globally defined rational function on the Segre surface $\overline{\mathcal{Z}}_q$, on the other side, their limit as $q\uparrow 1$ can be related to the asymptotics of $\Psix$, which are expressed in terms of the $x_1,x_2,x_3$ associated to the monodromy manifold.

To make the last step concrete, we will consider asymptotics of solutions to $q\Psix$ and $\Psix$ on the sheet
$\mathbb{C}\setminus\mathbb{R}_{\geq 0}$.
We thus pick a $t_0\in \mathbb{C}\setminus\mathbb{R}_{\geq 0}$ and assume $0<q<1$.
We start with explicit formulas for the asymptotics of $f(t)$, $t\in q^{\mathbb{Z}}t_0$, as $t\rightarrow 0$ and $t\rightarrow \infty$ in terms of $z\in \mathcal{Z}_q$ under the Riemann-Hilbert correspondence. 
For this purpose, the following triple of Tyurin ratios will be very helpful,
\begin{equation}\label{eq:gdefi}
g^{(q)}=(g_1^{(q)},g_2^{(q)},g_3^{(q)}),\qquad  g_1^{(q)}:=T_{62}^{(q)},\quad g_2^{(q)}:=T_{36}^{(q)},\quad g_3^{(q)}:=T_{64}^{(q)}.
\end{equation}
We recall that these are globally defined rational functions of $z\in\mathcal{Z}_q$. All the asymptotic formulas will be expressed in terms of these three Tyurin ratios.

For generic monodromy data $z\in \mathcal{Z}_q$, the following asymptotic results were obtained in \cite{rof_qpvi}. 

Define exponents $\sigma_{2}$ and $\sigma_{3}$, each with $0<\Re \sigma<\tfrac{1}{2}$, through
    \begin{equation}\label{eq:g2g3}
       g_2^{(q)}=\frac{\theta_q(q^{\sigma_2-\vartheta_t+\vartheta_\infty},q^{\sigma_2+\vartheta_t-\vartheta_\infty})}{\theta_q(q^{\sigma_2+\vartheta_t+\vartheta_\infty},q^{\sigma_2-\vartheta_t-\vartheta_\infty})},\quad
    g_3^{(q)}=\frac{\theta_q(q^{\sigma_3-\vartheta_1+\vartheta_\infty},q^{\sigma_3+\vartheta_1-\vartheta_\infty})}{\theta_q(q^{\sigma_3+\vartheta_1+\vartheta_\infty},q^{\sigma_3-\vartheta_1-\vartheta_\infty})}.
\end{equation}
The leading order asymptotic behaviour of $f(t)$ as $t\rightarrow 0$, is given by
\begin{equation*}
    f(t)\sim
q^{-\vartheta_t}\frac{\bigl(q^{\vartheta_t+\vartheta_0+\sigma_3}-1\bigr)\bigl(q^{\vartheta_t-\vartheta_0+\sigma_3}-1\bigr)\bigl(q^{\vartheta_1+\vartheta_\infty+\sigma_3}-1\bigr)}{\bigl(q^{\vartheta_1+\vartheta_\infty- \sigma_3}-1\bigr)\bigl(q^{\sigma_3}-q^{-\sigma_3}\bigr)^2} \frac{1}{c\, s}(-t)^{1-2\sigma_3}, 
\end{equation*}
 with the branch of $(-t)^{1-2\sigma_3}$ principle and $c$ and $s$ given by
\begin{align*}
c&=\frac{\Gamma_q(1-2\sigma_3)^2}{\Gamma_q(1+2\sigma_3)^2}\prod_{\epsilon=\pm1}\frac{ \Gamma_q(1+\vartheta_t+\epsilon\,\vartheta_0+\sigma_3)  \Gamma_q(1+\vartheta_1+\epsilon\,\vartheta_\infty+\sigma_3)}{ \Gamma_q(1+\vartheta_t+\epsilon\,\vartheta_0-\sigma_3)  \Gamma_q(1+\vartheta_1+\epsilon\,\vartheta_\infty-\sigma_3)},\\
s&=-(-t_0)^{-2\sigma_3}M(g_1^{(q)}),
\end{align*}
where $M(\cdot)$ is the M\"obius transformation
\begin{equation*}
M(G)=\frac{\theta_q(q^{\vartheta_1-\vartheta_\infty+\sigma_3})\theta_q(q^{\vartheta_t+\vartheta_\infty+\sigma_3}t_0^{-1})-G\,
\theta_q(q^{\vartheta_1+\vartheta_\infty+\sigma_3})\theta_q(q^{\vartheta_t-\vartheta_\infty+\sigma_3}t_0^{-1})}{\theta_q(q^{\vartheta_1-\vartheta_\infty-\sigma_3})\theta_q(q^{\vartheta_t+\vartheta_\infty-\sigma_3}t_0^{-1})-G\,
\theta_q(q^{\vartheta_1+\vartheta_\infty-\sigma_3})\theta_q(q^{\vartheta_t-\vartheta_\infty-\sigma_3}t_0^{-1})}.
\end{equation*}
We will refer to $s$ as the twist parameter.

The leading order asymptotic behaviour of $f(t)$ as $t\rightarrow \infty$, is given by
\begin{equation}\label{eq:asymptoticsinfty}
    f(t)\sim F_\infty (-1/t)^{-2\sigma_2},
\end{equation}
for a nonzero multiplier $F_\infty=F_\infty(q,\sigma_2,g_1)$, whose explicit expression is given in \cite{rof_qpvi} but will not be required for our purposes, since we only need the value of the exponent $\sigma_2$.

Next, we will consider the continuum limit of these asymptotic formulas, and compare them with Jimbo's asymptotic formulas for $\Psix$ \cite{J}.

As $q\uparrow 1$, a formal computation yields that $f(t)$ converges to a solution $u(t)$, as explained at the beginning of Section \ref{sec:pvi}. To take formal limits of the asymptotic formulas above, we use the limit laws in equation \eqref{eq:limitlaws} as well as
\begin{equation*}
    \lim_{q\uparrow 1}g_1^{(q)}=(-t_0)^{2\vartheta_\infty}g_1^{(1)},\quad
    \lim_{q\uparrow 1}g_2^{(q)}=g_2^{(1)},\quad
    \lim_{q\uparrow 1}g_3^{(q)}=g_3^{(1)},
\end{equation*}
where the discontinuity in the first expression is a consequence of the discontinuity of $\eta_2^{(q)}$ at $q=1$, see equation \eqref{eq:eta_limits}.
A direct formal computation now yields
\begin{equation*}
    u(t)\sim \frac{(\vartheta_0+\vartheta_t+\sigma_{3}\hspace{0.5mm})(-\vartheta_0+\vartheta_t+\sigma_{3}\hspace{0.5mm})(\vartheta_\infty+\vartheta_1+\sigma_{3}\hspace{0.5mm})}{4(\vartheta_\infty+\vartheta_1-\sigma_{3}\hspace{0.5mm})\sigma_{3}^2}\frac{1}{c\, s}(-t)^{1-2\sigma_{3}},
\end{equation*}
as $t\rightarrow 0$, where
\begin{align}
c&=\frac{\Gamma(1-2\sigma_3)^2}{\Gamma(1+2\sigma_3)^2}\prod_{\epsilon=\pm1}\frac{ \Gamma(1+\vartheta_t+\epsilon\,\vartheta_0+\sigma_3)  \Gamma(1+\vartheta_1+\epsilon\,\vartheta_\infty+\sigma_3)}{ \Gamma(1+\vartheta_t+\epsilon\,\vartheta_0-\sigma_3)  \Gamma(1+\vartheta_1+\epsilon\,\vartheta_\infty-\sigma_3)},\nonumber\\
s&=\frac{\sin(\pi(-\vartheta_\infty+\vartheta_1+\sigma_{3}))-g_1^{(1)}\sin(\pi(+\vartheta_\infty+\vartheta_1+\sigma_{3}))}{\sin(\pi(+\vartheta_\infty-\vartheta_1+\sigma_{3}))-g_1^{(1)}\sin(\pi(-\vartheta_\infty-\vartheta_1+\sigma_{3}))}.\label{eq:scontinuum}
\end{align}
At the same time, the asymptotic exponent at $t=\infty$ in \eqref{eq:asymptoticsinfty} is unchanged under the limit $q\uparrow 1$. Furthermore, note that, as $q\uparrow 1$, equations \eqref{eq:g2g3} become
\begin{equation}\label{eq:g2g3continuum}
\begin{aligned}
    g_2^{(1)}&=\frac{\sin(\pi(+\vartheta_\infty-\vartheta_t+\sigma_{2}))\sin(\pi(-\vartheta_\infty+\vartheta_t+\sigma_{2}))}{\sin(\pi(+\vartheta_\infty+\vartheta_t+\sigma_{2}))\sin(\pi(-\vartheta_\infty-\vartheta_t+\sigma_{2}))},\\
    g_3^{(1)}&=\frac{\sin(\pi(+\vartheta_\infty-\vartheta_1+\sigma_{3}))\sin(\pi(-\vartheta_\infty+\vartheta_1+\sigma_{3}))}{\sin(\pi(+\vartheta_\infty+\vartheta_1+\sigma_{3}))\sin(\pi(-\vartheta_\infty-\vartheta_1+\sigma_{3}))}.
\end{aligned}
\end{equation}

We now recall Jimbo's asymptotic formulas \cite{J}, which relate the exponents $\sigma_2$ and $\sigma_3$ at $t=0$ and $t=\infty$, as well as the value of twist parameter $s$, to the coordinates $(x_1,x_2,x_3)$ of the Jimbo-Fricke cubic \eqref{eq:pvicubic}.

Firstly, the exponents are given by
\begin{equation}\label{eq:x2x3formula}
x_2=2\cos(2\pi \sigma_{2}),\quad  x_3=2\cos(2\pi \sigma_{3}).
\end{equation}
The explicit formula for the twist parameter $s$ is given in \cite[eq. 1.8]{J}\footnote{We note the following typo in \cite{J}: the last `$\mp$' at the bottom of page 1141 should have been a $\pm$, as pointed out in \cite{guzzetti_tabulation}.},
and translates to the following in our notation,
\begin{align}
    s&=\frac{a+b\; e^{2\pi i\sigma_{3}}}{d},\label{eq:sjimbo}\\
    a&=\tfrac{1}{2}i \sin(2\pi \sigma_{3})x_1+\cos(2\pi\vartheta_0)\cos(2\pi\vartheta_1)+\cos(2\pi\vartheta_t)\cos(2\pi\vartheta_\infty),\nonumber\\
    b&=\tfrac{1}{2}i \sin(2\pi \sigma_{3})x_2-\cos(2\pi\vartheta_t)\cos(2\pi\vartheta_1)-\cos(2\pi\vartheta_0)\cos(2\pi\vartheta_\infty),\nonumber\\
    d&=4\prod_{\epsilon=\pm 1}\sin[\pi(\vartheta_0+\epsilon(\vartheta_t-\sigma_{3}))]\sin[\pi(\vartheta_\infty+\epsilon(\vartheta_1-\sigma_{3}))].\nonumber
\end{align}
To obtain these formulas, we note the following correspondence between the notation in Jimbo \cite{J} and ours,
\begin{align*}
    &s^{[\text{J}]}=-e^{-2\pi i \sigma_3}s,\\
    &\vartheta_k^{[\text{J}]}=2\vartheta_k\quad (k=0,t,1,\infty),\\
    &\sigma_{1t}^{[\text{J}]}=2\sigma_1,\quad
    \sigma_{01}^{[\text{J}]}=2\sigma_2',\quad
    \sigma^{[\text{J}]}=2\sigma_3,
\end{align*}
with
\begin{equation*}
 2\cos(2\pi \sigma_{1})=x_1,\quad 2\cos(2\pi \sigma_{2}')=x_2'=-x_2-x_1x_3+\nu_0\nu_1+\nu_t\nu_\infty,
\end{equation*}
where we recall that $\nu_k=2\cos(2\pi \vartheta_k)$.
The transformation $x_2\mapsto x_2'$ is an element of the standard extended modular group action \cite{iwasaki} on the Jimbo-Fricke cubic. Its application was necessary to obtain the correct result on the sheet $\mathbb{C}\setminus\mathbb{R}_{\geq 0}$, since Jimbo's original formulas are with respect to the sheet $\mathbb{C}\setminus((-\infty,0]\cup[1,+\infty))$.
 
Now, we are in a position to compare Jimbo's asymptotics formulas with those obtained from the continuum limit $q\uparrow 1$.
\begin{proposition}\label{prop:obtainingisomorphism}
Consistency of equations \eqref{eq:pvicubic}, \eqref{eq:scontinuum},  \eqref{eq:g2g3continuum}, \eqref{eq:x2x3formula} and \eqref{eq:sjimbo}
implies
\begin{equation}\label{eq:g2g3explicit}
g_2^{(1)}=\frac{x_2-(\upsilon_t\upsilon_\infty^{-1}+\upsilon_t^{-1}\upsilon_\infty)}{x_2-(\upsilon_t\upsilon_\infty+\upsilon_t^{-1}\upsilon_\infty^{-1})},\qquad
g_3^{(1)}=\frac{x_3-(\upsilon_1\upsilon_\infty^{-1}+\upsilon_1^{-1}\upsilon_\infty)}{x_3-(\upsilon_1\upsilon_\infty+\upsilon_1^{-1}\upsilon_\infty^{-1})},
\end{equation}
and
\begin{align*}
g_1^{(1)}&=\upsilon_\infty\frac{x_1-\upsilon_1\upsilon_\infty^{-1}x_2-\upsilon_t^{-1}\upsilon_\infty^{-1}x_3+x_2 x_3-
\nu_0\upsilon_\infty^{+1}-\upsilon_t\upsilon_1^{-1}+\upsilon_t^{-1}\upsilon_1 \upsilon_\infty^{-2}}{x_1-\upsilon_1\upsilon_\infty^{+1}x_2-\upsilon_t^{-1}\upsilon_\infty^{+1}x_3+x_2 x_3-
\nu_0\upsilon_\infty^{-1}-\upsilon_t\upsilon_1^{-1}+\upsilon_t^{-1}\upsilon_1 \upsilon_\infty^{+2}}\\
&=\upsilon_\infty^{-1}\frac{x_1+\upsilon_1^{-1}\upsilon_\infty^{+ 1} x_2+\upsilon_t\upsilon_\infty^{+ 1} x_3-\nu_0\upsilon_\infty^{+ 1}-\nu_1\upsilon_t-\upsilon_t^{-1}\upsilon_1 \upsilon_\infty^{+ 2}}{x_1+\upsilon_1^{-1}\upsilon_\infty^{- 1} x_2+\upsilon_t\upsilon_\infty^{- 1} x_3-\nu_0\upsilon_\infty^{- 1}-\nu_1\upsilon_t-\upsilon_t^{-1}\upsilon_1 \upsilon_\infty^{- 2}},
\end{align*}
where we remind the reader of the definition of $\nu_k$ and $\upsilon_k$, $k=0,t,1,\infty$, in equation \eqref{eq:pvi-om1}.
\end{proposition}
\begin{proof}
Let us start with deriving equations \eqref{eq:g2g3explicit}. By applying the familiar product to sum formula for the sine function to the numerator and denominator of the formula for $g_2^{(1)}$ in \eqref{eq:g2g3continuum}, we find
\begin{align*}
    g_2^{(1)}&=\frac{\sin(\pi(+\vartheta_\infty-\vartheta_t+\sigma_{2}))\sin(\pi(-\vartheta_\infty+\vartheta_t+\sigma_{2}))}{\sin(\pi(+\vartheta_\infty+\vartheta_t+\sigma_{2}))\sin(\pi(-\vartheta_\infty-\vartheta_t+\sigma_{2}))}\\
    &=\frac{\cos(2\pi \sigma_2)-\cos(2\pi(\vartheta_t-\vartheta_\infty))}{\cos(2\pi \sigma_2)-\cos(2\pi(\vartheta_t+\vartheta_\infty))}\\
    &=\frac{x_2-2\cos(2\pi(\vartheta_t-\vartheta_\infty))}{x_2-2\cos(2\pi(\vartheta_t+\vartheta_\infty))}\\
    &=\frac{x_2-(\upsilon_t\upsilon_\infty^{-1}+\upsilon_t^{-1}\upsilon_\infty)}{x_2-(\upsilon_t\upsilon_\infty+\upsilon_t^{-1}\upsilon_\infty^{-1})}.
\end{align*}
A similar computation leads to the formula for $g_3^{(1)}$ in \eqref{eq:g2g3explicit}. 

To obtain a formula for $g_1^{(1)}$, we compare the two formulas for the twist parameter that we have, equations \eqref{eq:scontinuum} and \eqref{eq:sjimbo}. Rather than dealing with trigonometric identities, we rewrite both formulas as rational functions in $\upsilon_0,\upsilon_t,\upsilon_1,\upsilon_\infty$ and $\upsilon_3$, where, by definition, $\upsilon_3=e^{2\pi i\sigma_3}$. Equation \eqref{eq:scontinuum} then reads
\begin{equation}\label{eq:scontinuumalg}
    s=\frac{(\upsilon_3\upsilon_1-\upsilon_\infty)-g_1^{(1)}(\upsilon_3\upsilon_1\upsilon_\infty-1)}{(\upsilon_3\upsilon_\infty-\upsilon_1)-g_1^{(1)}(\upsilon_3-\upsilon_1\upsilon_\infty)\hspace{0.14cm}},
\end{equation}
and equation \eqref{eq:sjimbo} becomes
\begin{align}
    s&=\frac{a+b\; \upsilon_3}{d},\label{eq:sjim}\\
    a&=\frac{1}{4}\left((\upsilon_3-\upsilon_3^{-1})x_1+(\upsilon_0+\upsilon_0^{-1})(\upsilon_1+\upsilon_1^{-1})+(\upsilon_t+\upsilon_t^{-1})(\upsilon_\infty+\upsilon_\infty^{-1})\right),\nonumber\\
    b&=\frac{1}{4}\left((\upsilon_3-\upsilon_3^{-1})x_2+(\upsilon_0+\upsilon_0^{-1})(\upsilon_\infty+\upsilon_\infty^{-1})+(\upsilon_t+\upsilon_t^{-1})(\upsilon_1+\upsilon_1^{-1})\right),\nonumber\\
    d&=\frac{(\upsilon_t-\upsilon_0\upsilon_3)(\upsilon_3-\upsilon_0\upsilon_t)(\upsilon_1-\upsilon_3\upsilon_\infty)(\upsilon_3-\upsilon_1\upsilon_\infty)}{4 \upsilon_0\upsilon_t\upsilon_1\upsilon_\infty \upsilon_3^2}.\nonumber
\end{align}

Now, we want to point out a symmetry that will be important in what follows. If we take either formula for the twist parameter $s$, and substitute $\upsilon_3$ for its reciprocal, then this is equivalent to taking the reciprocal of $s$, in other words,
\begin{equation}\label{eq:twistparametersymmetry}
    s|_{\upsilon_3\mapsto \upsilon_3^{-1}}=s^{-1}.
\end{equation}
In formula \eqref{eq:scontinuumalg}, this fact is obvious. In formula \eqref{eq:sjim} it is not that transparent and a direct computation shows that it is in fact equivalent to the cubic relation \eqref{eq:pvicubic} among $x_1,x_2$ and $x_3=\upsilon_3+\upsilon_3^{-1}$. 
We come back to this symmetry in a moment.

By equating the right-hand sides of equations \eqref{eq:scontinuumalg} and \eqref{eq:sjim}
and multiplying out the denominators, we obtain a polynomial equation, of degree $1$ in the variables $g_1^{(1)}$, $x_2$ and $x_3$, and of degree $3$ in $\upsilon_3$. Considering this as a polynomial equation just in $\upsilon_3$, we can eliminate the constant term and highest order term, by replacing respectively $1\mapsto \upsilon_3 x_3-\upsilon_3^2$
and $\upsilon_3^3\rightarrow \upsilon_3^2x_3-\upsilon_3$. Dividing the result by $\upsilon_3$, we obtain a polynomial equation in $\upsilon_3$ of degree one,
\begin{equation}\label{eq:polynomialeqg1}
\upsilon_1\upsilon_\infty^{-1}(A_+-\upsilon_\infty g_1^{(1)}A_-)-\upsilon_3(B_+-\upsilon_\infty^{-1}g_1^{(1)}B_-)=0,
\end{equation}
where
\begin{align*}
A_{\pm}&=x_1-\upsilon_1\upsilon_\infty^{\mp 1}x_2-\upsilon_t^{-1}\upsilon_\infty^{\mp 1}x_3+x_2 x_3-
\nu_0\upsilon_\infty^{\pm 1}-\upsilon_t\upsilon_1^{-1}+\upsilon_t^{-1}\upsilon_1 \upsilon_\infty^{\mp 2},\\
B_{\pm}&=x_1+\upsilon_1^{-1}\upsilon_\infty^{\pm 1} x_2+\upsilon_t\upsilon_\infty^{\pm 1} x_3-\nu_0\upsilon_\infty^{\pm 1}-\nu_1\upsilon_t-\upsilon_t^{-1}\upsilon_1 \upsilon_\infty^{\pm 2}.
\end{align*}
By symmetry \eqref{eq:twistparametersymmetry}, both equation \eqref{eq:polynomialeqg1} and the same equation with $\upsilon_3$ replaced by $\upsilon_3^{-1}$ hold true simultaneously. This implies 
\begin{equation*}
    (A_+-\upsilon_\infty g_1^{(1)}A_-)=0,\quad (B_+-\upsilon_\infty^{-1}g_1^{(1)}B_-)=0,
\end{equation*}
and therefore
\begin{equation*}
    g_1^{(1)}=\upsilon_\infty \frac{A_+}{A_-}=\upsilon_\infty^{-1}\frac{B_+}{B_-},
\end{equation*}
which are the two formulas for $g_1^{(1)}$ in the proposition, and concludes the proof.
\end{proof}

\begin{remark}
In Proposition \ref{prop:obtainingisomorphism}, two different formulas are given for $g_1^{(1)}$. Their equality is equivalent to the Jimbo-Fricke cubic equation.
\end{remark}

\subsection{An isomorphism}\label{subsec:isomorphism} 
In Section \ref{suse:iso}, we found explicit formulas relating the $z$-variables on the Segre surface $\mathcal{Z}_1$, with the $x$-variables on the Jimbo-Fricke cubic $\mathcal{X}$, through the intermediate Tyurin ratios $g^{(1)}$, see in particular Proposition \ref{prop:obtainingisomorphism}. These formulas lead to the isomorphism in the following theorem, which is the focus of this section.

\begin{theorem}\label{thm:isomorphism}
The Jimbo-Fricke cubic surface $\mathcal{X}$ and the Segre surface $\mathcal{Z}_1$, defined in Definition \ref{defi:pvi_segre}, are isomorphic as affine varieties. An explicit isomorphism is given by the polynomial mapping
\begin{equation*}
   \begin{array}{cccc}   \Phi_\mathcal{X}:&\mathcal{X}&\rightarrow &\mathcal{Z}_1\\
    & x&\mapsto& z\\
     \end{array}
\end{equation*} 
where
\begin{align*}
z_1&=+\gamma^{-1}\Bigl(\hspace{0.5mm}\upsilon_\infty\hspace{0.5mm} \Bigl(x_2-\upsilon_\infty^{-1}\nu_t\Bigr)\Bigl(x_3-\upsilon_\infty^{-1}\nu_1\Bigr)+(\upsilon_\infty-\upsilon_\infty^{-1})\Bigl(x_1-\hspace{0.5mm}\upsilon_\infty\hspace{0.5mm}\nu_0 \Bigr)\Bigr),\\
z_2&=+\gamma^{-1}\Bigl(\upsilon_\infty^{-1} \Bigl(x_2-\hspace{0.5mm}\upsilon_\infty\hspace{0.5mm}\nu_t\Bigr)\Bigl(x_3-\hspace{0.5mm}\upsilon_\infty\hspace{0.5mm}\nu_1\Bigr)-(\upsilon_\infty-\upsilon_\infty^{-1})\Bigl(x_1-\upsilon_\infty^{-1}\nu_0 \Bigr)\Bigr),\\
z_3&=-\delta^{-1}\Bigl(\upsilon_0-\upsilon_t\upsilon_1\upsilon_\infty\Bigr)\Bigl(\upsilon_0-\frac{1}{\upsilon_t\upsilon_1\upsilon_\infty}\Bigr)\Bigl(x_2-\frac{\upsilon_t}{\upsilon_\infty}-\frac{\upsilon_\infty}{\upsilon_t}\Bigr)\Bigl(x_3-\frac{\upsilon_1}{\upsilon_\infty}-\frac{\upsilon_\infty}{\upsilon_1}\Bigr),\\
z_4&=-\delta^{-1}\Bigl(\upsilon_0-\frac{\upsilon_t\upsilon_1}{\upsilon_\infty}\Bigr)\Bigl(\upsilon_0-\frac{\upsilon_\infty}{\upsilon_t\upsilon_1}\Bigr)\Bigl(x_2-\upsilon_t\upsilon_\infty-\frac{1}{\upsilon_t\upsilon_\infty}\Bigr)\Bigl(x_3-\upsilon_1\upsilon_\infty-\frac{1}{\upsilon_1\upsilon_\infty}\Bigr),\\
z_5&=+\delta^{-1}\Bigl(\upsilon_0-\frac{\upsilon_t\upsilon_\infty}{\upsilon_1}\Bigr)\Bigl(\upsilon_0-\frac{\upsilon_1}{\upsilon_t \upsilon_\infty}\Bigr)\Bigl(x_2-\frac{\upsilon_t}{\upsilon_\infty}-\frac{\upsilon_\infty}{\upsilon_t}\Bigr)\Bigl(x_3-\upsilon_1\upsilon_\infty-\frac{1}{\upsilon_1\upsilon_\infty}\Bigr),
\\
z_6&=+\delta^{-1}\Bigl(\upsilon_0-\frac{\upsilon_1\upsilon_\infty}{\upsilon_t}\Bigr)\Bigl(\upsilon_0-\frac{\upsilon_t}{\upsilon_1\upsilon_\infty}\Bigr) \Bigl(x_2-\upsilon_t\upsilon_\infty-\frac{1}{\upsilon_t\upsilon_\infty}\Bigr)\Bigl(x_3-\frac{\upsilon_1}{\upsilon_\infty}-\frac{\upsilon_\infty}{\upsilon_1}\Bigr),
\end{align*}
with the constants $\gamma$ and $\delta$ defined by
\begin{align*}
   \gamma&=(\upsilon_0-1)(\upsilon_0^{-1}-1)(\upsilon_\infty-\upsilon_\infty^{-1})^2,\\
   \delta&=(\upsilon_0-1)^2(\upsilon_t-\upsilon_t^{-1})(\upsilon_1-\upsilon_1^{-1})(\upsilon_\infty-\upsilon_\infty^{-1})^2.
\end{align*}
Composition of this isomorphism with the inverse of the blow-down mapping $\pi$ restricted to $\mathcal{Y}$,
\begin{equation*}
    \Phi_\mathcal{Y}:=\Phi_\mathcal{X}\circ (\pi^{-1}|_\mathcal{Y}),
\end{equation*}
defines an affine equivalence
\begin{equation*}
   \begin{array}{cccc}   \Phi_\mathcal{Y}:&\mathcal{Y}&\rightarrow &\mathcal{Z}_1\\
    & y&\mapsto& z\\
     \end{array}
\end{equation*} 
between the Segre surfaces $\mathcal{Y}$ and $\mathcal{Z}_1$ and, in particular, extends to a projective equivalence between their completions $\overline{\mathcal{Y}}$ and $\overline{\mathcal{Z}}_1$.
\end{theorem}
\begin{remark}
Note that the expressions for the polynomial mapping $\Phi_\mathcal{X}$ in Theorem \ref{thm:isomorphism} are affine linear in $\{x_1,x_2,x_3,x_2x_3\}$. Indeed, they can be written as
\begin{equation*}
 z_k = \xi_{0k}+\xi_{1k} x_1+\xi_{2k} x_2+\xi_{3k} x_3+\xi_{4k} x_2x_3,
\qquad 1\leq k\leq 6,   
\end{equation*}
for some coefficients that can be read off directly from the formulas in the theorem. Consequently, $\Phi_\mathcal{Y}$ can be written as
\begin{equation}\label{eq:affine_equivalence_pvi}
 z_k = \xi_{0k}+\xi_{1k} y_1+\xi_{2k} y_2+\xi_{3k} y_3+\xi_{4k} y_4,
\qquad 1\leq k\leq 6,
\end{equation}
and extends to an affine linear map from $\mathbb{C}^4$ to $\mathbb{C}^6$ that maps $\mathcal{Y}$ to $\mathcal{Z}_1$, as stated in the theorem.
\end{remark}
\begin{remark}\label{remark:lines_correspondence}
The affine equivalence $\Phi_\mathcal{Y}$ in Theorem \ref{thm:isomorphism} maps the $16$ lines in $\mathcal{Y}$, see equation \eqref{eq:linesS}, to the $16$ lines in $\mathcal{Z}_1$, see Section \ref{sec:contlines}, as follows
    \begin{equation}\label{eq:linescorrespondence}
\begin{aligned}
L_{9}^\mathcal{Y}&\mapsto \widetilde{L}_{2,4,5}^{(1)},   & 
L_{10}^\mathcal{Y}&\mapsto \widetilde{L}_{1,4,5}^{(1)},   & 
L_{11}^\mathcal{Y}&\mapsto \widetilde{L}_{2,3,6}^{(1)},   & 
L_{12}^\mathcal{Y}&\mapsto \widetilde{L}_{1,3,6}^{(1)},   \\
L_{13}^\mathcal{Y}&\mapsto L_{2,4,6}^{(1)},   & 
L_{14}^\mathcal{Y}&\mapsto L_{1,4,6}^{(1)},   & 
L_{15}^\mathcal{Y}&\mapsto L_{1,3,5}^{(1)},   & 
L_{16}^\mathcal{Y}&\mapsto L_{2,3,5}^{(1)},   \\
L_{17}^\mathcal{Y}&\mapsto \widetilde{L}_{2,4,6}^{(1)},   & 
L_{18}^\mathcal{Y}&\mapsto \widetilde{L}_{1,4,6}^{(1)},   & 
L_{19}^\mathcal{Y}&\mapsto \widetilde{L}_{2,3,5}^{(1)},   & 
L_{20}^\mathcal{Y}&\mapsto \widetilde{L}_{1,3,5}^{(1)},   \\
L_{21}^\mathcal{Y}&\mapsto L_{2,4,5}^{(1)},   & 
L_{22}^\mathcal{Y}&\mapsto L_{1,4,5}^{(1)},   & 
L_{23}^\mathcal{Y}&\mapsto L_{1,3,6}^{(1)},   & 
L_{24}^\mathcal{Y}&\mapsto L_{2,3,6}^{(1)}.
\end{aligned}
\end{equation}
In particular, it induces an isomorphism between the Clebsch graph in Figure \ref{fig:lines_intersection}  and the Clebsch subgraph in Figure \ref{fig:intersectiongraph}.
\end{remark}
\begin{remark}\label{remark:conics_correspondence}
Recall that the curve at infinity $\overline{\mathcal{Y}}\setminus\mathcal{Y}$ factorises into two conics \eqref{eq:conic1} and \eqref{eq:conic2} and, similarly, the curve at infinity $\overline{\mathcal{Z}}_1\setminus\mathcal{Z}_1$ factorises into two conics \eqref{eq:conicIcontinuum} and \eqref{eq:conicIIcontinuum}. The extension of $\Phi_\mathcal{Y}$ to a projective equivalence between $\overline{\mathcal{Y}}$ and $\overline{\mathcal{Z}}_1$ in Theorem \ref{thm:isomorphism}, maps conics \eqref{eq:conic1} and \eqref{eq:conic2} respectively to conics \eqref{eq:conicIcontinuum} and \eqref{eq:conicIIcontinuum}.

This is consistent with Remark \ref{remark:lines_correspondence} in the following way.
In Figure \ref{fig:lines_intersection}, lines in $\mathcal{Z}_1$ are coloured respectively ForestGreen or Fuchsia dependent on whether they intersect conic \eqref{eq:conicIcontinuum} or conic \eqref{eq:conicIIcontinuum} at infinity. The lines $L_k^\mathcal{Y}$, $9\leq k\leq 16$, which intersect  with conic \eqref{eq:conic1} at infinity, are mapped to the ForestGreen lines, whilst the lines $L_k^\mathcal{Y}$, $17\leq k\leq 24$, which  intersect  with conic \eqref{eq:conic2} at infinity, are mapped to the Fuchsia ones.
\end{remark}

\begin{remark}\label{remark:confluence}
  A question is raised in Item 7.2.10 (b) of \cite[\S 7.2]{RSpreprint}, about the correspondence between the monodromy manifold of $q\Psix$ and that of $\Psix$ based on the number of lines each contains. It is stated that ``\emph{It seems dubious that it can be tranlated into a birational map: the number of lines increases by confluence.''} (\/Note that the term confluence in this quote refers to the continuum limit.\/) But it is well known that one can increase the number of lines on an algebraic surface by allowing blow-ups. 
 Theorem \ref{thm:isomorphism}, shows that there is a  birational map of a most simple kind involving a blow-up between these two monodromy manifolds.
\end{remark}

We prove Theorem \ref{thm:isomorphism} in several steps and in a reverse order compared to the way the theorem is stated. Namely, we use the results in Section \ref{suse:iso}, to construct a bi-rational mapping between $\overline{\mathcal{Y}}$ and $\overline{\mathcal{Z}}_1$. We show that this mapping is a projective equivalence with the right properties, so that it induces an affine equivalence between $\mathcal{Y}$ and $\mathcal{Z}_1$ and in the end leads to the isomorphism $\Phi_\mathcal{X}$. We further prove Remarks \ref{remark:lines_correspondence} and \ref{remark:conics_correspondence} on the way.

Our starting point for the proof of Theorem \ref{thm:isomorphism}, is the set of explicit formulas obtained in Section \ref{suse:iso} that relate the $z$-variables on $\mathcal{Z}_1$ with the $x$-variables on $\mathcal{X}$, through the intermediate Tyurin ratios $g_k^{(1)}$, $1\leq k\leq 3$. We recall them here, but written with respect to homogeneous coordinates $Z$ and $Y$ on $\overline{\mathcal{Z}}_1$ and $\overline{\mathcal{Y}}$ respectively.

On the one hand, with respect to homogeneous coordinates for $\overline{\mathcal{Z}}_1$, we have
\begin{equation}\label{eq:ztogdescription}
\begin{aligned}
   g_1^{(1)}&=\frac{Z_6/\eta_6^{(1)}}{Z_2/\eta_2^{(1)}}=\frac{Z_1/\eta_1^{(1)}}{Z_5/\eta_5^{(1)}},\\
   g_2^{(1)}&=\frac{Z_3/\eta_3^{(1)}}{Z_6/\eta_6^{(1)}}=\frac{Z_5/\eta_5^{(1)}}{Z_4/\eta_4^{(1)}},\\
   g_3^{(1)}&=\frac{Z_6/\eta_6^{(1)}}{Z_4/\eta_4^{(1)}}=\frac{Z_3/\eta_3^{(1)}}{Z_5/\eta_5^{(1)}}.
\end{aligned}
\end{equation}
On the other hand,  with respect to homogeneous coordinates for $\overline{\mathcal{Y}}$, we have, see Proposition \ref{prop:obtainingisomorphism},
\begin{align}
g_1^{(1)}&=\upsilon_\infty\frac{Y_1-\upsilon_1\upsilon_\infty^{-1}Y_2-\upsilon_t^{-1}\upsilon_\infty^{-1}Y_3+Y_4-Y_0(
\nu_0\upsilon_\infty^{+1}+\upsilon_t\upsilon_1^{-1}-\upsilon_t^{-1}\upsilon_1 \upsilon_\infty^{-2})}{Y_1-\upsilon_1\upsilon_\infty^{+1}Y_2-\upsilon_t^{-1}\upsilon_\infty^{+1}Y_3+Y_4-Y_0(
\nu_0\upsilon_\infty^{-1}+\upsilon_t\upsilon_1^{-1}-\upsilon_t^{-1}\upsilon_1 \upsilon_\infty^{+2})}\nonumber\\
&=\upsilon_\infty^{-1}\frac{Y_1+\upsilon_1^{-1}\upsilon_\infty^{+ 1} Y_2+\upsilon_t\upsilon_\infty^{+ 1} Y_3-Y_0(\nu_0\upsilon_\infty^{+ 1}+\nu_1\upsilon_t+\upsilon_t^{-1}\upsilon_1 \upsilon_\infty^{+ 2})}{Y_1+\upsilon_1^{-1}\upsilon_\infty^{- 1} Y_2+\upsilon_t\upsilon_\infty^{- 1} Y_3-Y_0(\nu_0\upsilon_\infty^{- 1}+\nu_1\upsilon_t+\upsilon_t^{-1}\upsilon_1 \upsilon_\infty^{- 2})}.\nonumber\\
    g_2^{(1)}&=\frac{Y_2-Y_0(\upsilon_t\upsilon_\infty^{-1}+\upsilon_t^{-1}\upsilon_\infty)}{Y_2-Y_0(\upsilon_t\upsilon_\infty+\upsilon_t^{-1}\upsilon_\infty^{-1})},\label{eq:ytogdescription}\\
g_3^{(1)}&=\frac{Y_3-Y_0(\upsilon_1\upsilon_\infty^{-1}+\upsilon_1^{-1}\upsilon_\infty)}{Y_3-Y_0(\upsilon_1\upsilon_\infty+\upsilon_1^{-1}\upsilon_\infty^{-1})}.\nonumber
\end{align}

We correspondingly define the following mappings
\begin{align}
g_\mathcal{Z}&:\overline{\mathcal{Z}}_1\rightarrow (\mathbb{P}^1)^3,\quad Z\mapsto g^{(1)},\label{eq:gzdef}\\
g_\mathcal{Y}&:\overline{\mathcal{Y}}\rightarrow (\mathbb{P}^1)^3,\quad Y\mapsto g^{(1)},\label{eq:gydef}
\end{align}
which we prove to be analytic mappings in the following lemma.

\begin{lemma}\label{lem:analyticmaps}
The mappings $g_\mathcal{Z}$ and $g_\mathcal{Y}$ are analytic, that is, each of their components defines a meromorphic function on the respective complex varieties $\overline{\mathcal{Z}}_1$ and  $\overline{\mathcal{Y}}$.
\end{lemma}
\begin{proof}
For the mapping $g_\mathcal{Z}$ this assertion is trivial, since its components are Tyurin ratios on $\overline{\mathcal{Z}}_1$, which we know to be meromorphic functions (see the argument in the beginning of Section \ref{subsec:lines_qpviSegre}).

Regarding $g_\mathcal{Y}$, the components $g_2^{(1)}$ and $g_3^{(1)}$ are clearly meromorphic as they can be written as compositions of M\"obius transforms with the respective projections $y\mapsto y_2$ and $y\mapsto y_3$ from $\overline{\mathcal{Y}}$ to $\mathbb{P}^1$.

Regarding $g_1^{(1)}$, we have two formulas \eqref{eq:ytogdescription}, which are equivalent on $\overline{\mathcal{Y}}$. Now, if $g_1^{(1)}$ is not meromorphic at some point, then this point must be a common root of the numerators and denominators in both formulas. Each of those numerators and denominators is projective linear, and thus defines a hyperplane in $\mathbb{P}^4$. Their intersection is a single point, given by
\begin{align*}
Y_0&=\upsilon_t \upsilon_1-\upsilon_t^{-1} \upsilon_1^{-1},\\
Y_1&=\upsilon_t^2 \upsilon_1^2-\upsilon_t^{-2} \upsilon_1^{-2},\\
Y_2&=\nu_0(\upsilon_t-\upsilon_t^{-1})+\nu_\infty(\upsilon_1-\upsilon_1^{-1}),\\
Y_3&=\nu_0(\upsilon_1-\upsilon_1^{-1})+\nu_\infty(\upsilon_t-\upsilon_t^{-1}),\\
Y_4&=\left(\upsilon_t \upsilon_1-\upsilon_t^{-1} \upsilon_1^{-1}\right)\left(\nu_0\nu_\infty-(\upsilon_t-\upsilon_t^{-1})(\upsilon_1-\upsilon_1^{-1})\right).
\end{align*}
This point does not lie on $\overline{\mathcal{Y}}$ for generic parameter values, as it for example does not satisfy $Y_0Y_4-Y_2Y_3=0$. It follows that  $g_1^{(1)}$ is also a meromorphic function on $\overline{\mathcal{Y}}$, and the lemma follows.
\end{proof}

Since $\overline{\mathcal{Z}}_1$ is two-dimensional, the meromorphic functions $g_k^{(1)}$, $1\leq k\leq 3$, cannot be independent. They are explicitly related by the following identity,
\begin{equation}\label{eq:tyurincubic}
\eta_{1}^{(1)}g_2^{(1)} g_1^{(1)}+\eta_{2}^{(1)}g_3^{(1)}/g_1^{(1)}+\eta_{3}^{(1)}g_2^{(1)}g_3^{(1)}+\eta_{4}^{(1)}+\eta_{5}^{(1)}g_2^{(1)}+\eta_{6}^{(1)}g_3^{(1)}=0.
\end{equation}
To verify this identity, we take the left-hand side and first apply the multiplication rule \eqref{eq:tyurin_mult}, and then apply the symmetry ${T}_{ij}^{(1)}= T_{\alpha(j)\alpha(i)}^{(1)}$ a couple of times:
\begin{align*}
&\eta_{1}^{(1)}g_2^{(1)} g_1^{(1)}+\eta_{2}^{(1)}g_3^{(1)}/g_1^{(1)}+\eta_{3}^{(1)}g_2^{(1)}g_3^{(1)}+\eta_{4}^{(1)}+\eta_{5}^{(1)}g_2^{(1)}+\eta_{6}^{(1)}g_3^{(1)}=\\
&\eta_{1}^{(1)}{T}^{(1)}_{36}{T}^{(1)}_{62}+\eta_{2}^{(1)}{T}^{(1)}_{64}/{T}^{(1)}_{62}+\eta_{3}^{(1)}{T}^{(1)}_{36}{T}^{(1)}_{64}+\eta_{4}^{(1)}+\eta_{5}^{(1)}{T}^{(1)}_{36}+\eta_{6}^{(1)}{T}^{(1)}_{64}=\\
 &\eta_{1}^{(1)}{T}^{(1)}_{32}+\eta_{2}^{(1)}{T}^{(1)}_{24}+\eta_{3}^{(1)}{T}^{(1)}_{34}+\eta_{4}^{(1)}+\eta_{5}^{(1)}{T}^{(1)}_{36}+\eta_{6}^{(1)}{T}^{(1)}_{64}=\\
 &\eta_{1}^{(1)}{T}^{(1)}_{14}+\eta_{2}^{(1)}{T}^{(1)}_{24}+\eta_{3}^{(1)}{T}^{(1)}_{34}+\eta_{4}^{(1)}+\eta_{5}^{(1)}{T}^{(1)}_{54}+\eta_{6}^{(1)}{T}^{(1)}_{64}=\\
 &\frac{\eta_4^{(1)}}{Z_4}(Z_1+Z_2+Z_3+Z_4+Z_5+Z_6)=0,
\end{align*}
where the last equality follows from \eqref{qp6:seg_comp_a}.

Equation \eqref{eq:tyurincubic} is an equality among meromorphic functions on 
$\overline{\mathcal{Z}}_1$. However, it also defines a surface.
\begin{definition}\label{defi:surfaceG}
Define $\overline{\mathcal{G}}\subseteq \mathbb{P}^1\times \mathbb{P}^1\times \mathbb{P}^1$ as the topological closure of the surface defined by equation \eqref{eq:tyurincubic} in $\{(g_1,g_2,g_3)\in(\mathbb{C}^*)^3\}$.  Also, set
\begin{equation*}
    \mathcal{G}:=\overline{\mathcal{G}}\setminus \{g\in \overline{\mathcal{G}}: g_2=1\text{ or }g_3=1\},
\end{equation*}
so that $\overline{\mathcal{G}}$ is the topological closure of $\mathcal{G}$.
\end{definition}
\begin{remark}
Even though we defined the surface $\overline{\mathcal{G}}$ through topological closure, we note that it makes algebro-geometric sense. Namely, under the generalised Segre embedding
\begin{equation*}
\begin{array}{ccc}
     \mathbb{P}^1\times \mathbb{P}^1\times \mathbb{P}^1&\rightarrow &\mathbb{P}^7,\\
    (g_1,g_2,g_3)&\mapsto& U,\\
\end{array}
\end{equation*}
where, in homogeneous coordinates $g_k=[g_k^a:g_k^b]$, $1\leq k\leq 3$,
\begin{align*}
 &U=[U_{aaa}:U_{aab}:U_{aba}:U_{baa}:U_{abb}:U_{bab}:U_{bba}:U_{bbb}],\\
 &U_{s_1s_2s_3}=g_1^{s_1} g_2^{s_2} g_3^{s_3},\qquad (s_1,s_2,s_3\in\{a,b\}),
\end{align*}
the surface $\overline{\mathcal{G}}$ becomes a projective variety in $\mathbb{P}^7$. This projective variety admits a cumbersome description involving $13$ quadratic equations, $9$ of which describe the image of $\mathbb{P}^1\times \mathbb{P}^1\times \mathbb{P}^1$ itself under the embedding. For our purposes, it is much more convenient to work with the analytic definition of $\overline{\mathcal{G}}$ in $\mathbb{P}^1\times \mathbb{P}^1\times \mathbb{P}^1$.
\end{remark}

By a similar token, since $\overline{\mathcal{Y}}$ is two-dimensional,
the meromorphic functions $g_k^{(1)}$, $1\leq k\leq 3$, considered as functions of $Y$ through equations \eqref{eq:ytogdescription},
 cannot be independent. Remarkably, they are related by the exact same formula, equation \eqref{eq:tyurincubic}, as can be verified by direct computation.

For the next important preparatory step, for the proof of Theorem \ref{thm:isomorphism}, we define three dense open subsets
\begin{equation*}
   U_{\mathcal{Z}}\subseteq \overline{\mathcal{Z}}_1,\qquad
   U_{\mathcal{G}}\subseteq \overline{\mathcal{G}},\qquad
   U_{\mathcal{Y}}\subseteq \overline{\mathcal{Y}}.
\end{equation*}
The set $U_{\mathcal{Z}}$ is defined by the Segre surface $\overline{\mathcal{Z}}_1$, minus eight lines and the hyperplane section at infinity,
\begin{equation}\label{eq:defi_open_F}
  U_{\mathcal{Z}}=\mathcal{Z}_1\setminus \bigcup_{(i,j,k)\in I}L_{i,j,k}^{(1)},\qquad I:=\{1,2\}\times\{3,4\}\times \{5,6\},  
\end{equation}
where we recall the notation for lines on $\mathcal{Z}_1$ in Section \ref{sec:contlines}.
Similarly, $U_{\mathcal{Y}}$ is defined by the Segre surface $\overline{\mathcal{Y}}$, minus eight lines and the hyperplane section at infinity,
\begin{equation*}
U_{\mathcal{Y}}=\mathcal{Y}\setminus \bigcup_{j\in J}L_k^\mathcal{Y},\qquad J:=\{13,14,15,16,21,22,23,24\}, 
\end{equation*}
where we recall the notation for lines on $\overline{\mathcal{Y}}$ introduced in equation \eqref{eq:linesS}.\\
Finally, $U_{\mathcal{G}}$ is simply defined as the intersection
\begin{equation*}
U_{\mathcal{G}}=\mathcal{G}\cap(\mathbb{C}^*\times \mathbb{C}^*\times \mathbb{C}^*),
\end{equation*}
where we recall that $\mathcal{G}$ is defined in Definition \ref{defi:surfaceG}.

These dense open subsets are chosen exactly such that $g_1^{(1)},g_2^{(1)},g_3^{(1)}$ take finite, nonzero values on them. Regarding $U_{\mathcal{G}}$, this is tautological; regarding the two other open subsets, it follows from the following lemma.
\begin{lemma}\label{lem:denseopenbiholo}
    The mapping $g_\mathcal{Z}$ maps $U_\mathcal{Z}$ biholomorphically onto $U_\mathcal{G}$. Similarly, $g_\mathcal{Y}$ maps $U_\mathcal{Y}$ biholomorphically onto $U_\mathcal{G}$.
\end{lemma}
\begin{proof}
We start with the, due to Lemma \ref{lem:analyticmaps}, analytic mapping $g_\mathcal{Z}$. Note that, for any Tyurin ratio ${T}_{ij}^{(1)}$, 
\begin{align*}
    {T}_{ij}^{(1)}=0&\iff Z_i=Z_{\alpha(j)}=0,\text{ and}\\
        {T}_{ij}^{(1)}=\infty&\iff Z_j=Z_{\alpha(i)}=0.
\end{align*}
In particular, if $T^{(1)}(Z)\in\{0,\infty\}$, for some Tyurin ratio $T^{(1)}$, then $Z$ necessarily lies on one of the eight lines $L_{i,j,k}^{(1)}$, $(i,j,k)\in I$. By construction, those lines do not intersect with $U_{\mathcal{Z}}$. Since all three components of $g_\mathcal{Z}$ are Tyurin ratios, it follows that $g_\mathcal{Z}(U_{\mathcal{Z}})\subseteq (\mathbb{C}^*)^3$.

Now, take any point $z\in U_\mathcal{Z}$ and let $g^{(1)}=g_\mathcal{Z}(z)\in (\mathbb{C}^*)^3$, then the cubic equation \eqref{eq:tyurincubic} is satisfied by $g^{(1)}$. To infer that $g^{(1)}\in U_\mathcal{G}$, it remains to be checked that neither $g_2^{(1)}=1$ nor $g_3^{(1)}=1$.

Suppose, on the contrary, that $g_2^{(1)}=1$. Then
\begin{equation}\label{eq:g2is1}
Z_6=\frac{\eta_6^{(1)}}{\eta_3^{(1)}}Z_3,\quad Z_5=\frac{\eta_5^{(1)}}{\eta_4^{(1)}}Z_4.
\end{equation}
We now consider the difference of the first two defining equations, \eqref{qp6:seg_comp_a} and \eqref{qp6:seg_comp_b}, of the Segre surface $\overline{\mathcal{Z}}_1$, given by
\begin{equation*}
(\mu_3-1)Z_3+(\mu_4-1)Z_4+(\mu_5-1)Z_5+(\mu_6-1)Z_6=Z_0.
\end{equation*}
Direct substitution of \eqref{eq:g2is1} and the equation for $\mu_k$, $3\leq k\leq 6$, in \eqref{eq:Tk_identitiescontinuum} leads to
\begin{equation*}
1/\eta_3^{(1)}(\widehat{\eta}_3^{(1)}-\eta_3^{(1)}+\widehat{\eta}_6^{(1)}-\eta_6^{(1)})Z_3+1/\eta_4^{(1)}(\widehat{\eta}_4^{(1)}-\eta_4^{(1)}+\widehat{\eta}_5^{(1)}-\eta_5^{(1)})Z_4=Z_0.
\end{equation*}
Now, both the coefficient of $Z_3$ and of $Z_4$ in the above expression are zero, due to equation \eqref{eq:Tk_continuumhat}, and it follows that $Z_0=0$. So $z$ lies in the hyperplane section at infinity of $\overline{\mathcal{Z}}_1$ and in particular $z\notin U_{\mathcal{Z}}$. Similarly, it is shown that $g_3^{(1)}\neq 1$ on $U_{\mathcal{Z}}$. It follows that
\begin{equation*}
    g_\mathcal{Z}(U_{\mathcal{Z}})\subseteq U_\mathcal{G}.
\end{equation*}

Now, $g_\mathcal{Z}$ has a rational inverse on $\overline{\mathcal{G}}$,
 given by
 \begin{equation}\label{eq:gzinverse}
\begin{aligned}
  Z_1&=\eta_{1}^{(1)}g_2^{(1)}g_1^{(1)}, & Z_3&=\eta_3^{(1)} g_2^{(1)}g_3^{(1)} & Z_5&=\eta_5^{(1)} g_2^{(1)},\\
  Z_2&=\eta_{2}^{(1)}g_3^{(1)}/g_1^{(1)}, & Z_4&=\eta_4^{(1)},  & Z_6&=\eta_6^{(1)}g_3^{(1)},  
\end{aligned}
 \end{equation}
and
\begin{equation}\label{eq:Z0defi}
Z_0= \widehat{\eta}_{1}^{(1)}g_2^{(1)} g_1^{(1)}+\widehat{\eta}_{2}^{(1)}g_3^{(1)}/g_1^{(1)}+\widehat{\eta}_{3}^{(1)}g_2^{(1)}g_3^{(1)}+\widehat{\eta}_{4}^{(1)}+\widehat{\eta}_{5}^{(1)}g_2^{(1)}+\widehat{\eta}_{6}^{(1)}g_3^{(1)}.
\end{equation}
Note that this mapping is analytic on $U_\mathcal{G}$ and the image of $U_\mathcal{G}$ is a subset of $\overline{\mathcal{Z}}_1$. Furthermore, note that none of the $Z_k$, $1\leq k\leq 6$, can equal zero on $U_\mathcal{G}$, since the coordinates $g_k^{(1)}$, $1\leq k\leq 3$, are by definition finite and nonzero. In particular, the image of this mapping is disjoint from any of the lines $L_{i,j,k}^{(1)}$, $(i,j,k)\in I$.

It remains to be checked that the image of the mapping is disjoint with the hyperplane section at infinity of $\overline{\mathcal{Z}}_1$, that is, the right-hand side of equation \eqref{eq:Z0defi} does not vanish on $U_\mathcal{G}$. Take a $g^{(1)}\in U_\mathcal{G}$, then $g_2^{(1)},g_3^{(1)}\neq 1$. The resultant, with respect to $g_1^{(1)}$, of equation \eqref{eq:tyurincubic} and the right-hand side of equation \eqref{eq:Z0defi}, each after multiplication by $g_1^{(1)}$, is given by
\begin{equation*}
   \tfrac{1}{256}g_2^{(1)}g_3^{(1)}(g_2^{(1)}-1)^2(g_3^{(1)}-1)^2(\upsilon_0-1)^2(\upsilon_0^{-1}-1)^2(\upsilon_t-\upsilon_t^{-1})^2(\upsilon_1-\upsilon_1^{-1})^2.
\end{equation*}
In particular, it is nonzero and thus the right-hand side of equation \eqref{eq:Z0defi} cannot vanish on $U_\mathcal{G}$.
It follows that the rational inverse of $g_\mathcal{Z}$, given above, is analytic and maps $U_\mathcal{G}$ into $U_{\mathcal{Z}}$. We conclude that $g_\mathcal{Z}$ maps $U_{\mathcal{Z}}$ biholomorphically onto $U_{\mathcal{G}}$.

Next, we consider the analytic mapping $g_\mathcal{Y}$ on $\overline{\mathcal{Y}}$. Its three components satisfy the cubic relation \eqref{eq:tyurincubic}. From the explicit formula for $g_2^{(1)}$, it follows that $g_2^{(1)}=0$ if and only if 
\begin{equation}\label{eq:hyperplaneg20}
    Y_2-Y_0\left(\frac{\upsilon_t}{\upsilon_\infty}+\frac{\upsilon_\infty}{\upsilon_t}\right)=0.
\end{equation}
On the other hand, the intersection of the Segre surface $\overline{\mathcal{Y}}$ with the hyperplane \eqref{eq:hyperplaneg20}, is given by the union of the two lines $L_{15}$ and $L_{16}$. In this way, we obtain the following equivalences,
\begin{subequations}\label{eq:equivalencesg2g3}
\begin{align}
    &g_2^{(1)}=0 \iff Y\in L_{15}\cup L_{16}, && &g_2^{(1)}=\infty \iff Y\in L_{13}\cup L_{14},\\
    &g_3^{(1)}=0 \iff Y\in L_{23}\cup L_{24}, && &g_3^{(1)}=\infty \iff Y\in L_{21}\cup L_{22}.
\end{align}
\end{subequations}
Since, by construction, the open set $U_{\mathcal{Y}}$ is disjoint with the eight lines on the right-hand sides of the above equivalences, $g_2^{(1)}$ and $g_3^{(1)}$ are finite and nonzero on $U_{\mathcal{Y}}$. Consequently, by the cubic relation \eqref{eq:tyurincubic}, $g_1^{(1)}$ is also finite and nonzero on $U_{\mathcal{Y}}$. Furthermore, observe that $g_2^{(1)}$ and $g_3^{(1)}$ cannot attain the value $1$ in the affine part of the Segre surface, since this requires $y_2=\infty$ or $y_3=\infty$. All in all, it follows that
\begin{equation*}
    g_\mathcal{Y}(U_{\mathcal{Y}})\subseteq U_\mathcal{G}.
\end{equation*}
Now, $g_\mathcal{Y}$ has a rational inverse on $\overline{\mathcal{G}}$, given by
\begin{align*}
    y_2=&\frac{(\upsilon_t\upsilon_\infty+\upsilon_t^{-1}\upsilon_\infty^{-1})g_2^{(1)}-(\upsilon_t\upsilon_\infty^{-1}+\upsilon_t^{-1}\upsilon_\infty)}{g_2^{(1)}-1},\\
    y_3=&\frac{(\upsilon_1\upsilon_\infty+\upsilon_1^{-1}\upsilon_\infty^{-1})g_3^{(1)}-(\upsilon_1\upsilon_\infty^{-1}+\upsilon_1^{-1}\upsilon_\infty)}{g_3^{(1)}-1},\\
    y_4=&y_2^{(1)}y_3^{(1)},
\end{align*}
and
\begin{align*}
        y_1=&(\upsilon_t-\upsilon_t^{-1})(\upsilon_1-\upsilon_1^{-1})(\upsilon_\infty-\upsilon_\infty^{-1})\frac{g_1^{(1)}g_2^{(1)}}{(g_2^{(1)}-1)(g_3^{(1)}-1)}\\
    &+\upsilon_\infty(\upsilon_0+\upsilon_0^{-1})-\upsilon_\infty(\upsilon_\infty-\upsilon_\infty^{-1})\frac{(\upsilon_t g_2^{(1)} -\upsilon_t^{-1})(\upsilon_1 g_3^{(1)}-\upsilon_1^{-1})}{(g_2^{(1)}-1)(g_3^{(1)}-1)}.
\end{align*}
Note that this mapping is clearly analytic on $U_\mathcal{G}$ and the corresponding image lies in $U_\mathcal{Y}$. We conclude that $g_\mathcal{Y}$ maps $U_\mathcal{Y}$ biholomorphically onto $U_\mathcal{G}$ and the lemma follows.
\end{proof}

We now have all the ingredients to prove Theorem \ref{thm:isomorphism}.
\begin{proof}[Proof of Theorem \ref{thm:isomorphism}]
Consider the bi-rational mapping
\begin{equation}\label{eq:bi-taional}
   \begin{array}{cccc}   \Phi:&\overline{\mathcal{Y}}&\dashrightarrow &\overline{\mathcal{Z}}_1\\
    & Y&\mapsto& Z\\
     \end{array}
\end{equation} 
obtained by composing $g_\mathcal{Y}$, defined in equation \eqref{eq:gydef}, with the rational inverse of $g_\mathcal{Z}$, given in equation \eqref{eq:gzinverse}. Due to Lemma \ref{lem:denseopenbiholo}, $\Phi$ maps the dense open subset $U_\mathcal{Y}\subseteq \overline{\mathcal{Y}}$ biholomorphically onto the dense open subset
$U_\mathcal{Z}\subseteq \overline{\mathcal{Z}}_1$. We proceed to check that $\Phi$ is a biholomorphism between the entire domain and co-domain. 

Now, the complements $\overline{\mathcal{Y}}\setminus U_\mathcal{Y}$ and $\overline{\mathcal{Z}}_1\setminus U_\mathcal{Z}$ each consist of $8$ lines and $2$ conics.
We first consider the conics on either side. These conics lie in the hyperplane sections $\overline{\mathcal{Y}}\setminus \mathcal{Y}$ and $\overline{\mathcal{Z}}_1\setminus \mathcal{Z}_1$ of the domain and co-domain respectively.
Recall that $\overline{\mathcal{Y}}\setminus \mathcal{Y}$ factors into the two conics described in equations \eqref{eq:conics}. On the first conic, equation \eqref{eq:conic1}, we have $Y_0=Y_2=0$ and the components of $g_\mathcal{Y}$ read
\begin{align*}
g_1&=\upsilon_\infty\frac{Y_4+(1+\upsilon_1\upsilon_\infty^{-1})Y_1+(\upsilon_t\upsilon_1-\upsilon_t^{-1}\upsilon_\infty^{-1})Y_3}{Y_4+(1+\upsilon_1\hspace{0.4mm}\upsilon_\infty\hspace{0.4mm})Y_1+(\upsilon_t\upsilon_1-\upsilon_t^{-1}\hspace{0.4mm}\upsilon_\infty\hspace{0.4mm})Y_3},\\
g_2&=\frac{Y_4-(\upsilon_t\upsilon_\infty^{-1}+\upsilon_t^{-1}\hspace{0.4mm}\upsilon_\infty\hspace{0.4mm})Y_3}{Y_4-(\upsilon_t\hspace{0.4mm}\upsilon_\infty\hspace{0.4mm}+\upsilon_t^{-1}\upsilon_\infty^{-1})Y_3},\\
g_3&=1.
\end{align*}
Applying the rational inverse of $g_\mathcal{Z}$ to this, we obtain
\begin{align*}
    Z_0&=0,\\
    Z_1/\eta_1^{(1)}&=\hspace{0.4mm} \upsilon_\infty \hspace{0.4mm} Y_4+(\upsilon_\infty-\upsilon_\infty^{-1})Y_1-(\upsilon_t+\upsilon_t^{-1})Y_3,\\
    Z_2/\eta_2^{(1)}&=\upsilon_\infty^{-1} Y_4+(\upsilon_\infty^{-1}-\upsilon_\infty)Y_1-(\upsilon_t+\upsilon_t^{-1})Y_3,\\
    Z_3/\eta_3^{(1)}&=Z_5/\eta_5^{(1)}=Y_4-(\upsilon_t\upsilon_\infty^{-1}+\upsilon_t^{-1}\upsilon_\infty)Y_3,\\
    Z_4/\eta_4^{(1)}&=Z_6/\eta_6^{(1)}=Y_4-(\upsilon_t \upsilon_\infty+\upsilon_t^{-1}\upsilon_\infty^{-1})Y_3. 
\end{align*}
This defines an isomorphism between the conics \eqref{eq:conic1} and \eqref{eq:conicIcontinuum} in $\overline{\mathcal{Y}}\setminus \mathcal{Y}$ and $\overline{\mathcal{Z}}_1\setminus \mathcal{Z}_1$ respectively. In particular, $\Phi$ maps the conic \eqref{eq:conic1} onto the conic \eqref{eq:conicIcontinuum}.

Similarly,  on the second conic, defined by equation \eqref{eq:conic2}, we have $Y_0=Y_3=0$ and the components of $g_\mathcal{Y}$ read
\begin{align*}
g_1&=\upsilon_\infty\frac{Y_4+(1+\upsilon_1\upsilon_\infty^{-1})Y_1+(1-\upsilon_1 \upsilon_\infty^{-1})Y_2}{Y_4+(1+\upsilon_1\hspace{0.4mm}\upsilon_\infty\hspace{0.4mm})Y_1+(1-\upsilon_1\hspace{0.4mm}\upsilon_\infty\hspace{0.4mm})Y_2},\\
g_2&=1,\\
g_3&=\frac{Y_4-(\upsilon_1\upsilon_\infty^{-1}+\upsilon_1^{-1}\hspace{0.4mm}\upsilon_\infty\hspace{0.4mm})Y_2}{Y_4-(\upsilon_1\hspace{0.4mm}\upsilon_\infty\hspace{0.4mm}+\upsilon_1^{-1}\upsilon_\infty^{-1})Y_2},
\end{align*}
Applying the rational inverse of $g_\mathcal{Z}$ to this, we obtain
\begin{align*}
    Z_0&=0,\\
    Z_1/\eta_1^{(1)}&=\hspace{0.4mm}\upsilon_\infty\hspace{0.4mm} Y_4+(\upsilon_\infty-\upsilon_\infty^{-1})Y_1-(\upsilon_1+\upsilon_1^{-1})Y_2,\\
    Z_2/\eta_2^{(1)}&=\upsilon_\infty^{-1} Y_4+(\upsilon_\infty^{-1}-\upsilon_\infty)Y_1-(\upsilon_1+\upsilon_1^{-1})Y_2,\\
    Z_3/\eta_3^{(1)}&=Z_6/\eta_6^{(1)}=Y_4-(\upsilon_1\upsilon_\infty^{-1}+\upsilon_1^{-1}\upsilon_\infty)Y_2,\\
    Z_4/\eta_4^{(1)}&=Z_5/\eta_5^{(1)}=Y_4-(\upsilon_1 \upsilon_\infty+\upsilon_1^{-1}\upsilon_\infty^{-1})Y_2. 
\end{align*}
This defines an isomorphism between the conics \eqref{eq:conic1} and \eqref{eq:conicIIcontinuum} in $\overline{\mathcal{Y}}\setminus \mathcal{Y}$ and $\overline{\mathcal{Z}}_1\setminus \mathcal{Z}_1$ respectively. In particular, $\Phi$ maps the conic \eqref{eq:conic2} onto the conic \eqref{eq:conicIIcontinuum}.

Next, we consider the remaining lines on either side.
By the explicit equation for $g_1^{(1)}$ on $\overline{\mathcal{Z}}_1$ in equations \eqref{eq:ztogdescription}, we have
\begin{align*}
    g_1^{(1)}=0&\iff Z_1=Z_6=0\iff Z\in L_{1,3,6}^{(1)}\cup L_{1,4,6}^{(1)},\text{ and}\\
        g_1^{(1)}=\infty&\iff Z_2=Z_5=0\iff Z\in L_{2,3,5}^{(1)}\cup L_{2,4,5}^{(1)}.
\end{align*}
Similarly, for $g_2^{(1)}$ and $g_3^{(1)}$, we find
\begin{align*}
    g_2^{(1)}=0&\iff Z_3=Z_5=0\iff Z\in L_{1,3,5}^{(1)}\cup L_{2,3,5}^{(1)},\text{ and}\\
        g_2^{(1)}=\infty&\iff Z_4=Z_6=0\iff Z\in L_{1,4,6}^{(1)}\cup L_{2,4,6}^{(1)},\text{ and}\\
    g_3^{(1)}=0&\iff Z_3=Z_6=0\iff Z\in L_{1,3,6}^{(1)}\cup L_{2,3,6}^{(1)},\text{ and}\\
        g_3^{(1)}=\infty&\iff Z_4=Z_5=0\iff Z\in L_{1,4,5}^{(1)}\cup L_{2,4,5}^{(1)}.
\end{align*}
Comparing these equivalences with those in equation \eqref{eq:equivalencesg2g3}, we find that
\begin{align*}
    &\Phi(L_{15}^\mathcal{Y}\cup L_{16}^\mathcal{Y})=L_{1,3,5}^{(1)}\cup L_{2,3,5}^{(1)}, &&&\Phi(L_{13}^\mathcal{Y}\cup L_{14}^\mathcal{Y})=L_{1,4,6}^{(1)}\cup L_{2,4,6}^{(1)}, \\
    &\Phi(L_{23}^\mathcal{Y}\cup L_{24}^\mathcal{Y})=L_{1,3,6}^{(1)}\cup L_{2,3,6}^{(1)}, &&
    &\Phi(L_{21}^\mathcal{Y}\cup L_{22}^\mathcal{Y})=L_{1,4,5}^{(1)}\cup L_{2,4,5}^{(1)}.
\end{align*}
To disentangle which line gets sent to which, see Remark \ref{remark:lines_correspondence}, and really check that $\Phi$ maps these lines isomorphically to one another, we require some workable formulas for the mapping $\Phi$. To obtain these, it is convenient to first compute some formulas for the affine $z$-variables in terms of the affine $x$-variables on $\mathcal{X}$.

Upon substituting the formulas for $g_k^{(1)}$, $1\leq k\leq 3$, in terms of the $x$-variables in Proposition \ref{prop:obtainingisomorphism}, into equation \eqref{eq:Z0defi} for $Z_0$, and simplifying modulo the Jimbo-Fricke cubic, we get
\begin{equation*}
    Z_0=\frac{\delta}{4\upsilon_0(x_2-(\upsilon_t\upsilon_\infty+\frac{1}{\upsilon_t\upsilon_\infty}))(x_3-(\upsilon_1\upsilon_\infty+\frac{1}{\upsilon_1\upsilon_\infty}))},
\end{equation*}
where the constant $\delta$ is as defined in Theorem \ref{thm:isomorphism}.

As a consequence, substituting the formulas for $g_k^{(1)}$, $1\leq k\leq 3$, in terms of the $x$-variables in Proposition \ref{prop:obtainingisomorphism}, into equations \eqref{eq:gzinverse} for $Z_k$, $1\leq k\leq 6$,  and
computing the quotient $z_k=Z_k/Z_0$, immediately gives us the formulas for the $z_k$, $3\leq k\leq 6$, given in Theorem \ref{thm:isomorphism}.  On the other hand, the resulting formulas for $z_1$ and $z_2$ are not (yet) in polynomial form, like in Theorem \ref{thm:isomorphism}. Regardless,
we obtain the following corresponding formulas with regards to the homogeneous $Y$-variables on $\overline{\mathcal{Y}}$ and $Z$-variables on $\overline{\mathcal{Z}}_1$, for the bi-rational mapping $\Phi$,
\begingroup
\allowdisplaybreaks
\begin{align}
Z_0=&Y_0,\label{YtoZequations}\\
Z_1=&\gamma^{-1}\Bigl[Y_4-\Bigl(\upsilon_1\upsilon_\infty+\frac{1}{\upsilon_1\upsilon_\infty}\Bigr) Y_2-\Bigl(\frac{\upsilon_t}{\upsilon_\infty}+\frac{\upsilon_\infty}{\upsilon_t}\Bigr)Y_3\nonumber\\
&+
\Bigl(\frac{\upsilon_t}{\upsilon_\infty}+\frac{\upsilon_\infty}{\upsilon_t}\Bigr)\Bigl(\upsilon_1\upsilon_\infty+\frac{1}{\upsilon_1\upsilon_\infty}\Bigr)Y_0
\Bigr]g_1^{(1)},\nonumber\\
    Z_2=&\gamma^{-1}\Bigl[
Y_4-\Bigl(\frac{\upsilon_1}{\upsilon_\infty}+\frac{\upsilon_\infty}{\upsilon_1}\Bigr)Y_2-\Bigl(\upsilon_t\upsilon_\infty+\frac{1}{\upsilon_t\upsilon_\infty}\Bigr)Y_3\nonumber\\
&+
\Bigl(\upsilon_t\upsilon_\infty+\frac{1}{\upsilon_t\upsilon_\infty}\Bigr)\Bigl(\frac{\upsilon_1}{\upsilon_\infty}+\frac{\upsilon_\infty}{\upsilon_1}\Bigr)Y_0\Bigr]1/g_1^{(1)},\nonumber\\
Z_3=&-\delta^{-1}\Bigl(\upsilon_0-\upsilon_t\upsilon_1\upsilon_\infty\Bigr)\Bigl(\upsilon_0-\frac{1}{\upsilon_t\upsilon_1\upsilon_\infty}\Bigr)\Bigl[
Y_4
-\Bigl(\frac{\upsilon_1}{\upsilon_\infty}+\frac{\upsilon_\infty}{\upsilon_1}\Bigr)Y_2\nonumber\\
&-\Bigl(\frac{\upsilon_t}{\upsilon_\infty}+\frac{\upsilon_\infty}{\upsilon_t}\Bigr)Y_3
+\Bigl(\frac{\upsilon_t}{\upsilon_\infty}+\frac{\upsilon_\infty}{\upsilon_t}\Bigr)
\Bigl(\frac{\upsilon_1}{\upsilon_\infty}+\frac{\upsilon_\infty}{\upsilon_1}\Bigr)
Y_0
\Bigr],\nonumber\\
Z_4=&-\delta^{-1}\Bigl(\upsilon_0-\frac{\upsilon_t\upsilon_1}{\upsilon_\infty}\Bigr)\Bigl(\upsilon_0-\frac{\upsilon_\infty}{\upsilon_t\upsilon_1}\Bigr)
\Bigl[Y_4-
\Bigl(\upsilon_1\upsilon_\infty+\frac{1}{\upsilon_1\upsilon_\infty}\Bigr)Y_2\nonumber\\
&-\Bigl(\upsilon_t\upsilon_\infty+\frac{1}{\upsilon_t\upsilon_\infty}\Bigr)Y_3
+\Bigl(\upsilon_t\upsilon_\infty+\frac{1}{\upsilon_t\upsilon_\infty}\Bigr)\Bigl(\upsilon_1\upsilon_\infty+\frac{1}{\upsilon_1\upsilon_\infty}\Bigr)Y_0
\Bigr],\nonumber\\
Z_5=&+\delta^{-1}\Bigl(\upsilon_0-\frac{\upsilon_t\upsilon_\infty}{\upsilon_1}\Bigr)\Bigl(\upsilon_0-\frac{\upsilon_1}{\upsilon_t \upsilon_\infty}\Bigr)
\Bigl[Y_4-\Bigl(\upsilon_1\upsilon_\infty+\frac{1}{\upsilon_1\upsilon_\infty}\Bigr) Y_2\nonumber\\
&-\Bigl(\frac{\upsilon_t}{\upsilon_\infty}+\frac{\upsilon_\infty}{\upsilon_t}\Bigr)Y_3+
\Bigl(\frac{\upsilon_t}{\upsilon_\infty}+\frac{\upsilon_\infty}{\upsilon_t}\Bigr)\Bigl(\upsilon_1\upsilon_\infty+\frac{1}{\upsilon_1\upsilon_\infty}\Bigr)Y_0
\Bigr],\nonumber\\
Z_6=&+\delta^{-1}\Bigl(\upsilon_0-\frac{\upsilon_1\upsilon_\infty}{\upsilon_t}\Bigr)\Bigl(\upsilon_0-\frac{\upsilon_t}{\upsilon_1\upsilon_\infty}\Bigr) \Bigl[
Y_4-\Bigl(\frac{\upsilon_1}{\upsilon_\infty}+\frac{\upsilon_\infty}{\upsilon_1}\Bigr)Y_2\nonumber\\
&-\Bigl(\upsilon_t\upsilon_\infty+\frac{1}{\upsilon_t\upsilon_\infty}\Bigr)Y_3+
\Bigl(\upsilon_t\upsilon_\infty+\frac{1}{\upsilon_t\upsilon_\infty}\Bigr)\Bigl(\frac{\upsilon_1}{\upsilon_\infty}+\frac{\upsilon_\infty}{\upsilon_1}\Bigr)Y_0\Bigr].\nonumber
\end{align}
\endgroup
where $\gamma$ is as defined in Theorem \ref{thm:isomorphism}.  

Recalling the explicit formulas for $g_1^{(1)}$ with respect to $Y$ in equations \eqref{eq:ytogdescription}, these equations, as written, are not projective linear in $Y$. Nonetheless, they are simple enough for explicit computations.
In particular, by a local analysis around each line on $\overline{\mathcal{Y}}$, we find that $\Phi$ maps it projective linearly onto
a corresponding line in $\overline{\mathcal{Z}}_1$, as detailed in Remark \ref{remark:lines_correspondence}.

For example, on the line $L_{15}^Y$, parametrised by
\begin{equation*}
    Y_1=\Bigl(\upsilon_0\nu_\infty+\frac{\nu_1}{\upsilon_t}\Bigr)Y_0-\frac{\upsilon_\infty}{\upsilon_t} Y_3,\quad Y_2=\Bigl(\frac{\upsilon_t}{\upsilon_\infty}-\frac{\upsilon_\infty}{\upsilon_t}\Bigr)Y_0,\quad Y_4=\Bigl(\frac{\upsilon_t}{\upsilon_\infty}-\frac{\upsilon_\infty}{\upsilon_t}\Bigr)Y_3,
\end{equation*}
with $[Y_0:Y_3]\in\mathbb{CP}^1$,
we find
\begin{align*}
Z_0=&Y_0,\\
Z_1=&Z_3=Z_5=0,\\
Z_2=&+\gamma^{-1}\Bigl[\Bigl(\frac{\upsilon_\infty}{\upsilon_t}-\frac{\upsilon_t}{\upsilon_\infty}\Bigr)Y_3-Y_0\Bigr((\upsilon_0+\upsilon_0^{-1})(\upsilon_\infty-\upsilon_\infty^{-1})-(\upsilon_t-\upsilon_t)(\upsilon_1+1/\upsilon_1)\Bigr)\Bigr],\\
Z_4=&-\delta^{-1}\Bigl(\upsilon_0-\frac{\upsilon_t\upsilon_1}{\upsilon_\infty}\Bigr)\Bigl(\upsilon_0-\frac{\upsilon_\infty}{\upsilon_t\upsilon_1}\Bigr)
(\upsilon_t-\upsilon_t^{-1})(\upsilon_\infty-\upsilon_\infty^{-1})\Bigl(Y_3-Y_0\Bigl(\upsilon_1\upsilon_\infty+\frac{1}{\upsilon_1\upsilon_\infty}\Bigr)\Bigr),\\
Z_6=&+\delta^{-1}\Bigl(\upsilon_0-\frac{\upsilon_1\upsilon_\infty}{\upsilon_t}\Bigr)\Bigl(\upsilon_0-\frac{\upsilon_t}{\upsilon_1\upsilon_\infty}\Bigr)(\upsilon_t-\upsilon_t^{-1})(\upsilon_\infty-\upsilon_\infty^{-1})\Bigl(Y_3-Y_0\Bigl(\frac{\upsilon_1}{\upsilon_\infty}+\frac{\upsilon_\infty}{\upsilon_1}\Bigr)\Bigr),
\end{align*}
so that $\Phi$ maps $L_{15}^Y$ projective linearly onto
 $L_{1,3,5}^{(1)}$.

All in all, we know that $\Phi$ maps the dense open subset $U_\mathcal{Y}$ biholomorphically onto $U_\mathcal{Z}$ and it maps each irreducible component of $\overline{\mathcal{Y}}\setminus U_\mathcal{Y}$ onto a corresponding irreducible component of $\overline{\mathcal{Z}}_1\setminus U_\mathcal{Z}$. To conclude that $\Phi$ is an isomorphism, we may proceed in two ways.

The first, is to check that $\Phi$ is locally biholomorphic around each point in the complement of $U_\mathcal{Z}$. From this it follows that  $\Phi$ is a locally biholomorpic bi-rational mapping and thus an isomorphism.

Alternatively, one can check that $\Phi$ respects the intersections between the different irreducible components in $\overline{\mathcal{Y}}\setminus U_\mathcal{Y}$ and $\overline{\mathcal{Z}}_1\setminus U_\mathcal{Z}$, as follows by comparing the Clebsch graphs in Figures \ref{fig:lines_intersection} and \ref{fig:intersectiongraph} and equations \eqref{eq:linescorrespondence}, noting in particular that all the ForestGreen coloured lines in Figure  \ref{fig:lines_intersection} are mapped to the lines $L_k^\mathcal{Y}$, $9\leq k\leq 16$, and all the purple coloured lines are mapped to the lines $L_k^\mathcal{Y}$, $17\leq k\leq 24$, see Remark \ref{remark:conics_correspondence}. This way, we can check that $\Phi$ is a bi-rational bijection between $\overline{\mathcal{Y}}$ and $\overline{\mathcal{Z}}_1$. Since $\overline{\mathcal{Y}}$ and $\overline{\mathcal{Z}}_1$ are smooth and thus normal and irreducible, it follows from Zariski's main theorem that $\Phi$ is an isomorphism.

Either way yields the required result. We now apply Lemma \ref{lem:segre_isomorpism}, which shows that
 $\Phi$ must extend to a unique (rank four) projective linear map between the ambient spaces $\mathbb{P}^4$ and $\mathbb{P}^6$. Furthermore, since
 $\Phi$ maps the hyperplane sections at infinity to one another, it follows that $\Phi$ also restricts to an affine equivalence $\Phi_{\mathcal{Y}}:=\Phi|_{\mathcal{Y}}$ between the affine Segre surfaces $\mathcal{Y}$ and $\mathcal{Z}_1$,
 \begin{equation*}
 \Phi_{\mathcal{Y}}:\mathcal{Y}\rightarrow \mathcal{Z}_1.
 \end{equation*}

 Furthermore, since the blow-down mapping $\pi$, defined in \eqref{eq:blowdownpi}, is an isomorphism between $\mathcal{X}$ and $\mathcal{Y}$ as affine varieties when restricted to $\mathcal{X}$, the map
 \begin{equation*}
     \Phi_{\mathcal{X}}=\Phi_{\mathcal{Y}}\circ \pi|_{\mathcal{X}},
 \end{equation*}
is an isomorphism of affine varieties,
\begin{equation*}
   \begin{array}{cccc}   \Phi_\mathcal{X}:&\mathcal{X}&\rightarrow &\mathcal{Z}_1.
     \end{array}
\end{equation*} 

What is left, is to derive the explicit formulas for $\Phi_\mathcal{X}$ in the theorem. Since $\Phi_{\mathcal{Y}}$ is an affine equivalence, we know that it can be written as an affine linear map as in equation \eqref{eq:affine_equivalence_pvi}. Therefore, $\Phi_{\mathcal{X}}$
can be written as
\begin{equation}\label{eq:phiXexplicit}
 z_k = \xi_{0k}+\xi_{1k} \;x_1+\xi_{2k}\; x_2+\xi_{3k} \;x_3+\xi_{4k}\; x_2\;x_3,
\qquad 1\leq k\leq 6,   
\end{equation}
for some coefficients $\xi_{jk}$, $0\leq j\leq 4$, $1\leq k\leq 6$.

The formulas for $Z_k$, $k=0,3,4,5,6$, in equations \eqref{YtoZequations}, immediately give us the formulas for $z_k$, $3\leq k\leq 6$, in terms of the affine $x$-variables in the theorem, using
\begin{equation*}
    Y_1=x_1 Y_0,\quad Y_2=x_2 Y_0,\quad Y_3=x_3 Y_0,\quad Y_4=x_2x_3Y_0.
\end{equation*}
Note, in particular, that these formulas are of the form \eqref{eq:phiXexplicit}.

Finally, we compute the coefficients in \eqref{eq:phiXexplicit} for $k=1,2$ to obtain the formulas for $z_1$ and $z_2$ in Theorem \ref{thm:isomorphism}. We do this by computing the $z_1$ and $z_2$ entries of the images under $\Phi_{\mathcal{X}}$ of five explicit points on the Jimbo-Fricke cubic,
\begin{align*}
p_1:& & x_1&=\upsilon_0\upsilon_\infty+\upsilon_0^{-1}\upsilon_\infty^{-1}, & x_2&=2, & x_3&=-2\upsilon_0\upsilon_\infty+\upsilon_0\nu_t+\upsilon_\infty\nu_1,\\
p_2:& & x_1&=\upsilon_0\upsilon_\infty^{-1}+\upsilon_0^{-1}\upsilon_\infty, & x_2&=2, & x_3&=-2\upsilon_0\upsilon_\infty^{-1}+\upsilon_0\nu_t+\upsilon_\infty^{-1}\nu_1,\\
p_3:& &  x_2&=\upsilon_0\upsilon_1+\upsilon_0^{-1}\upsilon_1^{-1}, & x_1&=2, & x_3&=-2\upsilon_0\upsilon_1+\upsilon_0\nu_t+\upsilon_1\nu_\infty,\\
p_4:& &  x_3&=\upsilon_0\upsilon_t^{-1}+\upsilon_0^{-1}\upsilon_t, & x_1&=2, & x_2&=-2\upsilon_0^{-1}\upsilon_t+\upsilon_0^{-1}\nu_1+\upsilon_t\nu_\infty,\\
p_5:& &  x_3&=\upsilon_0\upsilon_t+\upsilon_0^{-1}\upsilon_t^{-1}, & x_1&=2, & x_2&=-2\upsilon_0 \upsilon_t+\upsilon_0\nu_1+\upsilon_t\nu_\infty.
\end{align*}
The corresponding values of $z_1$ and $z_2$ under $\Phi_{\mathcal{X}}$ of these points are as follows,
\begin{align*}
z_1(p_1)&=-\frac{(2\upsilon_0\upsilon_\infty-\upsilon_0\nu_t-\upsilon_1(\upsilon_\infty-\upsilon_\infty^{-1}))(2\upsilon_0\upsilon_\infty-\upsilon_0\nu_t-\upsilon_1^{-1}(\upsilon_\infty-\upsilon_\infty^{-1}))}{\upsilon_0\gamma},\\  
z_2(p_1)&=-\frac{\upsilon_0(\upsilon_t-\upsilon_\infty)^2(\upsilon_t\upsilon_\infty-1)^2}{ \upsilon_t^2\upsilon_\infty^2\gamma},\\    
z_1(p_2)&=-\frac{(2\upsilon_0\upsilon_\infty-\upsilon_1\nu_t-\upsilon_0(\upsilon_\infty-\upsilon_\infty^{-1}))(2\upsilon_1\upsilon_\infty^{-1}-\upsilon_1^{-1}\nu_t-\upsilon_0(\upsilon_\infty-\upsilon_\infty^{-1}))}{\upsilon_0\gamma},\\  
z_2(p_2)&=-\frac{(\upsilon_t-\upsilon_\infty)^2(\upsilon_t\upsilon_\infty-1)^2}{\upsilon_0 \upsilon_t^2\upsilon_\infty^2\gamma},\\    
z_1(p_3)&=-\frac{(\upsilon_0\upsilon_t\upsilon_1\upsilon_\infty-1)(\upsilon_0\upsilon_1\upsilon_\infty-\upsilon_t)(2\upsilon_0\upsilon_1-\upsilon_0\nu_t-\upsilon_\infty(\upsilon_1-\upsilon_1^{-1}))}{\upsilon_0\upsilon_t\upsilon_1\upsilon_\infty \gamma},\\
z_2(p_3)&=-\frac{(\upsilon_0\upsilon_t\upsilon_1-\upsilon_\infty)(\upsilon_0\upsilon_1-\upsilon_t\upsilon_\infty)(2\upsilon_0\upsilon_1-\upsilon_0\nu_t-\upsilon_\infty^{-1}(\upsilon_1-\upsilon_1^{-1}))}{\upsilon_0\upsilon_t\upsilon_1\upsilon_\infty \gamma},\\
z_1(p_4)&=-\frac{(\upsilon_0\upsilon_1-\upsilon_t\upsilon_\infty)(\upsilon_0-\upsilon_t\upsilon_1\upsilon_\infty)(2\upsilon_0^{-1}\upsilon_t-\upsilon_0^{-1}\nu_1-\upsilon_\infty(\upsilon_t-\upsilon_t^{-1}))}{\upsilon_0\upsilon_t\upsilon_1\upsilon_\infty \gamma},\\
z_2(p_4)&=-\frac{(\upsilon_0\upsilon_1\upsilon_\infty-\upsilon_t)(\upsilon_0\upsilon_\infty-\upsilon_t\upsilon_1)(2\upsilon_0^{-1}\upsilon_t-\upsilon_0^{-1}\nu_1-\upsilon_\infty^{-1}(\upsilon_t-\upsilon_t^{-1}))}{\upsilon_0\upsilon_t\upsilon_1\upsilon_\infty \gamma},\\
z_1(p_5)&=-\frac{(\upsilon_0\upsilon_t\upsilon_1\upsilon_\infty-1)(\upsilon_0\upsilon_t\upsilon_\infty-\upsilon_1)(2\upsilon_0\upsilon_t-\upsilon_0\nu_1-\upsilon_\infty(\upsilon_t-\upsilon_t^{-1}))}{\upsilon_0\upsilon_t\upsilon_1\upsilon_\infty \gamma},\\
z_2(p_5)&=-\frac{(\upsilon_0\upsilon_t\upsilon_1-\upsilon_\infty)(\upsilon_0\upsilon_t-\upsilon_1\upsilon_\infty)(2\upsilon_0 \upsilon_t-\upsilon_0\nu_1-\upsilon_\infty^{-1}(\upsilon_t-\upsilon_t^{-1}))}{\upsilon_0\upsilon_t\upsilon_1\upsilon_\infty \gamma}.
\end{align*}

 By plugging these values into equations \eqref{eq:phiXexplicit}, we get two sets of five equations among five unknown coefficients, which each have a unique solution. The result yields the two explicit equations for $z_1$ and $z_2$ in Theorem \ref{thm:isomorphism}, thus concluding the proof the theorem.
\end{proof}

\begin{remark}
    We note that the involutive automorphism of the Segre surface $\overline{\mathcal{Z}}_1$, given by swapping $z_1$ and $z_2$,
    corresponds to the following automorphism of $\overline{\mathcal{Y}}$,
    \begin{equation*}
       y_1\mapsto -\omega_1-y_1-y_4,\qquad y_{2,3,4}\mapsto y_{2,3,4},
    \end{equation*}
    via the isomorphism in Theorem \ref{thm:isomorphism}. In turn, via the blow-down of the Jimbo-Fricke cubic, this defines an automorphism of the affine cubic $\mathcal{X}$,
    \begin{equation*}
       x_1\mapsto -\omega_1-x_1-x_2x_3,\qquad x_{2,3}\mapsto x_{2,3},
    \end{equation*}
    which is one of the generators of the extended modular group action on the Jimbo-Fricke cubic \cite{iwasaki}.

    Regarding the three Tyurin ratios $g^{(q)}=(g_1^{(q)},g_2^{(q)},g_3^{(q)})$,  the relation \eqref{eq:tyurincubic} when $q=1$, generalises to
    \begin{equation}\label{eq:tyurincubicgen}
\eta_{1}^{(q)}g_2^{(q)} g_1^{(q)}+\eta_{2}^{(q)}g_3^{(q)}/g_1^{(q)}+\eta_{3}^{(q)}g_2^{(q)}g_3^{(q)}+\eta_{4}^{(q)}+\eta_{5}^{(q)}g_2^{(q)}+\eta_{6}^{(q)}g_3^{(q)}=0.
\end{equation}
    This equation is quadratic in $g_1^{(q)}$ and the involutive automorphism above is equivalent to sending $g_1^{(q)}$ to its other root of this quadratic, with $q=1$, explicitly
    \begin{equation*}
       g_1^{(q)}\mapsto -1/\eta_1^{(q)}(\eta_3^{(q)} g_3^{(q)}+\eta_4^{(q)}/g_2^{(q)}+\eta_5^{(q)} +\eta_6^{(q)}g_3^{(q)}/g_2^{(q)}),\qquad g_{2,3}^{(q)}\mapsto g_{2,3}^{(q)}.
    \end{equation*}
The last formula gives a natural generalisation of the above automorphism to $\mathcal{Z}_q$ for $q\neq 1$, though its action on the $z$-variables seems quite involved.
\end{remark}

\section{The other Painlev\'e differential equations}\label{se:allpsegres}
\tcm{The confluence scheme of the Painlev\'e differential equations from the viewpoint of isomonodromic deformations \cite{ohyamaokumura}},
$$
\xymatrix
{&& \tcb{P_{III}}\ar[dr] \ar[r]&\tcr{P_{III}^{D_7}}\ar[dr]\ar[r]&\tcr{P_{III}^{D_8}}\\
\tcb{P_{VI}}\ar[r]&\tcb{P_{V}}\ar[r]\ar[ur]\ar[dr]
&\tcr{P_{V}^{deg}}\ar[dr]\ar[ur]& \tcb{\Ptwojm}\ar[r]&\tcr{P_I}\\
 &&\tcb{P_{IV}}\ar[ur]\ar[r] &\tcr{P_{II}^{FN}}\ar[ur]\\
}
$$
\tcm{corresponds to appropriate limits on the associated monodromy manifolds \cite{CMR}. }
In this section, we show that, in all un-ramified cases (in blue), the confluence scheme of the Painlev\'e monodromy manifolds can be carried through to the affine transformation constructed in Theorem \ref{thm:isomorphism}, therefore producing isomorphisms between each $\mathcal Y$--Segre, i.e. those ones resulting from blowing down the monodromy manifold at a line at infinity, and the $\mathcal Z$--Segre obtained by confluencing $\mathcal Z_1$. 

For the ramified cases (in red), the confluence either produces a
reducible Segre surface or a family which does not have the correct number of free parameters. This is not surprising because the confluences to ramified and non-ramified Painlev\'e equations are deeply different in geometric as well as analytic terms. In Section
6, we perform an in depth study of the singularity structure of all cubic surfaces with a
triangle of lines at infinity, of their blow downs to the corresponding $\mathcal Y$--Segre surfaces
and the expected singularity structure of the $\mathcal Z$--Segre ones. This will allow us to build
the isomorphic $\mathcal Z$--Segre for all ramified cases.


This section is organised as follows: In Subsection \ref{suse:CMR} we summarise the confluence of monodromy manifolds obtained in \cite{CMR} for the non-ramified cases and deduce the confluence on the corresponding $\mathcal{Y}$--Segre surfaces. Then, in \ref{suse:confl-p}, we apply the confluence of monodromy manifolds  to the affine transformation constructed in Theorem \ref{thm:isomorphism} and provide the Segre surfaces for all non-ramified cases.

\subsection{Confluence of monodromy manifolds}\label{suse:CMR} 

Following \cite{SvdP} and \cite{CMR} the monodromy manifolds of all Painlev\'e equations are given by 
\begin{equation}\label{eq:mon-mf}
\mathcal X^{(d)}:= {{\rm Spec}}(\mathbb C[x_1,x_2,x_3]\slash\langle\phi^{(d)}=0\rangle)
\end{equation}
where the polynomial $\phi^{(d)}$ has always the same form in the variables $x_1,x_2,x_3$, but different coefficients for different Painlev\'e equations:
\begin{equation}\label{eq:phi}
\phi^{(d)}=x_1 x_2 x_3 + \epsilon_1^{(d)} x_1^2+ \epsilon_2^{(d)} x_2^2+ \epsilon_3^{(d)} x_3^2 + \omega_1^{(d)} x_1  + \omega_2^{(d)} x_2 + \omega_3^{(d)} x_3+
\omega_4^{(d)}=0.
\end{equation}
Here $d$ is an index running on the list of all the Painlev\'e auxiliary linear problems $\Psix, \Pfive,\Pfive^{\text{deg}},\Pfour$, 
$\Pthree^{D_6},\Pthree^{D_7}$, 
$\Pthree^{D_8},\Ptwojm ,\Ptwofn,\Pone$,  the parameters $ \epsilon^{(d)}_{i}$ may take value $0$ or $1$ according to the chosen $d$, and the parameters $\omega^{(d)}_{i}$, $i=1,2,3,4$ are related to the Painlev\'e equations constants as given in section 2 of \cite{CMR}.
For convenience, we summarise the monodromy manifolds of the 
non-ramified Painlev\'e equations in Table \ref{tb:monmfd}. 
\begin{table}[H]
    \centering
\begin{tabular}{|c || c || c|} 
 \hline
     &  monodromy manifold & parameters \\
 \hline
$\tcb{\Psix}$ & $ x_1 x_2 x_3+ \sum_{k=1}^3 (x_k^2+\omega_k x_k) +\omega_4=0 $ & $\begin{array}{cc}
    \omega_1=&-(\nu_0\nu_\infty+\nu_t \nu_1),\\
    \omega_2=&-(\nu_0\nu_1+\nu_t\nu_\infty),\\
    \omega_3=&-(\nu_0\nu_t+\nu_1\nu_\infty),\\  \omega_4=&\nu_0^2+\nu_t^2+\nu_1^2+\nu_\infty^2+\\
   & \nu_0\nu_t\nu_1\nu_\infty-4.\\ \end{array}$
\\
 \hline
 $\tcb{\Pfive}$ & $x_1 x_2 x_3+ x_1^2 +x_2^2 + \sum_{k=1}^3 \omega_k x_k +\omega_4=0$  &
 $\begin{array}{cc}
    \omega_1=&-(\nu_0\upsilon_\infty+\nu_t ),\\
    \omega_2=&-(\nu_0 +\nu_t\upsilon_\infty),\\
    \omega_3=&-(\nu_0\nu_t+\upsilon_\infty),\\  \omega_4=&1+\upsilon_\infty^2+ \nu_0\nu_t\upsilon_\infty.\\ \end{array}$\\
 \hline
 $\tcr{\Pfive^\text{deg}}$ &  $x_1 x_2 x_3+ x_1^2 +x_2^2 + \sum_{k=1}^3 \omega_k x_k +\omega_4=0$
  & $\begin{array}{cc}
    \omega_1=&- \nu_0\\
    \omega_2=&-\nu_t,\\
    \omega_3=&0,\\  \omega_4=& 1.\\ \end{array}$\\ \hline
 $\tcb{\Pfour}$  & $x_1 x_2 x_3+ x_1^2 + \sum_{k=1}^3 \omega_k x_k +\omega_4=0$  & $\begin{array}{cc}
    \omega_1=&-(\nu_0\upsilon_\infty+1 ),\\
    \omega_2=&\omega_3=-\upsilon_\infty,\\  \omega_4=&\upsilon_\infty^2+ \nu_0\upsilon_\infty.\\ \end{array}$ \\
\hline
 $\tcb{\Pthree^{D_6}}$ &  $x_1 x_2 x_3+ x_1^2 +x_2^2 + \sum_{k=1}^3 \omega_k x_k +\omega_4=0$
  & $\begin{array}{cc}
    \omega_1=&- \upsilon_0\upsilon_\infty-1\\
    \omega_2=&-(\upsilon_0 +\upsilon_\infty),\\
    \omega_3=&0,\\  \omega_4=& \upsilon_0\upsilon_\infty.\\ \end{array}$\\
 \hline
 $\tcr{\Pthree^{D_7}}$ &  $x_1 x_2 x_3+ x_1^2 +x_2^2 + \sum_{k=1}^3 \omega_k x_k +\omega_4=0$
  & $\begin{array}{cc}
    \omega_1=&-1\\
    \omega_2=&-\upsilon_\infty,\\
    \omega_3=&0,\\  \omega_4=&0.\\ \end{array}$\\
 \hline
 $\tcr{\Pthree^{D_8}}$ &  $x_1 x_2 x_3+ x_1^2 +x_2^2 + \sum_{k=1}^3 \omega_k x_k +\omega_4=0$
  & $\begin{array}{cc}
    \omega_1=
    \omega_3=
    \omega_4=&0,\\  \omega_2=&-1.\\ \end{array}$\\
 \hline
$\tcr{\Ptwofn}$ & $x_1 x_2 x_3+ x_1^2+\sum_{k=1}^3 \omega_k x_k +\omega_4=0$  &$\begin{array}{cc}
    \omega_1=&-(\upsilon_0+\upsilon_0^{-1}),\\
    \omega_2=-1&\omega_3=0,\\  \omega_4=&1.\\ \end{array}$\\\hline
$\tcb{\Ptwojm}$ & $x_1 x_2 x_3+ \sum_{k=1}^3 \omega_k x_k +\omega_4=0$  &$\begin{array}{cc}
    \omega_1=&
    \omega_2=
    \omega_3=-\upsilon_\infty,\\  \omega_4=& \upsilon_\infty(1+\upsilon_\infty).\\ \end{array}$\\
 \hline
 $\tcr{\Pone}$ & $x_1 x_2 x_3+ \sum_{k=1}^3 \omega_k x_k +\omega_4=0$  &$\begin{array}{cc}
    \omega_1=-1,&
    \omega_2=-1,\\
    \omega_3=0, & \omega_4= 1.\\ \end{array}$\\
 \hline
\end{tabular}
\caption{The monodromy manifolds of the  Painlev\'e differential equations - red denotes ramified cases and blue un-ramified. The parameters $\nu_i$, $\upsilon_i$, $i=0,t,1,\infty$ are defined in \eqref{eq:pvi-om1}.}\label{tb:monmfd}
\end{table}

\tcm{\begin{remark}
    Notice that in the online version of \cite{SvdP}, the sign in front of $x_2^2$ in the $\Pthree^{D_8}$ cubic was corrected to
   a minus. However, this sign change can be easily produced by the simple rescaling $x_1\to i x_1$, $x_2\to -x_2$, $x_3\to -i x_3$. We prefer to stick to the cubic with the plus sign as it is the one which naturally appears under confluence \cite{CMR}.
\end{remark}}

The confluence procedure always involves a re-scaling of two variables $x_i, x_j$, $i\neq j$ and of some of the parameters in $\epsilon$ followed by taking the limit for $\epsilon\to 0$. We blow down each monodromy manifold to a $\mathcal{Y}$--Segre surface always in the same way, namely we set
$$
y_1:=x_1,\quad y_2:=x_2,\quad y_3:=x_3,\quad y_4:= x_2 x_3.
$$
We summarise these confluences in table \ref{tb:confluence}.

\def\mystrut{\rule{0pt}{1.4\normalbaselineskip}}
\begin{table}[H]
    \centering
\begin{tabular}{|c || c || c|| c|} 
 \hline
   Confluence & $\begin{array}{c}
     \hbox{Re-scaling on}  \\
          x_1,x_2,x_3\\
   \end{array}$
    & Re-scaling on $\mathcal Y$--Segre & $\begin{array}{c}
     \hbox{Re-scaling on}  \\
          \hbox{parameters}\\
   \end{array}$\\ 
 \hline
 $\Psix\mapsto\Pfive$& $x_1\mapsto\frac{x_1}{\epsilon}, x_2\mapsto\frac{x_2}{\epsilon}$, \mystrut 
 &$y_1\mapsto\frac{y_1}{\epsilon}, y_2\mapsto\frac{y_2}{\epsilon}, y_4\mapsto\frac{y_4}{\epsilon}$& $\nu_\infty\mapsto\frac{\upsilon_\infty}{\epsilon}$, 
 $\nu_1\mapsto\frac{1}{\epsilon}$\\[2mm] 
 \hline
$\Pfive\mapsto\Pfour$& $x_1\mapsto\frac{x_1}{\epsilon}, x_3\mapsto\frac{x_3}{\epsilon},$ \mystrut 
 &$y_1\mapsto\frac{y_1}{\epsilon}, y_3\mapsto\frac{y_3}{\epsilon}, y_4\mapsto\frac{y_4}{\epsilon}$&$ \upsilon_\infty\mapsto\frac{\upsilon_\infty}{\epsilon}$, 
 $\nu_t\mapsto\frac{1}{\epsilon}$\\[2mm] 
\hline
$\Pfive\mapsto\Pfive^{\text{deg}}$& $x_1\mapsto\frac{x_1}{\epsilon}, x_2\mapsto\frac{x_2}{\epsilon},$\mystrut 
 &$y_1\mapsto\frac{y_1}{\epsilon}, y_2\mapsto\frac{y_2}{\epsilon}, y_4\mapsto\frac{y_4}{\epsilon}$ &$ \upsilon_\infty\mapsto\frac{1}{\epsilon}$, $ \upsilon_1\mapsto\frac{\upsilon_1}{\epsilon}$\\[2mm] \hline
$\Pfive\mapsto\Pthree^{D_6}$& $x_1\mapsto\frac{x_1}{\epsilon}, x_2\mapsto\frac{x_2}{\epsilon},$\mystrut 
 &$y_1\mapsto\frac{y_1}{\epsilon}, y_2\mapsto\frac{y_2}{\epsilon}, y_4\mapsto\frac{y_4}{\epsilon}$ &$ \nu_0\mapsto\frac{\upsilon_0}{\epsilon}$, 
 $\nu_t\mapsto\frac{1}{\epsilon}$\\[2mm] 
 \hline
 $\Pthree^{D_6}\mapsto \Pthree^{D_7}$& $x_1\mapsto\frac{x_1}{\epsilon}, x_2\mapsto\frac{x_2}{\epsilon},$\mystrut 
 &$y_1\mapsto\frac{y_1}{\epsilon}, y_2\mapsto\frac{y_2}{\epsilon}, y_4\mapsto\frac{y_4}{\epsilon}$ &$ \upsilon_0\mapsto\frac{1}{\epsilon}$, 
 $\upsilon_\infty\mapsto {\epsilon}\upsilon_\infty $\\[2mm] 
 \hline
$\Pfour\mapsto\Ptwojm$& $x_2\mapsto\frac{x_2}{\epsilon}, x_3\mapsto\frac{x_3}{\epsilon},$\mystrut 
 &$y_2\mapsto\frac{y_2}{\epsilon}, y_3\mapsto\frac{y_3}{\epsilon}, y_4\mapsto\frac{y_4}{\epsilon^2}$&$ \nu_0\mapsto\frac{1}{\epsilon}$, 
 $\upsilon_\infty\mapsto\frac{\upsilon_\infty}{\epsilon}$\\[2mm] 
 \hline
\end{tabular}
\caption{Re-scalings giving rise to the confluence of the non-ramified Painlev\'e monodromy manifolds and $\mathcal Y$--Segre surfaces.}\label{tb:confluence}
\end{table}

\subsection{Confluence of the affine transformation between $\mathcal Y$--Segre and $\mathcal Z$--Segre}\label{suse:confl-p}

The affine transformation between $\mathcal Y$--Segre and $\mathcal Z$--Segre has always the same form:
$$
z_k = \xi_{0k}+\xi_{1k} y_1+\xi_{2k} y_2+\xi_{3k} y_3+\xi_{4k} y_4,
\quad k=1,\dots,6
$$
where the coefficients $\xi_{jk}$, $k=1,\dots 6$, $j=0,\dots,4$, depend on $\upsilon_i$, $i=0,t,1,\infty$ as in \eqref{eq:pvi-om1}.

\subsubsection{Segre surface of $\Pfive$}
The confluence limit form $\Psix$ to $\Pfive$, we see that the re-scaling of the parameters in the last column of Table \ref{tb:confluence} produces the following re-scalings on the coefficients $\xi_{kj}$:
$$
\xi_{0k}\mapsto \xi_{0k},\quad \xi_{1k}\mapsto\frac{\xi_{1k}}{\epsilon},\quad \xi_{2k}\mapsto\frac{\xi_{2k}}{\epsilon},\quad\xi_{3k}\mapsto \xi_{3k},\quad \xi_{4k}\mapsto\frac{\xi_{4k}}{\epsilon},
$$
with the limiting $\xi_{jk}$ not all zero for some fixed $k$.
This means that all $z_k$ remain finite because the re-scalings for the coefficients $\xi_{jk}$ are compensated by the re-scaling in the $y_1,\dots,y_4$ variables. Therefore the $\mathcal Z$--Segre equations maintain the same form.
The parameters are re-scaled as follows
$$
\rho_3\to \epsilon\rho_3, \quad
\rho_4\to\rho_4, \quad \rho_4\to\rho_4,\quad
\rho_6\to \epsilon\rho_6, \quad
{\rho_5\rho_6}{\rho_3}\to \lambda_2,\quad
\lambda_1\to \lambda_1,
$$
therefore the parameters $\rho_3,\rho_6$ become $0$.
\begin{subequations}\label{p5:seg}
\begin{align}
&  {z}_{1} +   {z}_{2} + {z}_{3} + {z}_{4} + {z}_{5} +  {z}_{6}=0,\\
 &  {z}_{4} + \rho_{5} {z}_{5}  - 1=0,\\
&{z}_{3} {z}_{4} -\lambda_{1} {z}_{1} {z}_{2} =0,\\
&z_5 z_6 - \lambda_2  z_1  z_2=0.
\end{align}
\end{subequations}
Explicitly, setting
$$
\gamma=(\upsilon_0-1)\left(\frac{1}{\upsilon_0}\right)\upsilon_\infty^2,
\quad\delta= (\upsilon_0-1)^2\nu_t \upsilon_\infty^2,
$$
the formulae relating the $\mathcal{Z}$-Segre to the monodromy manifold in the $\Pfive$ case are:
\begin{align*}
    z_1= \gamma^{-1}\left(x_2(\upsilon_\infty x_3-1)+\upsilon_\infty
    \left(x_1-\upsilon_\infty \nu_0\right)
    \right),\\
     z_2= -\gamma^{-1}\left(x_2-\nu_t \upsilon_\infty)+\upsilon_\infty
    x_1
    \right),\\
    z_3=\delta^{-1} \upsilon_0 (\upsilon_t x_2-\upsilon
    _\infty)(\upsilon_\infty x_3-\upsilon_\infty^2-1),\\
    z_4=\delta^{-1}\upsilon_t^{-1} (\upsilon_\infty \upsilon_0 - \upsilon_t)(\upsilon_t \upsilon_0 - \upsilon_\infty)(x_2 -\upsilon_t \upsilon_\infty),\\
    z_5=-\delta^{-1}\upsilon_t^{-2}( \upsilon_\infty  \upsilon_0\upsilon_t-1)(\upsilon_0-\upsilon_t\upsilon_\infty)(\upsilon_t x_2-\upsilon_\infty ),\\
    z_6=-\delta^{-1}\upsilon_t^{-1}\upsilon_0(\upsilon_\infty x_3-\upsilon_\infty^2-1)(x_2-\upsilon_\infty\upsilon_t).
\end{align*}

\subsubsection{Segre surface of $\Pfour$}
The confluence limit form $\Pfive$ to $\Pfour$  behaves in a similar way to the one from $\Psix$ to $\Pfive$, so we omit the discussion.
\begin{subequations}\label{p4:seg}
\begin{align}
&{z}_{1} +  {z}_{2} + {z}_{3} + {z}_{4} + {z}_{5} +  {z}_{6}=0,\\
 &  {z}_{4}  - 1=0,\\
&{z}_{3} {z}_{4} - \lambda_{1}{z}_{1} {z}_{2}=0,\\
&z_5 z_6 - \lambda_2  z_1  z_2=0.
\end{align}
\end{subequations}
Explicitly, 
the formulae relating the $\mathcal{Z}$-Segre to the monodromy manifold in the $\Pfour$ case are:
\begin{align*}
    z_1=\frac{1+\upsilon_0^2-\frac{\upsilon_0}{\upsilon_\infty} x_1-\frac{\upsilon_\infty}{\upsilon_1} x_2 x_3}{(\upsilon_0-1)^2},\quad 
    z_2=\frac{\upsilon_0(x_1-1)}{\upsilon_\infty(\upsilon_0-1)^2},\\
    z_3=\frac{\upsilon_0(x_2-\upsilon_\infty)(x_3-\upsilon_\infty)}{\upsilon_\infty(\upsilon_0-1)^2},
    \quad
    z_4=\frac{(\upsilon_0-\upsilon_\infty)(\upsilon_\infty\upsilon_0-1)}{\upsilon_\infty(\upsilon_0-1)^2},\\
    z_5=\frac{\upsilon_0(x_2-\upsilon_\infty)}{(\upsilon_0-1)^2},\quad z_6=\frac{\upsilon_0(x_3-\upsilon_\infty)}{(\upsilon_0-1)^2}.
\end{align*}

\subsubsection{Segre surface of $\Pthree^{D_6}$}
The re-scaling of the parameters in the last column of Table \ref{tb:confluence} corresponding to the confluence limit from $\Pfive$ to $\Pthree^{D_6}$ 
produces the following re-scalings on the coefficients $\xi_{kj}$:
$$
\xi_{0k}\mapsto \xi_{0k},\quad \xi_{1k}\mapsto\frac{\xi_{1k}}{\epsilon},\quad \xi_{2k}\mapsto\frac{\xi_{2k}}{\epsilon},\quad\xi_{3k}\mapsto \xi_{3k},\quad \xi_{4k}\mapsto\frac{\xi_{4k}}{\epsilon},
$$
with $\xi_{j6}=\mathcal O(\epsilon^2)$ for $j=0,3$ and $\xi_{j6}=\mathcal O(\epsilon^3)$ for $j=1,2,4$.
This means that the variable $z_6$ is of order $\epsilon^2$ and drops out of the first two equations defining the $\mathcal Z$--Segre. The parameters are rescaled as follows
$$
\rho_5\to \rho_5, \quad\lambda_1\to \lambda_1, \quad \lambda_2\to\epsilon^2\lambda_2.
$$
These re-scalings imply that the last two equations defining the $\mathcal Z$--Segre preserve their form.By rescaling $z_6$, we can set $\lambda_2=1$:  
\begin{subequations}\label{p3:seg}
\begin{align}
& {z}_{1} +  {z}_{2} + {z}_{3} + {z}_{4} + {z}_{5} =0,\\
&  {z}_{4} + \rho_{5} {z}_{5} - 1=0,\\
&{z}_{3} {z}_{4} - \lambda_{1}{z}_{1} {z}_{2} =0,\\
&z_5 z_6 -   z_1  z_2=0.
\end{align}
\end{subequations}
Explicitly, 
the formulae relating the $\mathcal{Z}$-Segre to the monodromy manifold in the $\Pthree$ case are:
\begin{align*}
    z_1=\frac{\upsilon_0\upsilon_\infty^2-\upsilon_\infty x_1+x_2-\upsilon_\infty x_2 x_3}{\upsilon_0\upsilon_\infty^2},\quad 
    z_2=\frac{\upsilon_\infty x_1+x_2-\upsilon_\infty }{\upsilon_0\upsilon_\infty^2},\\
    z_3=\frac{x_2(\upsilon_\infty x_3-\upsilon_\infty^2 -1)}{\upsilon_0\upsilon_\infty^2},\qquad
    z_4=\frac{(\upsilon_0\upsilon_\infty-1)(x_2-\upsilon_\infty )}{\upsilon_0\upsilon_\infty^2},\\
    z_5=\frac{x_2(\upsilon_\infty -\upsilon_0)}{\upsilon_0\upsilon_\infty^2},\qquad
    z_6=\frac{(x_2-\upsilon_\infty)(1+\upsilon_\infty^2- -\upsilon_\infty x_3)}{\upsilon_0\upsilon_\infty^2}.
\end{align*}

\subsubsection{Segre surface of $\Ptwojm$}\label{subsec:pIIJM}
The re-scaling of the parameters in the last column of Table \ref{tb:confluence} corresponding to the confluence limit from $\Pfour$ to $\Ptwojm$
produces the following re-scaling on the coefficients $\xi_{kj}$:
$$
\xi_{0k}\mapsto \xi_{0k},\quad \xi_{1k}\mapsto\xi_{1k},\quad \xi_{2k}\mapsto\frac{\xi_{2k}}{\epsilon},\quad\xi_{3k}\mapsto \frac{\xi_{3k}}{\epsilon},\quad \xi_{4k}\mapsto\frac{\xi_{4k}}{\epsilon^2},
$$
with $\xi_{j2}=\mathcal O(\epsilon^2)$ for $j=0,\dots,4$.
This means that the variable $z_2$ is of order $\epsilon^2$ and drops out of the first equation defining the $\mathcal Z$--Segre. The parameters are re-scaled as follows
$$
\rho_4\to \rho_4, \quad\lambda_1\to \frac{\lambda_1}{\epsilon^2}, \quad \lambda_2\to\frac{\lambda_2}{\epsilon^2}, 
$$
so that the last two equations defining the $\mathcal Z$--Segre preserve their form. 
\begin{subequations}\label{p2jm:seg}
\begin{align}
&{z}_{1} +  {z}_{3} + {z}_{4} + {z}_{5} + {z}_{6}=0\\
&\rho_4 {z}_{4}  - 1=0\\
&{z}_{3} {z}_{4} -  \lambda_{1}{z}_{1} {z}_{2}=0\\
&{z}_{3} {z}_{4} -  \lambda_{2}{z}_{5} {z}_{6}=0.
\end{align}
\end{subequations}
Explicitly
\begin{align*}
    z_1=1-\frac{x_2 x_3}{\upsilon_\infty},
    \quad z_2=\frac{x_1-1}{\upsilon_\infty},\quad
    z_3=\upsilon_\infty-x_2-x_3+ \frac{x_2 x_3}{\upsilon_\infty},\\
    z_4=\upsilon_\infty-1,\quad z_5=x_2-\upsilon_\infty,\quad z_6=x_3-\upsilon_\infty.
\end{align*}

Note that we can re-scale all $z_i$ variables to set $\rho_4=1$, then we can absorb the parameter $\lambda_1$ in $z_2$, hence the final $\mathcal Z$--Segre for $\Ptwojm$ is a one parameter family as expected.



\section{Blow-downs of affine cubic surfaces}\label{se:blow-down}
Consider an embedded affine cubic surface  $\mathcal{X}\subseteq\mathbb{C}^3$ together with its canonical projective completion  $\overline{\mathcal{X}}\subseteq\mathbb{P}^3$. We say that $\mathcal{X}$ has a triangle (of lines) at infinity, if the hyperplane section 
at infinity, $\overline{\mathcal{X}}\setminus \mathcal{X}$, is a cubic curve which is the product of three lines that intersect pairwise at distinct points.

The embedded affine cubic surfaces corresponding to differential Painlev\'e equations \cite{CMR,SvdP} are all smooth for generic parameter values with a triangle of lines at infinity. They are, however, geometrically distinguished by their singularity structures on this triangle. {This is important in order to characterise the lines on the $\mathcal Y$--Segre obtained as blow down of the corresponding affine cubics. Indeed, as discussed in Subsections \ref{suse:iso} and \ref{se:allpsegres}, understanding the lines in the $\mathcal Y$--Segre is key to construct the isomorphism to the $\mathcal Z$--Segre in the ramified cases.}
        
In this section, we study smooth embedded affine cubic surfaces, with a triangle at infinity, and affine Segre surfaces naturally associated to them.

In Section \ref{subsec:classification}, we classify all embedded smooth cubic surfaces with a triangle at infinity. Then, in Section \ref{subsec:blowdown}, we give a natural construction of three associated Segre surfaces, in the regular case corresponding to blowing down any of the three lines at infinity.
In Section \ref{subsec:piicubics}, we use this construction to derive an explicit isomorphism between the two (decorated) character varieties for $\Ptwo$ known in the literature, one coming from the Jimbo-Miwa linear problem, the other from the Flaschka-Newell linear problem. 
Similarly, in sections \ref{suse:PI-Z-Segre}, \ref{suse:PVdeg-Z-Segre} and \ref{suse:PIIID7-Z-Segre} we build the remaining $\mathcal Z$--Segre surfaces for the ramified Painlev\'e equations.

\subsection{A classification}\label{subsec:classification}
In this section, we classify embedded smooth affine cubic surfaces with a triangle at infinity.
We start by deriving a normal form.
\begin{lemma}
Any embedded affine cubic surface in $\mathbb{C}^3$, with a triangle of lines at infinity, is affinely equivalent to
\begin{equation}\label{eq:cubic}
    x_1 x_2x_3+\epsilon_1 x_1^2+\epsilon_2 x_2^2+\epsilon_3 x_3^2+\omega_1 x_1+\omega_2 x_2+\omega_3 x_3+\omega_4=0,
\end{equation}
for some $\epsilon_{1,2,3}\in\{0,1\}$ and $\omega_k\in\mathbb{C}$, $1\leq k\leq 4$.
\end{lemma}
\begin{proof}
Let $\mathcal{X}$ be an affine cubic surface in $\{(x_1,x_2,x_3)\in\mathbb{C}^3\}$. Denote its canonical projective completion by $\overline{\mathcal{X}}$ in $\mathbb{P}^3$, so that the curve at infinity $\overline{\mathcal{X}}\setminus \mathcal{X}$ is a triangle of lines.
Using projective coordinates,
\begin{equation}\label{eq:homcoordinates}
    [X_0:X_1:X_2:X_3]=[1:x_1:x_2:x_3],
\end{equation}
the curve at infinity is thus described by
\begin{equation*}
    L_1 L_2 L_3=0,\quad X_0=0,
\end{equation*}
where each $L_k=L_k(X_1,X_2,X_3)$ is homogenous and linear. Since the three lines $L_k=0$, $1\leq k\leq 3$, pairwise intersect in distinct points, the affine map
\begin{equation*}
 x_k\mapsto L_k(x_1,x_2,x_3)\qquad (1\leq k\leq 3),   
\end{equation*}
has full rank and application of its inverse puts the cubic surface into the form
\begin{equation}\label{eq:cubic_mixed}
    x_1x_2x_3+a_1x_2x_3+a_2x_1x_3+a_3x_1x_2+Q(x_1,x_2,x_3)=0,
\end{equation}
for some $a_{1,2,3}\in\mathbb{C}$ and a quadratic polynomial $Q$ without mixed terms. Applying the affine scaling $x_k\mapsto x_k-a_k$, we may eliminate all mixed terms from \eqref{eq:cubic_mixed}, so that we are left with equation \eqref{eq:cubic}, for some $\epsilon_{1,2,3}\in\mathbb{C}$ and $\omega_k\in\mathbb{C}$, $1\leq k\leq 4$. By finally rescaling $x_k\mapsto c_kx_k$, for some $c_k\in\mathbb{C}^*$, $1\leq k\leq 3$, we can ensure that each $\epsilon_k\in\{0,1\}$ and the lemma follows.
\end{proof}

Next, we are going to have a look at singularities at infinity.
Take any (irreducible) affine cubic surface $\mathcal{X}\subseteq \mathbb{C}^3$, with a triangle of lines at infinity, in normal form \eqref{eq:cubic}. 
Using homogeneous coordinates \eqref{eq:homcoordinates}, its canonical projective completion is given by the zero locus of the homogeneous polynomial
\begin{equation*}
    F:=X_1 X_2X_3+(\epsilon_1 X_1^2+\epsilon_2 X_2^2+\epsilon_3 X_3^2)X_0+(\omega_1 X_1+\omega_2 X_2+\omega_3 X_3)X_0^2+\omega_4 X_0^3.
\end{equation*}
The curve at infinity $\overline{\mathcal{X}}\setminus \mathcal{X}$ is the triangle composed of the three lines
\begin{equation}\label{eq:lines_at_infinity}
    L_k^\infty=\{X\in \mathbb{P}^3: X_k=X_0=0\},\quad (k=1,2,3).
\end{equation}

The gradient of $F$ at $X_0=0$ is given by
\begin{equation*}
    \nabla F|_{X_0=0}=(\epsilon_1 X_1^2+\epsilon_2 X_2^2+\epsilon_3 X_3^2,X_2 X_3,X_1 X_3,X_1 X_2),
\end{equation*}
from which it immediately follows that the triangle of lines at infinity contains no singularities
 if and only if $\epsilon_1=\epsilon_2=\epsilon_3=1$. Furthermore, singularities can only be located at the three intersection points among the lines and, 
for any labelling $\{i,j,k\}=\{1,2,3\}$, the intersection point of $L_i^\infty$ and $L_j^\infty$ is a singularity if and only if $\epsilon_k=0$.

Let us now focus on one of the intersection points of the three lines at infinity, say
the point $S=[0:0:0:1]$. If $\epsilon_3=0$, then this is a singular point. We determine its type, in the notation by Arnol'd \cite{arnol_sing}, following \cite{bruce_wall79}. We start by introducing some inhomogeneous coordinates around $S$,
\begin{equation*}
    [X_0:X_1:X_2:X_3]=[u_1:u_2:u_3:1],
\end{equation*}
so that $S$ corresponds to $u=(0,0,0)$. The equation for the cubic in these local coordinates reads
\begin{equation*}
    f(u):=u_2u_3+\epsilon_1 u_1 u_2^2+\epsilon_2 u_1 u_3^2+\omega_1 u_1^2u_2+\omega_2 u_1^2u_3+\omega_3 u_1^2+\omega_4 u_1^3=0.
\end{equation*}
Next, we apply a weighted scaling 
\begin{equation*}
    u_k=r^{j_k}v_k\qquad (1\leq k\leq 3),
\end{equation*}
where $r$ is a free scalar, and we look for triples $(j_1,j_2,j_3)\in\mathbb{Q}_{>0}^3$, such that
\begin{equation*}
    f(u)=r\,f_0(v)+o(r)\qquad (r\rightarrow 0),
\end{equation*}
where the leading order coefficient $f_0(v)$ is such that $\{f_0(v)=0\}$ has an isolated singularity at $v=(0,0,0)$.
In such case, \cite[Lemma 1]{bruce_wall79} shows that the type of the singularity is given by
\begin{equation*}
    \operatorname{type}(S)=\begin{cases}
        A_n & \text{if }(j_1,j_2,j_3)=(\tfrac{1}{n+1},\tfrac{1}{2},\tfrac{1}{2}),\hspace{1.1cm} (n\geq 1),\\
        D_n & \text{if }(j_1,j_2,j_3)=(\tfrac{1}{n-1},\tfrac{n-2}{2(n-1)},\tfrac{1}{2}),\quad (n\geq 4),\\
        E_6 & \text{if }(j_1,j_2,j_3)=(\tfrac{1}{3},\tfrac{1}{4},\tfrac{1}{2}),\\
        E_7 & \text{if }(j_1,j_2,j_3)=(\tfrac{1}{3},\tfrac{2}{9},\tfrac{1}{2}),\\
        E_8 & \text{if }(j_1,j_2,j_3)=(\tfrac{1}{3},\tfrac{1}{5},\tfrac{1}{2}).
    \end{cases}
\end{equation*}

We first consider the following values for the weights $(j_1,j_2,j_3)=(\tfrac{1}{2},\tfrac{1}{2},\tfrac{1}{2})$, in which case
\begin{equation*}
    f(u)=r(v_2v_3+\omega_3 v_1^2)+o(r)\qquad (r\rightarrow 0).
\end{equation*}
Therefore, as long as $\omega_3\neq 0$, the leading order term has an isolated singularity at $v=(0,0,0)$, and $S$ is a singularity of type $A_1$.

Suppose now that $\omega_3=0$. Then we take the weights $(j_1,j_2,j_3)=(\tfrac{1}{3},\tfrac{1}{2},\tfrac{1}{2})$, so that
\begin{equation*}
    f(u)=r(v_2v_3+\omega_4v_1^3)+o(r)\qquad (r\rightarrow 0).
\end{equation*}
As long as $\omega_4\neq 0$, the leading order term has an isolated singularity at $v=(0,0,0)$, and $S$ is a singularity of type $A_2$.

Next, suppose that also $\omega_4=0$. Then we take the weights $(j_1,j_2,j_3)=(\tfrac{1}{4},\tfrac{1}{2},\tfrac{1}{2})$, so that
\begin{equation*}
    f(u)=r(v_2v_3+\omega_1 v_1^2v_2+\omega_2v_1^2v_3)+o(r)\qquad (r\rightarrow 0).
\end{equation*}
As long as $\omega_1\omega_2\neq 0$, the leading order term has an isolated singularity at $v=(0,0,0)$, and $S$ is a singularity of type $A_3$.

Next, suppose that also $\omega_1\omega_2=0$. Without loss of generality, we consider the case $\omega_1=0$. The only admissible choice of weights is $(j_1,j_2,j_3)=(\tfrac{1}{5},\tfrac{2}{5},\tfrac{3}{5})$, and these weights do not allow us to read of the singularity type.

We therefore first apply a locally invertible polynomial mapping, $u\mapsto (u_1,u_2-\omega_2u_1^2,u_3)$, before scaling with the weights $(j_1,j_2,j_3)=(\tfrac{1}{5},\tfrac{1}{2},\tfrac{1}{2})$, yielding
\begin{equation*}
    f(u_1,u_2-\omega_2u_1^2,u_3)=r(v_2v_3+\epsilon_1\omega_2^2 v_1^5)+o(r)\qquad (r\rightarrow 0).
\end{equation*}
Now, necessarily $\epsilon_1=1$ since else the cubic is reducible. Therefore, as long as $\omega_2\neq 0$, the leading order term has an isolated singularity at $v=(0,0,0)$, and $S$ is a singularity of type $A_4$. Similarly, if $\omega_2=0$ but $\omega_1\neq 0$, then $S$ is a singularity of type $A_4$.

All in all, we have
\begin{equation}\label{eq:sing_type}
    \operatorname{type}(S)=\begin{cases}
        \ns & \text{if }\epsilon_3=1,\\
        A_1 & \text{if }\epsilon_3=0,\omega_3\neq 0,\\
        A_2 & \text{if }\epsilon_3=0,\omega_3=0,\omega_4\neq 0,\\
        A_3 & \text{if }\epsilon_3=0,\omega_3=0,\omega_4= 0,\omega_1\omega_2\neq 0,\\
        A_4 & \text{if }\epsilon_3=0,\omega_3=0,\omega_4= 0,\omega_1\omega_2=0,\omega_1+\omega_2\neq 0,
    \end{cases}
\end{equation}
where $\ns$ signifies that $S$ is a regular point. In the final case when $\omega_k=0$, $1\leq k\leq 4$, the singularity is non-isolated. Indeed, in that case, necessarily $\epsilon_1=\epsilon_2=1$ (as otherwise the cubic surface is reducible) and its defining equation reads
\begin{equation*}
    X_1X_2X_3+X_0X_1^2+X_0 X_2^2=0.
\end{equation*}
This surface is singular on the line $\{X_1=X_2=0\}$ and has a further infinite number of lines lying in hyperplanes of the form $\{X_3=t\,X_0\}$, $t\in\mathbb{C}$, which all intersect this line.

Returning to the general discussion, we note, in particular, that singularities of types $A_5$, $D_4$, $D_5$, $E_6$, $\widetilde{E}_6$ cannot be realised as an intersection point in a triangle of lines on a cubic surface. These singularity types, however, do appear in cubic surfaces \cite{bruce_wall79}.
For example, the cubic surface
\begin{equation*}
   X_1X_2X_3-X_1^3 -X_2^3+X_0^2 X_1=0,
\end{equation*}
has an $A_5$ singularity at $[0:0:0:1]$. It contains only three lines,
\begin{align*}
    &\{X\in\mathbb{P}^3:X_1=X_2=0\},\\
    &\{X\in\mathbb{P}^3:X_2=0, X_0+X_1=0\},\\
    &\{X\in\mathbb{P}^3:X_2=0, X_0-X_1=0\},
\end{align*}
which all intersect at this singularity. In particular, this singularity is not the corner of a triangle.

It follows from the above considerations, that a smooth affine cubic surface, with a triangle of lines at infinity, can only have singularities (in its projective completion) at the three intersection points of lines at infinity, and they can only be of type $A_k$, $1\leq k\leq 4$.

By comparison with the classification of singular cubic surfaces, and their number of lines, in \cite{bruce_wall79}, we obtain a complete list of smooth affine cubic surfaces, with a triangle of lines at infinity, classified according to the singularity types of the corners. The result is given in Table \ref{table:afine_cubic_classification}.
We discuss a few examples in this table.

 \begin{example}[$A_1,A_1,A_3$]
By \eqref{eq:sing_type}, the general form of a cubic surface with these singularities at infinity is
\begin{equation*}
    x_1x_2x_3+\omega_1 x_1+\omega_2 x_2=0,\qquad \omega_1,\omega_2\neq 0.
\end{equation*}
It has $A_1$ singularities at $[0:1:0:0]$ and $[0:0:1:0]$, an $A_3$ singularity at $[0:0:0:1]$, and no further singularities in its canonical projective completion. Apart from the three lines at infinity, there are two further lines, given by $\{x_1=x_2=0\}$ and $\{x_1+x_2=x_3=0\}$.
By scaling $x_1\mapsto -\omega_2 x_1$ and $x_2\mapsto -\omega_1 x_2$, we may normalise the cubic such that $\omega_1=\omega_2=-1$.
 \end{example}

\begin{example}[$\ns,A_2,A_2$]
  By \eqref{eq:sing_type}, the general form of a cubic surface with these singularities at infinity is
\begin{equation*}
    x_1x_2x_3+x_1^2+\omega_1 x_1+\omega_4=0,\qquad \omega_4\neq 0.
\end{equation*}
 It has $A_2$ singularities at $[0:0:1:0]$ and $[0:0:0:1]$ in its projective completion, and no further singularities at infinity. Only when $\omega_1^2=4\omega_4$, the surface has a further finite singularity at
 $(x_1,x_2,x_3)=(-\tfrac{1}{2}\omega_1,0,0)$, of type $A_1$
 . Apart from the three lines at infinity, there are generically four further lines.
 By scaling $x_{1,2}\mapsto \sqrt{\omega_4} x_{1,2}$, we may normalise the cubic such that $\omega_4=1$, yielding a one-parameter family of affine cubic surfaces.
 \end{example}

\renewcommand{\arraystretch}{1.1}
\begin{table}[t]
\centering
\begin{tabular}{|c || c | c | c |} 
 \hline
singularities & P-eqn & \#lines & cubic \\
 \hline
$\ns$,$\ns$,$\ns$ & $\Psix$ & 24 & $x_1 x_2x_3+ x_1^2+ x_2^2+x_3^2+\omega_1 x_1+\omega_2 x_2+\omega_3 x_3+\omega_4$ \\
$\ns$,$\ns$,$A_1$ & $\Pfive$ & 18 & 
$x_1x_2x_3+x_1^2+x_2^2+\omega_1 x_1+\omega_2x_2+\omega_3x_3+R(\omega_{1,2,3})$\\
$\ns$,$\ns$,$A_2$ & $\Pthree^{D_6}$, $\Pfive^\text{deg}$ & 12 & 
$x_1x_2x_3+x_1^2+x_2^2+\omega_1 x_1+\omega_2x_2+\omega_1-1$\\
$\ns$,$\ns$,$A_3$ & $\Pthree^{D_7}$ & 7 & 
$x_1x_2x_3+x_1^2+x_2^2+\omega_1 x_1-x_2$\\
$\ns$,$\ns$,$A_4$ & $\Pthree^{D_8}$ & 3 & $x_1x_2x_3+x_1^2+x_2^2-x_2$\\
$\ns$,$A_1$,$A_1$ & $\Pfour$ & 13 & $x_1x_2x_3+x_1^2+\omega_1 x_1+\omega_2 (x_2+x_3)+\omega_2(1+\omega_1-\omega_2)$ \\
$A_1$,$A_1$,$A_1$ & $\Ptwo^\text{JM}$ & 9 & $x_1x_2x_3-x_1+\omega_2 x_2-x_3-\omega_2+1$ \\
$\ns$,$A_1$,$A_2$ & $\Ptwo^\text{FN}$ & 8 & $x_1x_2x_3+x_1^2+\omega_1 x_1-x_2+1$\\
$A_1$,$A_1$,$A_2$ & $\Pone$ & 5 & $x_1x_2x_3-x_1-x_2+1$ \\
\hline
$\ns$,$A_2$,$A_2$ & - & 4 & $x_1x_2x_3+x_1^2+\omega_1 x_1+1$\\
$\ns$,$A_1$,$A_3$ & - & 4 & $x_1x_2x_3+x_1^2-x_1-x_2$\\
$A_1$,$A_1$,$A_3$ & - & 2 & $x_1x_2x_3-x_1-x_2$\\
$A_1$,$A_2$,$A_2$ & - & 2 & $x_1x_2x_3-x_1+1$\\
$\ns$,$A_1$,$A_4$ & - & 1 & $x_1x_2x_3+x_1^2-x_2$ \\
$A_2$,$A_2$,$A_2$ & - & 0 & $x_1x_2x_3+1$\\
\hline
\end{tabular}
\caption{Table of smooth embedded affine cubic surfaces, with a triangle of lines at infinity, listed according to the types of singularities (in their canonical projective completions) at the three intersection points of lines at infinity. In the first column the types of singularities respectively at $[0:1:0:0]$, $[0:0:1:0]$, $[0:0:0:1]$, where the symbol '$\ns$' stands for a regular point. In the second column, the corresponding Painlev\'e equation(s), in the third column the number of affine lines and in the fourth column normal forms for the cubics. All the $\omega$'s are considered generic and in the second row, the rational function $R$ is given by
$R=1+\omega_3^2-\frac{\omega_3(\omega_2+\omega_1\omega_3)(\omega_1+\omega_2\omega_3)}{(\omega_3^2-1)^2}$.
}
\label{table:afine_cubic_classification}
\end{table}

\begin{example}[$A_1,A_1,A_2$]
By \eqref{eq:sing_type}, the general form of the cubic surface is
\begin{equation*}
    x_1x_2x_3+\omega_1 x_1+\omega_2 x_2+\omega_4=0,\qquad \omega_{1},\omega_2,\omega_4\neq 0.
\end{equation*}
 It has $A_1$ singularities at $[0:1:0:0]$ and $[0:0:1:0]$, an $A_2$ singularity at $[0:0:0:1]$, and no further singularities in its projective completion. Apart from the three lines at infinity, there are five further lines. 
 By scaling
 \begin{equation*}
     x_1\mapsto -\frac{\omega_1}{\omega_4} x_1,\quad
     x_2\mapsto -\frac{\omega_2}{\omega_4} x_2,\quad
     x_3\mapsto \frac{\omega_1\omega_2}{\omega_4} x_3,
 \end{equation*} 
 we may normalise the cubic such that $-\omega_1=-\omega_2=\omega_4=1$. This is the decorated character variety of $\Pone$.
\end{example}

\begin{example}[$\ns,\ns,A_2$]
By \eqref{eq:sing_type}, the general form of the cubic surface is
\begin{equation*}
    x_1x_2x_3+x_1^2+x_2^2+\omega_1 x_1+\omega_2 x_2+\omega_4=0,\qquad \omega_4\neq 0.
\end{equation*}
 It has an $A_2$ singularity at $[0:0:0:1]$, and no further singularities in its projective completion, unless
 \begin{equation*}
(\omega_1^2-4\;\omega_4^2)(\omega_2^2-4\;\omega_4^2)=0.
 \end{equation*}
  Apart from the three lines at infinity, there are generically $12$ further lines.
 By scaling $x_{1,2}\mapsto u\; x_{1,2}$, where $u$ is a root of $u^2-4\omega_1+\omega_4=0$, we may normalise the surface such that $\omega_4=\omega_1-1$, leading to the two-parameter family of cubic surfaces given in Table \ref{table:afine_cubic_classification}. 
 This is the decorated character variety of $\Pthree^{D_6}$.
\end{example}

\begin{example}[$\ns,A_1,A_1$]
By \eqref{eq:sing_type}, the general form of the cubic surface is
\begin{equation*}
    x_1x_2x_3+x_1^2+\omega_1 x_1+\omega_2 x_2+\omega_3 x_3+\omega_4=0,\qquad \omega_2,\omega_3,\omega_4\neq 0.
\end{equation*}
 It has $A_1$ singularities at $[0:0:1:0]$ and $[0:0:0:1]$, and no further singularities in its projective completion, unless
 \begin{equation*}
18\;\omega_1\omega_2\omega_3\omega_4+27\;\omega_2^2\omega_3^2-\omega_1^2\omega_4^2+4\;\omega_4^3-4\;\omega_1^3\omega_2\omega_3=0.
 \end{equation*}
 Apart from the three lines at infinity, there are generically $13$ further lines.
 By scaling, we can ensure that $\omega_2=\omega_3$ and $\omega_4=\omega_2(1+\omega_1-\omega_2)$, leading to the two-parameter family of cubic surfaces given in Table \ref{table:afine_cubic_classification}. 
 This is the decorated character variety of $\Pfour$.
\end{example}


\subsection{Constructing Segre surfaces}\label{subsec:blowdown}
In Section \ref{sec:blowdownsegre}, we showed how to blow down one of the lines at infinity of the cubic surface for $\Psix$, leading to an associated affine Segre surface. In this section, we consider this construction for general affine cubic surfaces with a triangle of lines at infinity.

Let us return to the general affine cubic surface $\mathcal{X}\subseteq \mathbb{C}^3$, with a triangle of lines at infinity, in normal form \eqref{eq:cubic}. For simplicity, let us further assume that $\mathcal{X}$ has no finite singularities.

We focus on the construction of an affine Segre surface which, in the regular case, comes from a blow-down of the line $L_1^\infty$ at infinity. Correspondingly introducing the variables
\begin{equation*}
y_1=x_1,\quad y_2=x_2,\quad y_3=x_3,\quad y_4=x_2x_3,
\end{equation*}
we obtain an affine Segre surface $\mathcal{Y}\subseteq\mathbb{C}^4$, given by
\begin{align}
    &y_2y_3-y_4=0,\\
    &y_1y_4+\epsilon_1 y_1^2+\epsilon_2 y_2^2+\epsilon_3 y_3^2+\omega_1 y_1+\omega_2 y_2+\omega_3 y_3+\omega_4=0.
\end{align}
The polynomial mapping
\begin{equation*}
\pi: \mathcal{X}\rightarrow \mathcal{Y}, x\mapsto y,
\end{equation*}
is an isomorphism between the affine varieties $\mathcal{X}$ and $\mathcal{Y}$.
Using homogeneous coordinates
\begin{equation}\label{eq:homcoordinatesy}
    [Y_0:Y_1:Y_2:Y_3:Y_4]=[1:y_1:y_2:y_3:y_4],
\end{equation}
we define the canonical projective completion $\overline{\mathcal{Y}}\subseteq \mathbb{P}^4$ of $\mathcal{Y}$, by the homogeneous equations,
\begin{equation}\label{eq:segre_surface_eqns}
\begin{aligned}
    &Y_2Y_3-Y_4Y_0=0,\\
    &Y_1Y_4+\epsilon_1 Y_1^2+\epsilon_2 Y_2^2+\epsilon_3 Y_3^2+(\omega_1 Y_1+\omega_2 Y_2+\omega_3 Y_3)Y_0+\omega_4 Y_0^2=0.
\end{aligned}
\end{equation}
The curve at infinity, $\overline{\mathcal{Y}}\setminus \mathcal{Y}$, is described by
\begin{equation*}
Y_1Y_4+\epsilon_1 Y_1^2+\epsilon_2 Y_2^2+\epsilon_3 Y_3^2=0,\quad Y_2Y_3=0,\quad Y_0=0.
\end{equation*}
This quartic curve factorises into two quadratic curves,
\begin{equation*}
C_2^\infty:\quad Y_1Y_4+\epsilon_1 Y_1^2+\epsilon_3 Y_3^2=0,\quad Y_2=0,\quad Y_0=0,
\end{equation*}
and 
\begin{equation*}
C_3^\infty:\quad Y_1Y_4+\epsilon_1 Y_1^2+\epsilon_2 Y_2^2=0,\quad Y_3=0,\quad Y_0=0,
\end{equation*}
which meet in two points,
\begin{equation*}
[0:1:0:0:-\epsilon_1],\qquad q_1^\infty:=[0:0:0:0:1].
\end{equation*}
The mapping $\pi$ extends to a regular bi-rational mapping
\begin{equation}\label{eq:blow_down}
    \pi:\overline{\mathcal{X}}\rightarrow \overline{\mathcal{Y}},
\end{equation}
which is described on the three lines at infinity by
\begin{align*}
L_1^\infty:& & \pi([0:0:X_2:X_3])&=[0:0:0:0:1],\\
L_2^\infty:& & \pi([0:X_1:0:X_3])&=[0:X_1^2:0:X_1X_3:-(\epsilon_1 X_1^2+\epsilon_3 X_3^2)],\\
L_3^\infty:& & \pi([0:X_1:X_2:0])&=[0:X_1^2:X_1X_2:0:-(\epsilon_1 X_1^2+\epsilon_2 X_2^2)],
\end{align*}
see equations \eqref{eq:lines_at_infinity}. In particular, $\pi(L_1^\infty)=\{q_1^\infty\}$ and $\pi(L_k^\infty)\subseteq C_k^\infty$ for $k=2,3$. Let us introduce some notation for the corner points of the triangle at infinity of the cubic,
\begin{align*}
 p_{12}^\infty&=[0:0:0:1],\\
 p_{13}^\infty&=[0:0:1:0],\\
 p_{23}^\infty&=[0:1:0:0],
\end{align*}
so that $p_{jk}^\infty$ is the intersection point of $L_j^\infty$ and $L_k^\infty$ for appropriate indices $j,k$, see Figure \ref{fig:triangle_at_infinity}. Then $\pi$ maps these corner points respectively to the following points on $\overline{\mathcal{Y}}$,
\begin{align*}
    q_{12}^\infty&=[0:0:0:1-\epsilon_3:\epsilon_3],\\
        q_{13}^\infty&=[0:0:1-\epsilon_2:0:\epsilon_2],\\
            q_{23}^\infty&=[0:1:0:0:-\epsilon_1].
\end{align*}

\begin{figure}[ht]
	\centering
	\begin{tikzpicture}[scale=0.8]

	\tikzstyle{star}  = [circle, minimum width=3.5pt, fill, inner sep=0pt];
	\tikzstyle{starsmall}  = [circle, minimum width=3.5pt, fill, inner sep=0pt];
	\tikzstyle{dot}  = [circle, minimum width=2.5pt, fill, inner sep=0pt];

\node[star]     (p23) at (0,2) {};
	\node     at ($(p23)+(0,0.8)$) {$p_{23}^\infty$};

\node[star]     (p12) at ({2*cos((-1*pi/6) r)},{2*sin((-1*pi/6) r)} ) {};
	\node     at ($(p12)+(0.7,-0.3)$) {$p_{12}^\infty$};

 \node[star]     (p13) at ({2*cos((-5*pi/6) r)},{2*sin((-5*pi/6) r)} ) {};
	\node     at ($(p13)+(-0.6,-0.3)$) {$p_{13}^\infty$};

	\draw[domain=-0.4:1.4,smooth,variable=\x,blue] plot ({\x*2*cos((-5*pi/6) r)+(1-\x)*2*cos((-1*pi/6) r)},{\x*2*sin((-5*pi/6) r)+(1-\x)*2*sin((-1*pi/6) r)});
 
	\draw[domain=-0.4:1.4,smooth,variable=\x,blue] plot ({\x*0+(1-\x)*2*cos((-1*pi/6) r)},{\x*2+(1-\x)*2*sin((-1*pi/6) r)});

    \draw[domain=-0.4:1.4,smooth,variable=\x,blue] plot ({\x*0+(1-\x)*2*cos((-5*pi/6) r)},{\x*2+(1-\x)*2*sin((-5*pi/6) r)});




     \node at ($0.5*(p13)+0.5*(p12)-(0,0.3)$) {$L_1^\infty$};
     \node at ($0.5*(p23)+0.5*(p12)+0.4*({cos((1*pi/6) r)},{sin((1*pi/6) r)} )$) {$L_2^\infty$};
     \node at ($0.5*(p23)+0.5*(p13)+0.3*({cos((5*pi/6) r)},{sin((5*pi/6) r)} )$) {$L_3^\infty$};

	\end{tikzpicture}
	\caption{Notation for lines and intersection points in triangle at infinity of the embedded affine cubic $\mathcal{X}$.}
	\label{fig:triangle_at_infinity}
\end{figure}
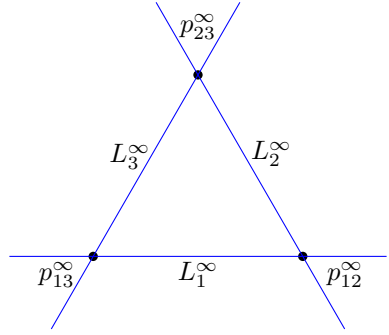

If $\epsilon_{2}=\epsilon_3=1$, then both $C_2^\infty$ and $C_3^\infty$ are irreducible conics, $\pi(L_k^\infty)=C_k^\infty$ for $k=2,3$, and $\pi$ is the blow-up of the Segre surface $\overline{\mathcal{Y}}$ at $q_1^\infty$, with exceptional divisor $L_1^\infty$.

If $\epsilon_2=0$, then $C_3^\infty$ is the product of two lines
\begin{align*}
C_{3,a}^\infty:& &Y_4+\epsilon_1 Y_1=0, & &Y_3=0, & &Y_0=0,\\
C_{3,b}^\infty:& &Y_1=0, & & Y_3=0, & &Y_0=0,
\end{align*}
which intersect at $q_{13}^\infty=[0:0:1:0:0]$, and $\pi(L_3^\infty)=C_{3,a}^\infty$.

If $\epsilon_3=0$, then $C_2^\infty$ is the product of two lines
\begin{align*}
C_{2,a}^\infty:& &Y_4+\epsilon_1 Y_1=0, & &Y_2=0, & &Y_0=0,\\
C_{2,b}^\infty:& &Y_1=0, & & Y_2=0, & &Y_0=0,
\end{align*}
which intersect at $q_{12}^\infty=[0:0:0:1:0]$, and $\pi(L_2^\infty)=C_{2,a}^\infty$. 

We conclude that the quartic curve at infinity, is either the product of two irreducible conics, or the product of an irreducible conic and two lines, or the product of four lines. In Figure \ref{fig:triangle_segre11}, the different cases are displayed.

Next, we consider singularities on $\overline{\mathcal{Y}}$. By assumption, $\mathcal{X}$ is smooth, so $\mathcal{Y}$ is smooth and thus singularities can only exist on the curve at infinity. The Jacobian of \eqref{eq:segre_surface_eqns} with respect to $Y$, at $Y_0=0$, is given by
\begin{equation*}
    J_Y|_{Y_0=0}=\begin{bmatrix}
        Y_4 & 0 & Y_3 & Y_2 & 0\\
        \omega_1 Y_1+\omega_2Y_2+\omega_3 Y_3 & Y_4+2\epsilon_1 Y_1 & 2\epsilon_2 Y_2 & 2\epsilon_3 Y_3 & Y_1
    \end{bmatrix}.
\end{equation*}
A point $Y$ at infinity is a singularity of $\overline{\mathcal{Y}}$ if and only if this Jacobian has rank less than two. It follows from this, that $\overline{\mathcal{Y}}$ can only have singularities at the intersection points of irreducible components of the curve at infinity. Note, furthermore, that $q_1^\infty$ is always a regular point, and $q_{23}^\infty\in \overline{\mathcal{Y}}$ is a singular point if and only if $p_{23}^\infty\in \overline{\mathcal{X}}$ is a singular point in which case their types are the same.

Note that the Jacobian has rank less than two at $q_{12}^\infty$ if and only if $\epsilon_3=\omega_3=0$. In such case, we can determine the singularity type of $q_{12}^\infty$ analogous to how we arrived at equation \eqref{eq:sing_type}. We use local affine variables $u$ defined through
\begin{equation*}
    [Y_0:Y_1:Y_2:Y_3:Y_4]=[u_1:u_2:u_1u_3:1:u_3],
\end{equation*}
so that $q_{12}^\infty$ corresponds to $u=(0,0,0)$. The first equation in \eqref{eq:segre_surface_eqns} is now trivially satisfied, and the second equation becomes
\begin{equation*}
    f(u):=u_2u_3+\epsilon_1u_2^2+\epsilon_2 u_1^2u_3^2+\omega_1 u_1u_2+\omega_2 u_1^2 u_3+\omega_4u_1^2=0.
\end{equation*}
The remainder of the procedure is the same as how we obtained equation \eqref{eq:sing_type}.  We apply a weighted scaling 
\begin{equation*}
    u_k=r^{j_k}v_k\qquad (1\leq k\leq 3),
\end{equation*}
where $r$ is a free scalar, and we look for triples $(j_1,j_2,j_3)\in\mathbb{Q}_{>0}^3$, such that a balance of overall weight $1$ occurs.

Putting $(j_1,j_2,j_3)=(\tfrac{1}{2},\tfrac{1}{2},\tfrac{1}{2})$, we have
\begin{equation*}
    f(u)=r(v_2v_3+\epsilon_1 v_2^2+\omega_1v_1v_2+\omega_4v_1^2)+o(r)\qquad (r\rightarrow 0).
\end{equation*}
As long as $\omega_4\neq 0$, the leading order term has an isolated singularity at $v=(0,0,0)$, and $q_{12}^\infty$ is a singularity of type $A_1$.

Next, suppose that also $\omega_4=0$. Then we apply the locally invertible polynomial mapping $u_3\mapsto u_3-\omega_1 u_1$
and take weights $(j_1,j_2,j_3)=(\tfrac{1}{3},\tfrac{1}{2},\tfrac{1}{2})$, to obtain
\begin{equation*}
    f(u_1,u_2,u_3-\omega_1 u_1)=r(v_2v_3+\epsilon_1v_2^2-\omega_1\omega_2 v_1^3)+o(r)\qquad (r\rightarrow 0).
\end{equation*}
As long as $\omega_1\omega_2\neq 0$, the leading order term has an isolated singularity at $v=(0,0,0)$, and $q_{12}^\infty$ is a singularity of type $A_2$.

Next, suppose that also $\omega_1\omega_2=0$. Without loss of generality, we consider the case $\omega_1=0$. Taking $(j_1,j_2,j_3)=(\tfrac{1}{4},\tfrac{1}{2},\tfrac{1}{2})$, gives the balance
\begin{equation*}
    f(u)=r(v_2v_3+\epsilon_1v_2^2+\omega_2 v_1^2v_3)+o(r)\qquad (r\rightarrow 0).
\end{equation*}
Now, necessarily $\epsilon_1=1$ since else the Segre surface is reducible. Therefore, as long as $\omega_2\neq 0$, the leading order term has an isolated singularity at $v=(0,0,0)$, and $q_{12}^\infty$ is a singularity of type $A_3$. Similarly, if $\omega_2=0$ but $\omega_1\neq 0$, then $q_{12}^\infty$ is a singularity of type $A_3$.

All in all,
\begin{equation*}
    \operatorname{type}(q_{12}^\infty)=\begin{cases}
        \ns & \text{if }\epsilon_3=1,\\
        \ns & \text{if }\epsilon_3=0,\omega_3\neq 0,\\
        A_1 & \text{if }\epsilon_3=0,\omega_3=0,\omega_4\neq 0,\\
        A_2 & \text{if }\epsilon_3=0,\omega_3=0,\omega_4= 0,\omega_1\omega_2\neq 0,\\
        A_3 & \text{if }\epsilon_3=0,\omega_3=0,\omega_4= 0,\omega_1\omega_2=0,\omega_1+\omega_2\neq 0,
    \end{cases}
\end{equation*}
 and in the final case when $\omega_k=0$, $1\leq k\leq 4$, the singularity is non-isolated and our running assumption that $\mathcal{X}$ is smooth is violated.
Comparing with \eqref{eq:sing_type}, we see that $q_{12}^\infty\in \overline{\mathcal{Y}}$ is an $A_{j-1}$ singularity when $p_{12}^\infty\in \overline{\mathcal{X}}$ is an $A_{j}$ singularity, for $1\leq j\leq 4$, where $A_0:=\ns$. Analogously, $q_{13}^\infty\in \overline{\mathcal{Y}}$ is an $A_{i-1}$ singularity when $p_{13}^\infty\in\overline{\mathcal{X}}$ is an $A_{i}$ singularity, $1\leq i\leq 4$. The correspondence between singularities on the triangle at infinity of the cubic and singularities on the quartic curve at infinity of the Segre surface are summarised in Figure \ref{fig:triangle_segre11}.

\begin{figure}[hb]

\centering
\begin{subfigure}{\textwidth}
 	\centering
	\begin{tikzpicture}[scale=0.8]

	\tikzstyle{star}  = [circle, minimum width=3.5pt, fill, inner sep=0pt];
	\tikzstyle{starsmall}  = [circle, minimum width=3.5pt, fill, inner sep=0pt];
	\tikzstyle{dot}  = [circle, minimum width=2.5pt, fill, inner sep=0pt];

\node[star]     (p23) at (0,2) {};
	\node     at ($(p23)+(0,0.8)$) {$p_{23}^\infty$};

\node[star]     (p12) at ({2*cos((-1*pi/6) r)},{2*sin((-1*pi/6) r)} ) {};
	\node     at ($(p12)+(0.7,-0.3)$) {$p_{12}^\infty$};

 \node[star]     (p13) at ({2*cos((-5*pi/6) r)},{2*sin((-5*pi/6) r)} ) {};
	\node     at ($(p13)+(-0.6,-0.3)$) {$p_{13}^\infty$};

	\draw[domain=-0.4:1.4,smooth,variable=\x,blue] plot ({\x*2*cos((-5*pi/6) r)+(1-\x)*2*cos((-1*pi/6) r)},{\x*2*sin((-5*pi/6) r)+(1-\x)*2*sin((-1*pi/6) r)});
 
	\draw[domain=-0.4:1.4,smooth,variable=\x,blue] plot ({\x*0+(1-\x)*2*cos((-1*pi/6) r)},{\x*2+(1-\x)*2*sin((-1*pi/6) r)});

    \draw[domain=-0.4:1.4,smooth,variable=\x,blue] plot ({\x*0+(1-\x)*2*cos((-5*pi/6) r)},{\x*2+(1-\x)*2*sin((-5*pi/6) r)});

     \node at ($0.5*(p13)+0.5*(p12)-(0,0.3)$) {$L_1^\infty$};
     \node at ($0.5*(p23)+0.5*(p12)+0.4*({cos((1*pi/6) r)},{sin((1*pi/6) r)} )$) {$L_2^\infty$};
     \node at ($0.5*(p23)+0.5*(p13)+0.3*({cos((5*pi/6) r)},{sin((5*pi/6) r)} )$) {$L_3^\infty$};

    \node[purple] at ($(p12)+(0.3,0.3)$) {{\boldmath $\ns$}};
    
    \node[purple] at ($(p23)+(0.5,0)$) {{\boldmath $A_k$}};

    \node[purple] at ($(p13)+(-0.3,0.3)$) {{\boldmath $\ns$}};

	\end{tikzpicture}\hspace{0.5cm}
 \begin{tikzpicture}[scale=0.8]

	\tikzstyle{star}  = [circle, minimum width=3.5pt, fill, inner sep=0pt];
	\tikzstyle{starsmall}  = [circle, minimum width=3.5pt, fill, inner sep=0pt];
	\tikzstyle{dot}  = [circle, minimum width=2.5pt, fill, inner sep=0pt];

    \draw[domain=-0.43:1.4,smooth,variable=\x,white] plot ({\x*0+(1-\x)*2*cos((-5*pi/6) r)},{\x*2+(1-\x)*2*sin((-5*pi/6) r)});

\node[star]     (q23) at (0,2) {};
	\node     at ($(q23)+(0,0.6)$) {$q_{23}^\infty$};

     \node[star] (q0) at ($(0,{2*sin((-1*pi/6) r)})$) {};
     \node at ($(q0)-(0,0.5)$) {$q_{1}^\infty$};

	\draw[domain=-1.5:2.5,smooth,variable=\y,blue] plot ({0.5*(\y-2)*(\y-2*sin((-1*pi/6) r))},{\y});
	\draw[domain=-1.5:2.5,smooth,variable=\y,blue] plot ({-0.5*(\y-2)*(\y-2*sin((-1*pi/6) r))},{\y});

	\node     at ($(1.6,0.5)$) {$C_2^\infty$};
	\node     at ($(-1.6,0.5)$) {$C_3^\infty$};

    \node[purple] at ($(q23)+(0,-0.4)$) {{\boldmath $A_k$}};

   \node[purple] at ($(q0)+(0,0.4)$) {{\boldmath $\ns$}};

	\end{tikzpicture}
    \caption{$\epsilon_2=\epsilon_3=1$}
    \label{fig:triangle_segre11A}
\end{subfigure}
\hfill
\begin{subfigure}{\textwidth}
	\centering
	\begin{tikzpicture}[scale=0.8]

	\tikzstyle{star}  = [circle, minimum width=3.5pt, fill, inner sep=0pt];
	\tikzstyle{starsmall}  = [circle, minimum width=3.5pt, fill, inner sep=0pt];
	\tikzstyle{dot}  = [circle, minimum width=2.5pt, fill, inner sep=0pt];

\node[star]     (p23) at (0,2) {};
	\node     at ($(p23)+(0,0.8)$) {$p_{23}^\infty$};

\node[star]     (p12) at ({2*cos((-1*pi/6) r)},{2*sin((-1*pi/6) r)} ) {};
	\node     at ($(p12)+(0.7,-0.3)$) {$p_{12}^\infty$};

 \node[star]     (p13) at ({2*cos((-5*pi/6) r)},{2*sin((-5*pi/6) r)} ) {};
	\node     at ($(p13)+(-0.6,-0.3)$) {$p_{13}^\infty$};

	\draw[domain=-0.4:1.4,smooth,variable=\x,blue] plot ({\x*2*cos((-5*pi/6) r)+(1-\x)*2*cos((-1*pi/6) r)},{\x*2*sin((-5*pi/6) r)+(1-\x)*2*sin((-1*pi/6) r)});
 
	\draw[domain=-0.4:1.4,smooth,variable=\x,blue] plot ({\x*0+(1-\x)*2*cos((-1*pi/6) r)},{\x*2+(1-\x)*2*sin((-1*pi/6) r)});

    \draw[domain=-0.4:1.4,smooth,variable=\x,blue] plot ({\x*0+(1-\x)*2*cos((-5*pi/6) r)},{\x*2+(1-\x)*2*sin((-5*pi/6) r)});

     \node at ($0.5*(p13)+0.5*(p12)-(0,0.3)$) {$L_1^\infty$};
     \node at ($0.5*(p23)+0.5*(p12)+0.4*({cos((1*pi/6) r)},{sin((1*pi/6) r)} )$) {$L_2^\infty$};
     \node at ($0.5*(p23)+0.5*(p13)+0.3*({cos((5*pi/6) r)},{sin((5*pi/6) r)} )$) {$L_3^\infty$};

    \node[purple] at ($(p12)+(0.3,0.3)$) {{\boldmath $\ns$}};
    
    \node[purple] at ($(p23)+(0.5,0)$) {{\boldmath $A_k$}};

    \node[purple] at ($(p13)+(-0.3,0.3)$) {{\boldmath $A_i$}};

	\end{tikzpicture}\hspace{0.5cm}
 \begin{tikzpicture}[scale=0.8]

	\tikzstyle{star}  = [circle, minimum width=3.5pt, fill, inner sep=0pt];
	\tikzstyle{starsmall}  = [circle, minimum width=3.5pt, fill, inner sep=0pt];
	\tikzstyle{dot}  = [circle, minimum width=2.5pt, fill, inner sep=0pt];

    \draw[domain=-0.43:1.4,smooth,variable=\x,white] plot ({\x*0+(1-\x)*2*cos((-5*pi/6) r)},{\x*2+(1-\x)*2*sin((-5*pi/6) r)});

\node[star]     (q23) at (0,2) {};
	\node     at ($(q23)+(0,0.6)$) {$q_{23}^\infty$};

     \node[star] (q0) at ($(0,{2*sin((-1*pi/6) r)})$) {};
     \node at ($(q0)-(0,0.5)$) {$q_{1}^\infty$};

\node[star]     (q13) at (-1.5,0.5) {};
	\node     at ($(q13)+(-0.6,0)$) {$q_{13}^\infty$};

 	\draw[domain=-0.4:1.4,smooth,variable=\x,blue] plot 
  ({\x*(-1.5)+(1-\x)*0},{\x*0.5+(1-\x)*2});
 	\draw[domain=-0.4:1.4,smooth,variable=\x,blue] plot 
  ({\x*(-1.5)+(1-\x)*0},{\x*0.5+(1-\x)*(-1)});
	\draw[domain=-1.5:2.5,smooth,variable=\y,blue] plot ({-0.5*(\y-2)*(\y-2*sin((-1*pi/6) r))},{\y});

	\node     at ($(1.65,0.5)$) {$C_2^\infty$};
	\node     at ($(-1.1,1.5)$) {$C_{3,a}^\infty$};
 	\node     at ($(-1.2,-0.45)$) {$C_{3,b}^\infty$};

   \node[purple] at ($(q13)+(0.75,0)$) {{\boldmath $A_{i-1}$}};

   \node[purple] at ($(q23)+(0,-0.4)$) {{\boldmath $A_k$}};

   \node[purple] at ($(q0)+(0,0.4)$) {{\boldmath $\ns$}};

	\end{tikzpicture}
    \caption{$\epsilon_2=0,\epsilon_3=1$}
    \label{fig:triangle_segre1B}
\end{subfigure}
\end{figure}
\begin{figure}[ht]\ContinuedFloat

\centering
\begin{subfigure}{\textwidth}
	\centering
	\begin{tikzpicture}[scale=0.8]

	\tikzstyle{star}  = [circle, minimum width=3.5pt, fill, inner sep=0pt];
	\tikzstyle{starsmall}  = [circle, minimum width=3.5pt, fill, inner sep=0pt];
	\tikzstyle{dot}  = [circle, minimum width=2.5pt, fill, inner sep=0pt];

\node[star]     (p23) at (0,2) {};
	\node     at ($(p23)+(0,0.8)$) {$p_{23}^\infty$};

\node[star]     (p12) at ({2*cos((-1*pi/6) r)},{2*sin((-1*pi/6) r)} ) {};
	\node     at ($(p12)+(0.7,-0.3)$) {$p_{12}^\infty$};

 \node[star]     (p13) at ({2*cos((-5*pi/6) r)},{2*sin((-5*pi/6) r)} ) {};
	\node     at ($(p13)+(-0.6,-0.3)$) {$p_{13}^\infty$};

	\draw[domain=-0.4:1.4,smooth,variable=\x,blue] plot ({\x*2*cos((-5*pi/6) r)+(1-\x)*2*cos((-1*pi/6) r)},{\x*2*sin((-5*pi/6) r)+(1-\x)*2*sin((-1*pi/6) r)});
 
	\draw[domain=-0.4:1.4,smooth,variable=\x,blue] plot ({\x*0+(1-\x)*2*cos((-1*pi/6) r)},{\x*2+(1-\x)*2*sin((-1*pi/6) r)});

    \draw[domain=-0.4:1.4,smooth,variable=\x,blue] plot ({\x*0+(1-\x)*2*cos((-5*pi/6) r)},{\x*2+(1-\x)*2*sin((-5*pi/6) r)});

     \node at ($0.5*(p13)+0.5*(p12)-(0,0.3)$) {$L_1^\infty$};
     \node at ($0.5*(p23)+0.5*(p12)+0.4*({cos((1*pi/6) r)},{sin((1*pi/6) r)} )$) {$L_2^\infty$};
     \node at ($0.5*(p23)+0.5*(p13)+0.3*({cos((5*pi/6) r)},{sin((5*pi/6) r)} )$) {$L_3^\infty$};

    \node[purple] at ($(p12)+(0.3,0.3)$) {{\boldmath $A_j$}};
    
    \node[purple] at ($(p23)+(0.5,0)$) {{\boldmath $A_k$}};

    \node[purple] at ($(p13)+(-0.3,0.3)$) {{\boldmath $\ns$}};

	\end{tikzpicture}\hspace{0.5cm}
 \begin{tikzpicture}[scale=0.8]

	\tikzstyle{star}  = [circle, minimum width=3.5pt, fill, inner sep=0pt];
	\tikzstyle{starsmall}  = [circle, minimum width=3.5pt, fill, inner sep=0pt];
	\tikzstyle{dot}  = [circle, minimum width=2.5pt, fill, inner sep=0pt];

    \draw[domain=-0.43:1.4,smooth,variable=\x,white] plot ({\x*0+(1-\x)*2*cos((-5*pi/6) r)},{\x*2+(1-\x)*2*sin((-5*pi/6) r)});

\node[star]     (q23) at (0,2) {};
	\node     at ($(q23)+(0,0.6)$) {$q_{23}^\infty$};

     \node[star] (q0) at ($(0,{2*sin((-1*pi/6) r)})$) {};
     \node at ($(q0)-(0,0.5)$) {$q_{1}^\infty$};

\node[star]     (q12) at (1.5,0.5) {};
	\node     at ($(q12)+(+0.6,0)$) {$q_{12}^\infty$};

	\draw[domain=-1.5:2.5,smooth,variable=\y,blue] plot ({0.5*(\y-2)*(\y-2*sin((-1*pi/6) r))},{\y});
 	\draw[domain=-0.4:1.4,smooth,variable=\x,blue] plot 
  ({\x*(1.5)+(1-\x)*0},{\x*0.5+(1-\x)*2});
 	\draw[domain=-0.4:1.4,smooth,variable=\x,blue] plot 
  ({\x*(1.5)+(1-\x)*0},{\x*0.5+(1-\x)*(-1)});

	\node     at ($(1.1,1.5)$) {$C_{2,a}^\infty$};
 	\node     at ($(1.2,-0.45)$) {$C_{2,b}^\infty$};
	\node     at ($(-1.6,0.5)$) {$C_3^\infty$};

   \node[purple] at ($(q23)+(0,-0.4)$) {{\boldmath $A_k$}};

   \node[purple] at ($(q0)+(0,0.4)$) {{\boldmath $\ns$}};

    \node[purple] at ($(q12)+(-0.8,0)$) {{\boldmath $A_{j-1}$}};

	\end{tikzpicture}
    \caption{$\epsilon_2=1,\epsilon_3=0$}
    \label{fig:triangle_segre11C}
\end{subfigure}
\hfill
\begin{subfigure}{\textwidth}
	\centering
	\begin{tikzpicture}[scale=0.8]

	\tikzstyle{star}  = [circle, minimum width=3.5pt, fill, inner sep=0pt];
	\tikzstyle{starsmall}  = [circle, minimum width=3.5pt, fill, inner sep=0pt];
	\tikzstyle{dot}  = [circle, minimum width=2.5pt, fill, inner sep=0pt];

\node[star]     (p23) at (0,2) {};
	\node     at ($(p23)+(0,0.8)$) {$p_{23}^\infty$};

\node[star]     (p12) at ({2*cos((-1*pi/6) r)},{2*sin((-1*pi/6) r)} ) {};
	\node     at ($(p12)+(0.7,-0.3)$) {$p_{12}^\infty$};

 \node[star]     (p13) at ({2*cos((-5*pi/6) r)},{2*sin((-5*pi/6) r)} ) {};
	\node     at ($(p13)+(-0.6,-0.3)$) {$p_{13}^\infty$};

	\draw[domain=-0.4:1.4,smooth,variable=\x,blue] plot ({\x*2*cos((-5*pi/6) r)+(1-\x)*2*cos((-1*pi/6) r)},{\x*2*sin((-5*pi/6) r)+(1-\x)*2*sin((-1*pi/6) r)});
 
	\draw[domain=-0.4:1.4,smooth,variable=\x,blue] plot ({\x*0+(1-\x)*2*cos((-1*pi/6) r)},{\x*2+(1-\x)*2*sin((-1*pi/6) r)});

    \draw[domain=-0.4:1.4,smooth,variable=\x,blue] plot ({\x*0+(1-\x)*2*cos((-5*pi/6) r)},{\x*2+(1-\x)*2*sin((-5*pi/6) r)});

     \node at ($0.5*(p13)+0.5*(p12)-(0,0.3)$) {$L_1^\infty$};
     \node at ($0.5*(p23)+0.5*(p12)+0.4*({cos((1*pi/6) r)},{sin((1*pi/6) r)} )$) {$L_2^\infty$};
     \node at ($0.5*(p23)+0.5*(p13)+0.3*({cos((5*pi/6) r)},{sin((5*pi/6) r)} )$) {$L_3^\infty$};

    \node[purple] at ($(p12)+(0.3,0.3)$) {{\boldmath $A_j$}};
    
    \node[purple] at ($(p23)+(0.5,0)$) {{\boldmath $A_k$}};

    \node[purple] at ($(p13)+(-0.3,0.3)$) {{\boldmath $A_i$}};

	\end{tikzpicture}\hspace{0.5cm}
 \begin{tikzpicture}[scale=0.8]

	\tikzstyle{star}  = [circle, minimum width=3.5pt, fill, inner sep=0pt];
	\tikzstyle{starsmall}  = [circle, minimum width=3.5pt, fill, inner sep=0pt];
	\tikzstyle{dot}  = [circle, minimum width=2.5pt, fill, inner sep=0pt];

    \draw[domain=-0.43:1.4,smooth,variable=\x,white] plot ({\x*0+(1-\x)*2*cos((-5*pi/6) r)},{\x*2+(1-\x)*2*sin((-5*pi/6) r)});

\node[star]     (q23) at (0,2) {};
	\node     at ($(q23)+(0,0.6)$) {$q_{23}^\infty$};

     \node[star] (q0) at ($(0,{2*sin((-1*pi/6) r)})$) {};
     \node at ($(q0)-(0,0.5)$) {$q_{1}^\infty$};

\node[star]     (q13) at (-1.5,0.5) {};
	\node     at ($(q13)+(-0.6,0)$) {$q_{13}^\infty$};

\node[star]     (q12) at (1.5,0.5) {};
	\node     at ($(q12)+(+0.6,0)$) {$q_{12}^\infty$};

 	\draw[domain=-0.4:1.4,smooth,variable=\x,blue] plot 
  ({\x*(-1.5)+(1-\x)*0},{\x*0.5+(1-\x)*2});
 	\draw[domain=-0.4:1.4,smooth,variable=\x,blue] plot 
  ({\x*(-1.5)+(1-\x)*0},{\x*0.5+(1-\x)*(-1)});
   	\draw[domain=-0.4:1.4,smooth,variable=\x,blue] plot 
  ({\x*(1.5)+(1-\x)*0},{\x*0.5+(1-\x)*2});
 	\draw[domain=-0.4:1.4,smooth,variable=\x,blue] plot 
  ({\x*(1.5)+(1-\x)*0},{\x*0.5+(1-\x)*(-1)});

	\node     at ($(1.1,1.5)$) {$C_{2,a}^\infty$};
 	\node     at ($(1.2,-0.45)$) {$C_{2,b}^\infty$};
	\node     at ($(-1.1,1.5)$) {$C_{3,a}^\infty$};
 	\node     at ($(-1.2,-0.45)$) {$C_{3,b}^\infty$};

   \node[purple] at ($(q23)+(0,-0.4)$) {{\boldmath $A_k$}};

   \node[purple] at ($(q0)+(0,0.4)$) {{\boldmath $\ns$}};

   \node[purple] at ($(q13)+(0.75,0)$) {{\boldmath $A_{i-1}$}};

       \node[purple] at ($(q12)+(-0.8,0)$) {{\boldmath $A_{j-1}$}};

	\end{tikzpicture}
    \caption{$\epsilon_2=\epsilon_3=0$}
    \label{fig:triangle_segre11D}
\end{subfigure}        
\caption{The configurations of lines and singularities on the curves at infinity of the embedded affine cubic surface $\mathcal{X}$ and embedded affine Segre surface $\mathcal{Y}$, for different choices of $\epsilon_2,\epsilon_3\in\{0,1\}$. Here the singularity type at each intersection point is depicted in purple, with corresponding indices satisfying $0\leq k\leq 4$, $1\leq i,j\leq 4$, and $A_0:=\ns$. The admissable values of $(i,j,k)$ can be read off Table \ref{table:afine_cubic_classification}.}
	\label{fig:triangle_segre11}

\end{figure}
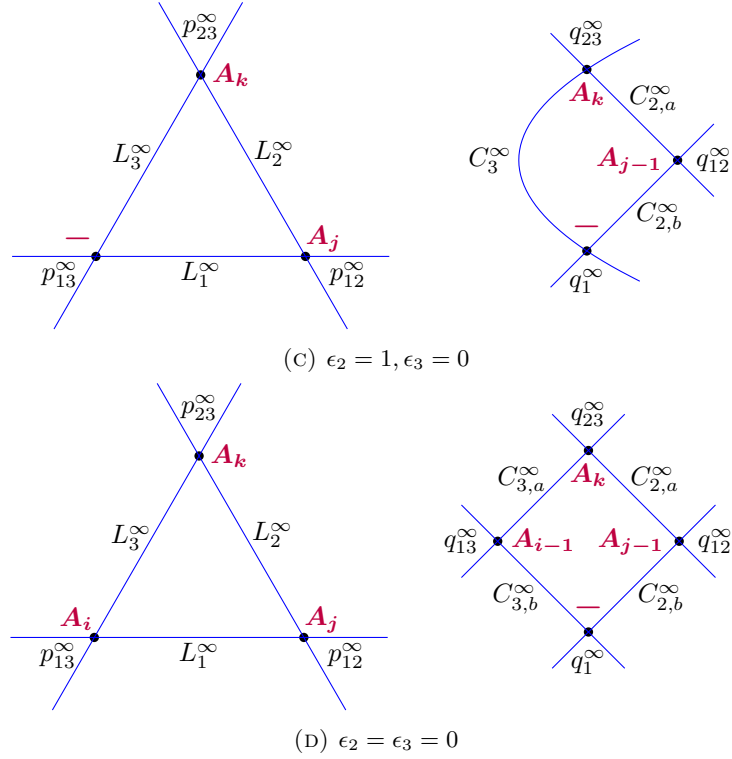

Finally, we discuss the relationship between affine lines on the cubic and affine lines on the Segre surface, in the following lemma.
\begin{lemma}\label{lem:linescubicsegre}
Let $L$ be any line on the cubic surface $\overline{\mathcal{X}}$, not at infinity. If $L$ intersects only with the line $L_1^\infty$ at infinity, then $\pi$ maps $L$ to a conic in $\overline{\mathcal{Y}}$. Otherwise, $\pi$ maps $L$ to a line in $\overline{\mathcal{Y}}$. All the lines in the Segre surface $\overline{\mathcal{Y}}$, not at infinity, arise in this way.
\end{lemma}
\begin{proof}
Take any line $L$ in the cubic surface $\overline{\mathcal{X}}$, not at infinity. Suppose that this line intersects only $L_1^\infty$ at infinity. Then its affine part admits a parametrisation
\begin{equation}\label{eq:rational_parametrisation}
    x=t\,a+b\quad (t\in\mathbb{C}),
\end{equation}
for some $a,b\in\mathbb{C}^3$, with $a_1=0$ and $a_{2},a_3\neq 0$. Therefore
\begin{equation*}
    \pi(L)=\{[t_0^2:b_1 t_0^2:a_2t_1t_0+b_2 t_0^2:a_3t_1t_0+b_3t_0^2:(a_2t_1+b_2 t_0)(a_3t_1+b_3t_0)]:[t_0:t_1]\in\mathbb{P}^1\},
\end{equation*}
is a conic.

Otherwise, the affine part of the line admits a parametrisation \eqref{eq:rational_parametrisation}, with either $a_2=0$ or $a_3=0$, and the image under $\pi$ is thus a line, since $y_4=y_2 y_3$ will be affine linear in $t$.

Conversely, take any line $L$ in the Segre surface $\overline{\mathcal{Y}}$, not at infinity. The affine part of the line admits a parametrisation of the form
\begin{equation*}
y=a\,t+b,
\end{equation*}
for some $a,b\in\mathbb{C}^4$, with at least one $a_k\neq 0$, $1\leq k\leq 3$, since $y_4=y_2y_3$. It follows that
\begin{equation*}
\pi^{-1}(L)=\{[a_1 t_1+b_1 t_0:a_2t_1+b_2 t_0:a_3t_1+b_3t_0]:[t_0:t_1]\in\mathbb{P}^1\},
\end{equation*}
is a line. In particular $L$ is the image under $\pi$ of a line in $\overline{\mathcal{X}}$ and the lemma follows.
\end{proof}

\begin{remark}
The classical Cayley–Salmon theorem \cite{cay1849,sal1849} states that any smooth cubic surface in $\mathbb{P}^3$ contains precisely 27 lines. Furthermore, it is well-known that any line in a smooth cubic surface intersects with precisely 10 others. Therefore, setting $\epsilon_{1,2,3}=1$, it follows from the above lemma that the smooth Segre surface $\overline{\mathcal{Y}}$ contains precisely $27-10-1=16$ lines, as expected. This is also consistent with the line counting on the blow-up model in Remark \ref{remark:lines_segre}.
\end{remark}

Starting with any embedded affine cubic $\mathcal{X}$ given by the normal form \eqref{eq:cubic}, we can carry out the above construction with respect to any of the three lines at infinity, leading, in general, to three inequivalent affine Segre surfaces. The results of this for the cubic surfaces corresponding to Painlev\'e equations are given in Table \ref{table:segre_summary}.

In the following sections, we will discuss some examples of this. For this purpose, and for general reference, 
we have provided Table \ref{table:segre_classification} which details the possible singularity configurations and number of lines on Segre surfaces. We remark that the singularity configuration does not always determine the number of lines, and vice versa. For example, as follows from the table, there exist Segre surfaces with one singularity, of type $A_3$, both with $5$ lines and with $4$ lines on them.

\renewcommand{\arraystretch}{1.2}
\begin{table}[t]
\centering
\begin{tabular}{|c | c || c | c | c |} 
 \hline
$\mathcal Y$ singularities & P-eqn & $L_1^\infty$ & $L_2^\infty$ & $L_3^\infty$ \\
 \hline
$\ns$,$\ns$,$\ns$ & $\Psix$ & $\ns$ $(16,0)$ & $\ns$ $(16,0)$ & $\ns$ $(16,0)$\\ \hline
$\ns$,$\ns$,$A_1$ & $\Pfive$ & $\ns$ $(14,2)$ & $\ns$ $(14,2)$ & $A_1$ $(12,0)$\\ \hline
$\ns$,$\ns$,$A_2$ & $\Pthree^{D_6}$, $\Pfive^\text{deg}$ & $A_1$ $(10,2)$ & $A_1$ $(10,2)$ & $A_2$ $(8,0)$\\ \hline
$\ns$,$\ns$,$A_3$ & $\Pthree^{D_7}$ & $A_2$ $(6,2)$ & $A_2$ $(6,2)$ & $A_3$ $(5,0)$\\ \hline
$\ns$,$\ns$,$A_4$ & $\Pthree^{D_8}$ & $A_3$ $(3,2)$ & $A_3$ $(2,2)$ & $A_4$ $(3,0)$\\ \hline
$\ns$,$A_1$,$A_1$ & $\Pfour$ & $\ns$ $(12,4)$ & $A_1$ $(10,2)$ & $A_1$ $(10,2)$\\ \hline
$A_1$,$A_1$,$A_1$ & $\Ptwo^\text{JM}$ & $A_1$ $(8,4)$ & $A_1$ $(8,4)$ & $A_1$ $(8,4)$\\ \hline
$\ns$,$A_1$,$A_2$ & $\Ptwo^\text{FN}$ & $A_1$ $(8,4)$ & $2A_1$ $(7,2)$ & $A_2$ $(6,2)$\\ \hline
$A_1$,$A_1$,$A_2$ & $\Pone$ & $2A_1$ $(5,4)$ & $2A_1$ $(5,4)$ & $A_2$ $(4,4)$\\ \hline
\end{tabular}
\caption{Table summarising geometric data of embedded affine Segre surfaces obtained by applying the construction in Section \ref{subsec:blowdown} with respect to any of the three lines at infinity, for each of the Painlev\'e cubic surfaces, for generic parameter values. The rows represent the different Painlev\'e cubic surfaces and, in the three columns on the right,
the line at infinity used in the construction is specified. Each entry takes the form `$S$, $(m,n)$', with $S$ listing the singularities at infinity on the Segre surface, $m$ the number of affine lines and $n$ the number of lines at infinity.}\label{table:segre_summary}
\end{table}

\renewcommand{\arraystretch}{1.2}
\begin{table}[t]
\centering
\begin{tabular}{|c | c | c | c | c | c |} 
 \hline
 $\ns$ & $A_1$ & $A_2$ & $A_3$ & $A_3$ & $A_4$ \\
  \hline
 $16$ & $12$ & $8$ & $5$ & $4$ & $3$\\
 \hline \hline
   $2A_1$& $2A_1$ & $A_1+A_2$ & $A_1+A_3$ & $3 A_1$ & $2A_1+A_2$ \\
   \hline
  $9$ & $8$ & $6$ & $3$ & $6$ & $4$\\
   \hline \hline
   $2A_1+A_3$ & $4A_1$ & $D_4$ & $D_5$ & &\\
   \hline
   $2$ & $4$ & $2$ & $1$ & &\\
   \hline
\end{tabular}
\caption{Possible singularity configurations and number of lines on Segre surfaces in $\mathbb C\mathbb P^4$, taken from Dolgachev \cite[Table 8.6]{Dolgachev}. In each of the three blocks, the first row gives the singularities and the second row the number lines. }\label{table:segre_classification}
\end{table}



\subsection{Affine cubic and Segre surfaces associated with $\Ptwo$}\label{subsec:piicubics}
In this section, we study the two affine cubic surfaces associated with $\Ptwo$ and corresponding Segre surfaces. In particular, we are going to use the construction in the last section to derive an explicit isomorphism between the two affine cubics. In particular, this shows that the $\mathcal{Z}$-Segre for $\Ptwo^{\text{JM}}$ is also a natural $\mathcal{Z}$-Segre for $\Ptwo^{\text{FN}}$.

Consider the general $\Ptwo^{\text{JM}}$ cubic
\begin{equation}\label{eq:piijmcubic}
    x_1x_2x_3-x_1+\omega_2x_2-x_3-\omega_2+1=0.
\end{equation}
It is smooth unless $\omega_2=-1$ and reduces to the character variety of $\Pone$ as $\omega_2\rightarrow 0$.

For generic $\omega_2$, the embedded cubic contains $9$ affine lines, explicitly given by
\begin{align*}
&L_1^{\text{JM}}=\{x_1+\omega_2=0, x_3=1\} & 
&L_6^{\text{JM}}=\{x_1-1=0,x_2=1\}, \\
&L_2^{\text{JM}}=\{x_3+\omega_2=0, x_1=1\}, & 
&L_7^{\text{JM}}=\{x_3-1=0,x_2=1\}, \\
&L_3^{\text{JM}}=\{x_1+\omega_2=0,x_2=-\omega_2^{-1}\}, & 
&L_8^{\text{JM}}=\{x_1=0,x_3-\omega_2x_2=1-\omega_2\}, \\
&L_4^{\text{JM}}=\{x_3+\omega_2=0,x_2=-\omega_2^{-1}\}, & 
&L_9^{\text{JM}}=\{x_3=0,x_1-\omega_2x_2=1-\omega_2\}, \\
&L_5^{\text{JM}}=\{x_2=0,x_1+x_3=1-\omega_2\}. & 
&
\end{align*}

We follow the construction in the previous section with respect to the line $L_3^\infty$ at infinity and introduce variables
\begin{equation*}
    y_1=x_1,\quad y_2=x_2,\quad y_3=x_3,\quad y_4=x_1 x_2,
\end{equation*}
to obtain the corresponding affine Segre surface $\mathcal{Y}^\text{JM}\subseteq\mathbb{C}^4$, defined by
\begin{subequations}
\begin{align}
    &y_1y_2-y_4=0,\label{eq:PIIJMsegre:a}\\
    &y_3y_4-y_1+\omega_2y_2-y_3-\omega_2+1=0.\label{eq:PIIJMsegre:b}
\end{align}
\end{subequations}
We obtain the canonical projective completion $\overline{\mathcal{Y}}^\text{JM}\subseteq\mathbb{P}^4$ through homogeneous coordinates \eqref{eq:homcoordinatesy}, and the curve at infinity is given by
\begin{equation*}
    Y_0=0,\quad Y_1 Y_2=0,\quad Y_3 Y_4=0.
\end{equation*}
This curve is the product of four lines, intersecting in four points, forming a rectangle, as in Figure \ref{fig:triangle_segre11D}. The four intersection points are given by
\begin{equation*}
p_k^\text{JM}:\quad Y_k=1,\quad Y_j=0\quad (0\leq j\leq 4, j\neq k),
\end{equation*}
where $1\leq k\leq 4$. The point $p_3^\text{JM}$ is an $A_1$ singularity and the other three are regular points in $\overline{\mathcal{Y}}^\text{JM}$.
Note, furthermore, that $x\rightarrow y$ maps $L_k^{\text{JM}}$, $1\leq k\leq 8$ to affine lines on $\mathcal{Y}^\text{JM}$, but $L_9^{\text{JM}}$ is mapped to a conic. Thus $\overline{\mathcal{Y}}^\text{JM}$ is a Segre surface with one $A_1$ singularity and $12$ lines, four of which lie at infinity, as indicated in the corresponding entry of Table \ref{table:segre_summary}.

Now, recall the affine Segre surface \eqref{p2jm:seg} of $\Ptwo^\text{JM}$ obtained by confluence in Section \ref{subsec:pIIJM}. It is affinely equivalent to $\mathcal{Y}^\text{JM}$ under
\begin{align*}
    z_1&=\frac{\omega_2(y_4-1)}{1+\omega_2}, &  z_4&=1,\\
    z_2&=-1-\frac{y_3}{\omega_2}, & z_5&=\frac{y_1-1}{1+\omega_2},\\
    z_3&=\frac{1-y_1+\omega_2y_2-\omega_2y_4}{1+\omega_2}, & z_6&=-\frac{1+\omega_2y_2}{1+\omega_2},
\end{align*}
with parameter values in \eqref{p2jm:seg} given by
$\rho_4=\lambda_1=1$ and $\lambda_2=1+\omega_2$.

Next, we turn our attention to the embedded affine cubic surface of $\Ptwo^\text{FN}$, defined by
\begin{equation}\label{eq:piifncubic}
  \tilde x_1\tilde x_2\tilde x_3+\tilde x_1^2+\omega_1 \tilde x_1-\tilde x_2+1=0.
\end{equation}
It is smooth unless $\omega_1^2=4$. It will be helpful to write $\omega_1=-(\upsilon_0+\upsilon_0^{-1})$, so that the $8$ affine lines on the cubic admit a simple description,
\begin{align*}
&L_1^{\text{FN}}=\{\tilde{x}_1=\upsilon_0, \tilde{x}_3=\upsilon_0^{-1}\}, & 
&L_5^{\text{FN}}=\{x_2=1,\tilde{x}_1+\tilde{x}_3=\upsilon_0+\upsilon_0^{-1}\}, \\
&L_2^{\text{FN}}=\{\tilde{x}_3=\upsilon_0, \tilde{x}_1=\upsilon_0^{-1}\}, & 
&L_6^{\text{FN}}=\{\tilde{x}_1=\upsilon_0^{-1},\tilde{x}_2=0\}, \\
&L_3^{\text{FN}}=\{\tilde{x}_1=\upsilon_0,\tilde{x}_2=0\}, & 
&L_7^{\text{FN}}=\{\upsilon_0\tilde{x}_1+\tilde{x}_2=1,\tilde{x}_3=\upsilon_0^{-1}\}, \\
&L_4^{\text{FN}}=\{\tilde{x}_3=\upsilon_0,\tilde{x}_1+\upsilon_0\tilde{x}_2=\upsilon_0\}, & 
&L_8^{\text{FN}}=\{\tilde{x}_1=0,\tilde{x}_2=1\}.
\end{align*}


Next, we consider the construction Section \ref{sec:blowdownsegre} with respect to the line $L_1^\infty$ at infinity, so that the resulting affine Segre surface has only an $A_1$ singularity in its canonical projective completion, like $\overline{\mathcal{Y}}^\text{JM}$. So, we introduce the variables
\begin{equation*}
    \tilde y_1=\tilde x_1,\quad \tilde y_2=\tilde x_2,\quad \tilde y_3=\tilde x_3,\quad \tilde y_4=\tilde x_2 \tilde x_3,
\end{equation*}
to obtain the affine Segre surface $\mathcal{Y}^\text{FN}\subseteq\mathbb{C}^4$ defined by
\begin{align*}
    &\tilde{y}_2\tilde{y}_3-\tilde{y}_4=0,\\
    &\tilde{y}_1\tilde{y}_4+\tilde{y}_1^2+\omega_1\tilde{y}_1-\tilde{y}_2+1=0.
\end{align*}
Its canonical projective completion $\overline{\mathcal{Y}}^\text{FN}\subseteq\mathbb{P}^4$ through homogeneous coordinates $\widetilde{Y}$, see equation \eqref{eq:homcoordinatesy}, and the curve at infinity is given by
\begin{equation*}
    \tilde Y_0=0,\quad \tilde Y_2 \tilde Y_3=0,\quad \tilde Y_1(\tilde Y_1+\tilde Y_4)=0.
\end{equation*}
This curve is also the product of four lines, intersecting in four points, forming a rectangle, as in Figure \ref{fig:triangle_segre11D}. The four intersection points are given by
\begin{align*}
p_1^\text{FN}&=[0:1:0:0:-1],\\
p_2^\text{FN}&=[0:0:0:0:1],\\
p_3^\text{FN}&=[0:0:0:1:0],\\
p_4^\text{FN}&=[0:0:1:0:0],
\end{align*}
where $1\leq k\leq 4$. The point $p_3^\text{FN}$ is an $A_1$ singularity and the other three are regular points in $\overline{\mathcal{Y}}^\text{FN}$.
Note, furthermore, that $x\rightarrow y$ maps each of the affine lines $L_k^{\text{FN}}$, $1\leq k\leq 8$, to an affine line on $\mathcal{Y}^\text{FN}$.

Since $\mathcal{Y}^\text{JM}$ and $\mathcal{Y}^\text{FN}$ are both smooth affine Segre surfaces with a `rectangle of lines' at infinity in their canonical projective completions, with one $A_1$ singularity, it is natural to ask whether they might be affinely equivalent. We thus look for an $A\in GL_5(\mathbb{C})$, such that the mapping
\begin{equation*}
   \mathbb{P}^4\rightarrow \mathbb{P}^4, Y\mapsto \widetilde{Y}=AY,
\end{equation*}
maps $\mathcal{Y}^\text{JM}$ onto $\mathcal{Y}^\text{FN}$. We require that the curves at infinity are mapped to one another, and to this end we impose that the intersection point $p_k^\text{JM}$ is mapped to $p_k^\text{FN}$, $1\leq k\leq 4$. As a result, $A$ takes the form
\begin{equation*}
    A=\begin{bmatrix}
        a_{11} & 0 & 0 & 0 & 0\\
        a_{21} & a_{22} & 0 & 0 & 0\\
        a_{31} & 0 & 0 & 0 & a_{35}\\
        a_{41} & 0 & 0 & a_{44} & 0\\
        a_{51} & -a_{22} & a_{53} & 0 & 0\\
    \end{bmatrix}.
\end{equation*}
where we may scale $A$ such that $a_{11}=1$. By direct substitution of
\begin{equation}\label{eq:yytildemapping}
    \widetilde{y}=A\cdot(1,y_1,y_2,y_3,y_4)^T,
\end{equation}
into the equations defining $\mathcal{Y}^\text{FN}$, and simplification modulo the equations defining 
$\mathcal{Y}^\text{JM}$, we obtain an overdetermined linear system, from which the remaining coefficients of $A$ are determined,
\begin{equation*}
    A=\begin{bmatrix}
        1 & 0 & 0 & 0 & 0\\
        0 & \frac{1-\omega_2}{\omega_1\omega_2} & 0 & 0 & 0\\
        1 & 0 & 0 & 0 & -1\\
        0 & 0 & 0 & \frac{1-\omega_2}{\omega_1\omega_2} & 0\\
        -\omega_1 & \frac{\omega_2-1}{\omega_1\omega_2} & \frac{1-\omega_2}{\omega_1} & 0 & 0\\
    \end{bmatrix},
\end{equation*}
as well as a remaining necessary and sufficient condition relating $\omega_1$ and $\omega_2$, namely 
\begin{equation}\label{eq:omega_relations}
   \omega_1^2+2=\omega_2+\omega_2^{-1}.
\end{equation}
In other words, when \eqref{eq:omega_relations} holds, then \eqref{eq:yytildemapping} defines an affine transformation that induces an isomorphism from $\mathcal{Y}^\text{JM}$ to $\mathcal{Y}^\text{FN}$.
As a consequence, we obtain the following result.
\begin{proposition}
\label{prop:PII}
The affine cubic surface of $\Ptwo^{\text{JM}}$, defined by equation \eqref{eq:piijmcubic}, and
 the affine cubic surface of $\Ptwo^{\text{FN}}$, defined by equation \eqref{eq:piifncubic}, 
are isomorphic when their parameters are related by \eqref{eq:omega_relations}.
Using the rational parametrisation 
\begin{equation*}
    \omega_1=-(\upsilon_0+\upsilon_0^{-1}), \quad \omega_2=-\upsilon_0^2, \quad \upsilon_0\in\mathbb{C}^*,
\end{equation*} an explicit isomorphism is given by
 \begin{equation*}
 \tilde x_1=\upsilon_0^{-1}x_1,\quad \tilde x_2=1-x_1x_2,\quad    \tilde x_3=\upsilon_0^{-1}x_3, 
 \end{equation*}
 with inverse
\begin{equation*}
x_1=\upsilon_0\tilde x_1,\quad x_2=1+\upsilon_0^{-2}-\upsilon_0^{-1}(\tilde x_1+\tilde x_2\tilde x_3),\quad    x_3=\upsilon_0\tilde x_3.
 \end{equation*} 
\end{proposition}
\begin{proof}
From equation \eqref{eq:yytildemapping}, we obtain the affine transformation
\begin{align*}
    \widetilde{y}_1=\upsilon_0^{-1}y_1,\quad
    \widetilde{y}_2=1-y_4,\quad
    \widetilde{y}_3=\upsilon_0^{-1}y_3,\quad
    \widetilde{y}_4=\upsilon_0+\upsilon_0^{-1}-\upsilon_0^{-1}y_1-\upsilon_0y_2,
\end{align*}
which induces an isomorphism from $\mathcal{Y}^\text{JM}$ to $\mathcal{Y}^\text{FN}$.
Now, note that the mapping $x\mapsto y$ from the affine cubic of $\mathcal{Y}^\text{JM}$ to $\mathcal{Y}^\text{JM}$ is an isomorphism of affine varieties, and similarly $\widetilde{x}\mapsto \widetilde{y}$ is   from the affine cubic of $\mathcal{Y}^\text{FN}$ to $\mathcal{Y}^\text{FN}$ is an isomorphism. Composing the three isomorphisms, $x\mapsto y$, $y\mapsto \widetilde{y}$ and $\widetilde{y}\mapsto \widetilde{x}$, yields the isomorphism between the affine cubic surfaces of $\mathcal{Y}^\text{JM}$ and $\mathcal{Y}^\text{FN}$ in the proposition.
\end{proof}
\begin{remark}
 The isomorphism in Proposition \ref{prop:PII} maps the line $L_k^{\text{FN}}$ to $L_k^{\text{JM}}$, $1\leq k\leq 8$. Furthermore, the inverse image of the remaining affine line, $L_9^{\text{JM}}$, is a conic in $\mathcal{Y}^\text{FN}$.
\end{remark}

\begin{remark}
    We note that a bi-rational mapping between the cubic surfaces of $\Ptwo^{\text{FN}}$ and $\Ptwo^{\text{JM}}$ was given in \cite[Remark 2.1]{CMR}. That mapping is singular along $x_2x_3=0$ and does not provide a global isomorphism as in Proposition \ref{prop:PII}.
\end{remark}

We further obtain an affine equivalence between the $\mathcal{Z}$--Segre listed in Table \ref{tb:Z-Segre}, with parameter value $\lambda_2=1-\upsilon_0^2$, and the Segre surface $\mathcal{Y}^\text{FN}$ with $\omega_1=-(\upsilon_0+\upsilon_0^{-1})$, given by
\begin{align*}
    z_1&=\frac{\upsilon_0^2\,\widetilde{y}_2}{1-\upsilon_0^2}, &  z_4&=1,\\
    z_2&=\frac{\widetilde{y}_3}{\upsilon_0}-1, & z_5&=\frac{\upsilon_0 \widetilde{y}_1-1}{1-\upsilon_0^2},\\
    z_3&=\frac{\upsilon_0(\widetilde{y}_4-\upsilon_0\widetilde{y}_2)}{1-\upsilon_0^2}, & z_6&=\frac{\upsilon_0(\upsilon_0-\widetilde{y}_1-\widetilde{y}_4)}{1-\upsilon_0^2}.
\end{align*}

\subsection{Cubic and Segre surfaces associated with $\Pone$}\label{suse:PI-Z-Segre}

Consider the general embedded affine cubic for $\Pone$  given by
\begin{equation*}
    x_1x_2x_3-x_1-x_2+1=0.
\end{equation*}
It is smooth and contains $5$ affine lines, given by
\begin{align*}
&L_1^{\text{I}}=\{x_1=0, x_2=1\} & 
&L_4^{\text{I}}=\{x_2=1,x_3=1\}, \\
&L_2^{\text{I}}=\{x_1=1, x_2=0\}, & 
&L_5^{\text{I}}=\{x_3=0,x_1+x_2=1\}, \\
&L_3^{\text{I}}=\{x_1=1,x_3=1\}. & 
&
\end{align*}

We follow the construction in the Section \ref{subsec:blowdown} with respect to the line $L_3^\infty$ at infinity and introduce variables
\begin{equation*}
    y_1=x_1,\quad y_2=x_2,\quad y_3=x_3,\quad y_4=x_1 x_2,
\end{equation*}
to obtain the corresponding affine Segre surface $\mathcal{Y}^\text{I}\subseteq\mathbb{C}^4$, defined by
\begin{align*}
    &y_1y_2-y_4=0,\\
    &y_3y_4-y_1-y_2+1=0.
\end{align*}

We obtain the canonical projective completion $\overline{\mathcal{Y}}^\text{I}\subseteq\mathbb{P}^4$ through homogeneous coordinates \eqref{eq:homcoordinatesy}, and the curve at infinity is given by
\begin{equation*}
    Y_0=0,\quad Y_1 Y_2=0,\quad Y_3 Y_4=0.
\end{equation*}
This curve is the product of four lines, intersecting in four points, forming a rectangle, as in Figure \ref{fig:triangle_segre11D}. The four intersection points are given by
\begin{equation*}
p_k^\text{I}:\quad Y_k=1,\quad Y_j=0\quad (0\leq j\leq 4, j\neq k),
\end{equation*}
where $1\leq k\leq 4$. The point $p_3^\text{I}$ is an $A_2$ singularity and the other three are regular points in $\overline{\mathcal{Y}}^\text{I}$.
Note, furthermore, that $x\rightarrow y$ maps $L_k^{\text{I}}$, $1\leq k\leq 4$ to affine lines on $\mathcal{Y}^\text{I}$, but $L_5^{\text{I}}$ is mapped to a conic. Thus $\overline{\mathcal{Y}}^\text{I}$ is a Segre surface with one $A_2$ singularity and $8$ lines, four of which lie at infinity, as indicated in the corresponding entry of Table \ref{table:segre_summary}. 
Setting
\begin{align*}
    z_1&=-y_4, &  z_4&=1,\\
    z_2&=y_3-1, & z_5&=y_1-1,\\
    z_3&=1-y_1-y_2+y_4, & z_6&=y_2-1,
\end{align*}
and vice-versa
\begin{align*}
y_1=z_5+1,& & y_2=z_6+1, & & y_3=z_2+1, & & y_4=-z_1,
\end{align*}
we obtain that ${\mathcal{Y}}^\text{I}$ is isomorphic to the following $\mathcal Z$--Segre:
\begin{subequations}\label{p1:segalt}
 \begin{align}
&{z}_{1} + {z}_{3} + {z}_{4} + {z}_{5} + {z}_{6}=0\\
&  z_4-1=0\\
 &{z}_{3} {z}_{4} - {z}_{1} {z}_{2} =0\\
 &{z}_{3} {z}_{4} - {z}_{5} {z}_{6} =0.
 \end{align}
 \end{subequations}

\subsection{Cubic and Segre surfaces associated with $\Pfive^{(deg)}$}\label{suse:PVdeg-Z-Segre} 

Looking at Table \ref{tb:monmfd}, we see that the $\Pthree^{D_6}$ and the $\Pfive^{(deg)}$ coincide, even though the parameters $\upsilon_i$ and $\nu_j$ in the two cases don't. Denoting the $\Pthree^{D_6}$ variables and parameters with the index ${}^{III}$ and the ones of $\Pfive^{(deg)}$ with an index ${}^{V_d}$ is easy to see that setting
$$
x_1^{V_d}= \frac{x_1^{III}}{\sqrt{\upsilon_0^{III}\upsilon_\infty^{III}}},
\quad x_2^{V_d}= \frac{x_2^{III}}{\sqrt{\upsilon_0^{III}\upsilon_\infty^{III}}},
\quad x_3^{V_d}= x_3^{III},
$$
$$
\upsilon_0^{V_d}=
\sqrt{\upsilon_0^{III}\upsilon_\infty^{III}},\quad \upsilon_t^{V_d}=\sqrt{\frac{\upsilon_\infty^{III}}{\upsilon_0^{III}}}
$$
we map the $\Pfive^{(deg)}$  cubic to the $\Pthree^{D_6}$ one. This shows that the two $\mathcal Z$--Segre surfaces coincide.

We can also produce the $\Pfive^{deg}$ by blowing down the $\Pfive^{deg}$ cubic
$$
x_1 x_2 x_3+ x_1^2 + x_2^2 + \omega_1 x_1+ \omega_2 x_2,
$$
at the line $X_0=X_3=0$. By an linear affine transformation (see mathematica file), this leads to
\begin{subequations}\label{p5deg:seg}
\begin{align}
& z_1+ {z}_{3} + {z}_{4} + {z}_{5} + {z}_{6}=0\\
&   \rho_{3} {z}_{3}+ {z}_{4} + \rho_{5} {z}_{5}+ \rho_{6} {z}_{6}=0\\
&{z}_{3} {z}_{4} - {z}_{1} {z}_{2}=0\\
&{z}_{3} {z}_{4} - {z}_{5} {z}_{6} =0,
\end{align}
\end{subequations}
where $\rho_6= \frac{\rho_3}{\rho_5}$.
The family \eqref{p5deg:seg} only depends on two independent parameters.

\subsection{Cubic and Segre surfaces associated with $\Pthree^{D_7}$}\label{suse:PIIID7-Z-Segre}
The affine cubic surface associated with $\Pthree^{D_7}$ is defined by
\begin{equation*}
    x_1x_2x_3+x_1^2+x_2^2+\omega_1x_1-x_2=0,
\end{equation*}
with $\omega_1\in\mathbb{C}^*$.
It is smooth for any value of $\omega_1$ and reduces to the cubic surface for $\Pthree^{D_8}$ as $\omega_1\rightarrow 0$.

The embedded surface contains $7$ affine lines, given by
\begin{align*}
&L_1=\{x_1=0,\, x_2=0\}, & 
&L_5=\{x_3+\omega_1+\omega_1^{-1}=0,\,\omega_1x_1-x_2=0\}, \\
&L_2=\{x_1=0,\, x_2=1\}, & 
&L_6=\{x_3+\omega_1+\omega_1^{-1}=0,\,x_1-\omega_1x_2+\omega_1=0\}, \\
&L_3=\{x_1+\omega_1=0,\, x_2=0\}, & 
&L_7=\{x_1+\omega_1=0,\,1-x_2+\omega_1x_3=0\}, \\
&L_4=\{x_2=1,\, x_1+x_3+\omega_1=0\}, &
\end{align*}
and three lines at infinity.

We follow the construction in Section \ref{subsec:blowdown} with respect to the line $L_1^\infty$ at infinity and introduce variables
\begin{equation*}
    y_1=x_1,\quad y_2=x_2,\quad y_3=x_3,\quad y_4=x_2 x_3,
\end{equation*}
to obtain the corresponding affine Segre surface $\mathcal{Y}\subseteq\mathbb{C}^4$, defined by
\begin{align*}
    &y_2y_3-y_4=0,\\
    &y_1y_4+y_1^2+y_2^2+\omega_1y_1-y_2=0.
\end{align*}

We obtain the canonical projective completion $\overline{\mathcal{Y}}\subseteq\mathbb{P}^4$ through homogeneous coordinates \eqref{eq:homcoordinatesy}, and the curve at infinity is given by
\begin{equation*}
    Y_0=0,\quad Y_2 Y_3=0,\quad Y_1(Y_1+Y_4)=0.
\end{equation*}
This curve is the product of two lines and a conic, intersecting in three points, see Figure \ref{fig:triangle_segre11C}. The point where the two lines at infinity intersect,
\begin{equation*}
    [Y_0:Y_1:Y_2:Y_3:Y_4]=[0:0:0:1:0],
\end{equation*}
is an $A_2$ singularity, and $\overline{\mathcal{Y}}$ is smooth elsewhere.

Note, furthermore, that $x\rightarrow y$ maps the lines $L_k$, $1\leq k\leq 6$ to affine lines on $\mathcal{Y}$, but $L_7$ is mapped to a conic. Thus $\overline{\mathcal{Y}}$ is a Segre surface with one $A_2$ singularity and $8$ lines, two of which lie at infinity, as indicated in the corresponding entry of Table \ref{table:segre_summary}. 
Setting
\begin{align*}
    z_1&=\omega_1^{-1}y_1+y_2+\omega_1^{-1}y_4, &  z_4&=1-y_2,\\
    z_2&=-1-\omega_1^{-1}y_1+y_2, & z_5&=\omega_1^{-2}y_2,\\
    z_3&=-((1+\omega_1^{-2})y_2+\omega_1^{-1}y_4), & z_6&=-1+(1+\omega_1^2)(y_2-1)-\omega_1y_3+\omega_1y_4,
\end{align*}
we obtain the corresponding $\mathcal Z$--Segre surface 
 \begin{align*}
&{z}_{1} + {z}_{2} + {z}_{3} + {z}_{4} + {z}_{5}=0,\\
&  z_4+\omega_1^2 z_5-1=0,\\
 &{z}_{3} {z}_{4} - {z}_{1} {z}_{2} =0,\\
 &{z}_{3} {z}_{4} - {z}_{5} {z}_{6} =0.
 \end{align*}
This is the entry in Table \ref{tb:Z-Segre} for $\Pthree^{D_7}$ with $\rho_5=\omega_1^2$.


\section{Poisson structure}

Any algebraic variety defined as zero set of $n-2$ polynomials in $n$ variables is endowed by a natural Poisson bracket that was introduced by Nambu \cite{nambu}. \\tc{On surfaces, this bracket coincides with the one induced by the natural symplectic structure \cite{tak}.}

\tcm{In the case of the Painlev\'e differential equations, for each equation, there are three surfaces:  the monodromy manifold $\mathcal X^{(d)}$, its blow down to  $\mathcal Y^{(d)}$ and the isomorphic Segre surface $\mathcal Z^{(d)}$}. 

\tcm{In this section, we give explicit formulae for the natural Poisson brackets on each of these surfaces and show that the blow down and the isomorphism preserve the natural symplectic structure and are therefore Poisson maps}.

\subsection{qPVI}\label{subsec:qpvi_poisson}
For the Segre surface $\mathcal Z_q$ the Nambu Poisson bracket is defined by the following formula
\begin{equation}\label{eq:NT}
    \{f,g\}_{{}_{\mathcal Z}}:= \frac{df\wedge dg\wedge dh_1\wedge dh_2\wedge dh_3\wedge dh_4}{dz_1\wedge dz_2\wedge dz_3\wedge dz_4\wedge dz_5\wedge dz_6}.
\end{equation}
Applying this formula to $h_1,h_3,h_4$ defined by \eqref{qp6:segh}, and
 $h_2''$ defined in \eqref{eqp-h2-rho} in  \eqref{eq:NT}, we obtain the following:
{\allowdisplaybreaks
\begin{align}\label{eq:NTZ}
\{z_1,z_2\}_{{}_{\mathcal Z}}=& \lambda_2 ( z_3 ( z_5 (- \rho_3+ \rho_5)+ z_6 ( \rho_3- \rho_6))+ z_4 ( z_5- z_5
 \rho_5+ z_6 (-1+ \rho_6))),\\
\{z_1,z_3\}_{{}_{\mathcal Z}}=& z_3 \lambda_2 ( z_5 ( \rho_2- \rho_5)- z_6 (\rho_2- \rho_6))+ z_1 \lambda_1 ( z_5\lambda_2 (1- \rho_5)- z_6 \lambda_2 (1- \rho_6)+\nn\\
&+
z_3 (\rho_6-\rho_5)),\nn\\
\{z_1,z_4\}_{{}_{\mathcal Z}}=&  z_1 \lambda_1 \lambda_2 ( z_5 ( \rho_5- \rho_3)- z_6 ( \rho_6- \rho_3))+ z_4
( z_5 \lambda_2 ( \rho_5- \rho_2)+ z_6 \lambda_2 ( \rho_2- \rho_6)+\nn\\
&+ z_1 \lambda_1 ( \rho_5- \rho_6)),\nn\\
\{z_1,z_5\}_{{}_{\mathcal Z}}=& z_5 \lambda_2 ( z_4 (\rho_2-1)+ z_3 ( \rho_3- \rho_2))+ z_1 \lambda_1
( z_5 \lambda_2 ( \rho_3-1)+ z_3 ( \rho_3- \rho_6)+\nn \\
& z_4 ( \rho_6-1)),\nn\\
  \{z_1,z_6\}_{{}_{\mathcal Z}}=&z_6 \lambda_2 ( z_4- z_4  \rho_2+ z_3 ( \rho_2- \rho_3))+ z_1 \lambda_1 ( z_4+ z_6 \lambda_2- z_3  \rho_3- z_6 \lambda_2  \rho_3+\nn\\ 
  & z_3  \rho_5- z_4
 \rho_5),\nn\\
   \{z_2,z_3\}_{{}_{\mathcal Z}}=& z_3 \lambda_2 ( z_5  \rho_5- z_6  \rho_6)+ z_2 \lambda_1 ( z_6  \lambda_2+ z_5 \lambda_2 ( \rho_5-1)+ z_3  \rho_5- z_3  \rho_6- z_6 \lambda_2 \rho_6)  ,\nn\\
         \{z_2,z_4\}_{{}_{\mathcal Z}}=& z_4 \lambda_2 ( z_6  \rho_6- z_5  \rho_5)+ z_2 \lambda_1 ( z_5 \lambda_2 ( \rho_3- \rho_5)+ z_6 \lambda_2 ( \rho_6- \rho_3)+ z_4 (\rho_6- \rho_5)),\nn\\
  \{z_2,z_5\}_{{}_{\mathcal Z}}=& z_5 \lambda_2 ( z_4- z_3  \rho_3)+ z_2 \lambda_1 ( z_4+ z_5 \lambda_2- z_3
 \rho_3- z_5 \lambda_2  \rho_3+ z_3  \rho_6- z_4  \rho_6),\nn\\ 
  \{z_2,z_6\}_{{}_{\mathcal Z}}=& z_6 \lambda_2 (z_3  \rho_3- z_4)+ z_2 \lambda_1 ( z_6 \lambda_2 (\rho_3-1)+ z_3 ( \rho_3- \rho_5)+ z_4 ( \rho_5-1)),\nn\\ 
  \{z_3,z_4\}_{{}_{\mathcal Z}}=&\lambda_1 \lambda_2 ( z_2  z_5 (- \rho_2+ \rho_5)+ z_2  z_6 ( \rho_2-\rho_6)+ z_1 (z_6  \rho_6- z_5  \rho_5))
 ,\nn\\ 
  \{z_3,z_5\}_{{}_{\mathcal Z}}=& z_3  z_5 \lambda_2  \rho_2+ z_2 \lambda_1 ( z_5 \lambda_2 (\rho_2-1)+ z_3
( \rho_2- \rho_6))+ z_1 \lambda_1 ( z_5 \lambda_2+ z_3  \rho_6),\nn\\ 
  \{z_3,z_6\}_{{}_{\mathcal Z}}=& - z_3  z_6 \lambda_2  \rho_2- z_1 \lambda_1 ( z_6 \lambda_2+ z_3 \rho_5)+ z_2 \lambda_1 ( z_6 \lambda_2- z_3  \rho_2- z_6 \lambda_2  \rho_2+ z_3
 \rho_5),\nn\\
  \{z_4,z_5\}_{{}_{\mathcal Z}}=&
   z_4  z_6 \lambda_2  \rho_2+ z_1  z_6 \lambda_1 \lambda_2  \rho_3+ z_2
\lambda_1 ( z_6 \lambda_2 ( \rho_2- \rho_3)+ z_4 ( \rho_2- \rho_5))+ z_1  z_4
\lambda_1  \rho_5
  ,\nn\\
  \{z_5,z_6\}_{{}_{\mathcal Z}}=& \lambda_1 ( z_2 ( z_4- z_4  \rho_2+ z_3 ( \rho_2- \rho_3))+ z_1 (- z_4+ z_3
 \rho_3)),\nn
 \end{align}}
By construction, this bracket  defines a Poisson bracket (namely it satisfies the Jacobi identity \cite{tak}) on $\mathbb C[z_1,\dots,z_6]$ and the functions $h_1,h_2'',h_3, h_4$ are central elements. This means that it can be restricted to the Segre surface
\tcm{$$
\mathcal Z^{(d)}_q=\hbox{Spec}\left( \mathbb{C}[z_1,\dots,z_6]\slash\langle h_1,h_2'',h_3, h_4\rangle\right),
$$
where for each $d$, the polynomials $h_1,h_2'',h_3, h_4$ are given in Table \ref{tb:Z-Segre}.}

\begin{lemma}
All the central elements of the Poisson bracket \eqref{eq:NTZ} are algebraically dependent on $h_1,\dots,h_4$.
\end{lemma}

\begin{proof}
The proof of this statement can be extracted from Section 3 of \cite{OdR}. In our case the proof can also be done directly by observing that
     thanks to the proof of Proposition \ref{prop:qpvialgebraic}, the Jacobian matrix of $h_1,\dots,h_4$ has maximal rank $4$. Then 
because any $f$ is central element must satisfy$\{f,z_i\}_{\mathcal Z}=0$ for all $i=1,\dots,6$, it is easy to prove that 
$$
\frac{\partial f}{\partial z_j} = \sum_{i=1}^4 \alpha_i  \frac{\partial h_i}{\partial z_j} , \quad j=1,2,\dots,6,
$$
for some polynomials $\alpha_1,\dots,\alpha_4$ that are constant on the Segre surface.
We
conclude that $f$ must be a algebraically dependent on $h_1,\dots,h_4$.
\end{proof}

\subsection{Poisson structure on the monodromy manifolds of the differential equations and their blow down.}
In this subsection we prove the following 
\begin{lemma}\label{lm:poisson}
    For each Painlev\'e differential equation, the blow down of its monodromy manifold $\mathcal X^{(d)}$, defined as zero set of the polynomial \tcm{$\phi^{(d)}$} defined in \eqref{eq:phi}, is a Poisson map.
\end{lemma}

\begin{proof} 
    Observe that $\mathcal X^{(d)}$ admits the 
natural Poisson bracket 
defined by:
\begin{equation}\label{eq:nambuX}
\{x_1,x_2\}_{{}_{\mathcal X}}
=\frac{\partial\phi^{(d)}}{\partial x_3},\quad \{x_2,x_3\}_{{}_{\mathcal X}}=\frac{\partial\phi^{(d)}}{\partial x_1},\quad 
\{x_3,x_1\}_{{}_{\mathcal X}}=\frac{\partial\phi^{(d)}}{\partial x_2}.
\end{equation}
Let us choose the following blow down:
\begin{equation}\label{gen-Segre-Y1}
\mathcal Y^{(d)}:=\{(y_1,y_2,y_3,y_4)\in\mathbb C^4|  y_4 - y_2 \tcm{y_3},  y_1 y_4 +  \sum_{i=1}^3(\epsilon_i^{(d)} i^2+\omega_i^{(d)} y_i)  + \
\omega_4^{(d)}=0\}.
\end{equation}
Then, the Nambu Poisson bracket on the surface $\mathcal Y^{(d)}$ is defined as follows
$$
\{f,g\}_{{}_{\mathcal Y}}=\frac{df\wedge dg\wedge d\Phi_1^{(d)}\wedge d \Psi_1}{dy_1\wedge d y_2\wedge dy_3\wedge dy_4}.
$$
Therefore, for $i<j$, 
\begin{equation}\label{eq:PBY}
    \{y_i,y_j\}_{{}_{\mathcal Y}} =\left[\begin{array}{cccc}
    0 & -y_3 & -y_2 & 1  \\
    y_4 +2 \epsilon_1^{(d)} y_1 + \omega_1^{(d)}& 
    2 \epsilon_2^{(d)} y_2 + \omega_2^{(d)}&2 \epsilon_3^{(d)} y_3 + \omega_3^{(d)} & y_1\\
\end{array}\right]_{ij},
\end{equation}
where $[\,]_{ij}$ denotes the determinant of the matrix obtained erasing the i-th and the j-th columns. Using this formula, it is a straightforward computation to prove that 
$$
\{y_i,y_j\}_{{}_{\mathcal Y}}=-\{y_i(x_1,x_2,x_3),y_j(x_1,x_2,x_3)\}_{\mathcal X},
$$
hence proving that the blow down map is Poisson.
\end{proof}

\subsubsection{Poisson bracket on the $\mathcal Z$--Segre.}

In this subsection we prove the following 

\begin{proposition}\label{prop:poisson}
    For each Painlev\'e differential equation, the linear affine transformation 
    $$
    \Phi^{(d)}:\mathcal Y^{(d)}\to \mathcal Z^{(d)},
    $$
    where $d$ is an index that runs trough the list of the Painlev\'e equations, 
    is \tcm{a Poisson map} with Poisson inverse.
\end{proposition}

\begin{proof}
    We already know that for each Painlev\'e equation $\Phi^{(d)}$ is invertible. We only need to prove that Poisson relations are sent to Poisson relations. We start by observing that
the Poisson structure on the Segre $\mathcal Z_1$ is defined in the same way as for the $\mathcal Z_q$, namely by \eqref{eq:NT}. The properties of this bracket are the same for 
all members of Table \ref{tb:Z-Segre}, so that we obtain a unified formula that by specializing the parameters $\epsilon^{(d)}_i$ and  $\rho^{(d)}_i,\lambda_i^{(d)}$
gives the Poisson bracket on $\mathcal Z^{(d)}$:
\begin{equation}
\begin{split}
\{z_1,z_2\}_{{}_{\mathcal Z}}&= \lambda_2^{(d)} z_3\left(z_5 (\rho_5^{(d)}-\rho_3^{(d)})+z_6 (\epsilon_6^{(d)} \rho_3^{(d)}-\epsilon_4^{(d)}\rho_6^{(d)})\right)\\
&+\lambda_2^{(d)} z_4\left(z_5  (1-\epsilon_4^{(d)}\rho_5^{(d)})+z_6 (\rho_6^{(d)}-\epsilon_6^{(d)})\right),\\
\{z_1,z_3\}_{{}_{\mathcal Z}}&= 
z_1\lambda_1^{(d)}\left(z_5\lambda_2^{(d)}(1-\epsilon_4^{(d)}\rho_5^{(d)})+z_3(\rho_6^{(d)}-\epsilon_6^{(d)} \rho_5^{(d)}) +\right.\\
&+\left. z_6\lambda_2^{(d)}(\epsilon_4^{(d)}\rho_6^{(d)}-\epsilon_6^{(d)})\right)
+z_3\epsilon_2^{(d)}\lambda_2^{(d)}(\rho_6^{(d)} z_6-\rho_5^{(d)} z_5),\\
\{z_1,z_4\}_{{}_{\mathcal Z}}&=\lambda_1^{(d)} z_1\left(z_5 \lambda_2^{(d)}(\rho_5^{(d)}-\rho_3^{(d)})+
z_4  (\epsilon_6^{(d)}\rho_5^{(d)}-\epsilon_4^{(d)}\rho_6^{(d)})
\right.
\\
&\left.+ z_6 \lambda_2^{(d)}(\epsilon_6^{(d)}\rho_3^{(d)}-\rho_6^{(d)})\right)+z_4\lambda_2^{(d)}\epsilon_2^{(d)}(z_5 \rho_5^{(d)}-z_6 \rho_6^{(d)}),\\
\{z_1,z_5\}_{{}_{\mathcal Z}}&= \lambda_1^{(d)} z_1\left (z_5 \lambda_2^{(d)}(\epsilon_4^{(d)}\rho_3^{(d)}-1)
+z_3 (\epsilon_6^{(d)}\rho_3^{(d)}-\rho_6^{(d)})+ z_4(\epsilon_4^{(d)}\rho_6^{(d)}-\epsilon_6^{(d)})
  \right)\\
   &+ z_5\epsilon_2^{(d)} \lambda_2^{(d)}(z_3\rho_3^{(d)}-z_4),\\
   \{z_1,z_6\}_{{}_{\mathcal Z}}&=  \lambda_1^{(d)} z_1\left (z_6\lambda_2^{(d)}(1-\epsilon_4^{(d)}\rho_3^{(d)})
+z_3(\rho_5^{(d)}-\rho_3^{(d)})+ z_4(1-\epsilon_4^{(d)}\rho_5^{(d)})
  \right)\\
   &+ z_6\epsilon_2^{(d)} \lambda_2^{(d)}(z_4-z_3\rho_3^{(d)}),\\
   \{z_2,z_3\}_{{}_{\mathcal Z}}&= \lambda_1^{(d)} z_2\left (z_5 
   \lambda_2^{(d)}(\epsilon_4^{(d)}\rho_5^{(d)}-1)
+z_3(\epsilon_6^{(d)}\rho_5^{(d)}-\epsilon_4^{(d)}\rho_6^{(d)})+\right.\\
   &+\left. z_6 \lambda_2^{(d)}(\epsilon_6^{(d)}-\rho_6^{(d)})
  \right)+ z_3\epsilon_1^{(d)} \lambda_2^{(d)}(z_5\rho_5^{(d)}-z_6\rho_6^{(d)}),\\
   \{z_2,z_4\}_{{}_{\mathcal Z}}&=  \lambda_1^{(d)} z_2\left(z_5   \lambda_2^{(d)}(\rho_3^{(d)}-\rho_5^{(d)})+z_4 (\rho_6^{(d)}-\epsilon_6^{(d)}\rho_5^{(d)})+z_6 \lambda_2^{(d)} (\rho_6^{(d)}-\epsilon_6^{(d)}\rho_3^{(d)})\right)\\
   &+z_4\epsilon_1^{(d)}\lambda_2^{(d)}(z_6\rho_6^{(d)}-z_5\rho_5^{(d)}),\\
    \{z_2,z_5\}_{{}_{\mathcal Z}}&= \lambda_1^{(d)}z_2\left(z_5 \lambda_2^{(d)}(1-\epsilon_4^{(d)}\rho_3^{(d)})+z_4(\epsilon_6^{(d)}-\epsilon_4^{(d)}\rho_6^{(d)})+
    z_3(\rho_6^{(d)}-\epsilon_6^{(d)}\rho_3^{(d)})\right)\\
    &+z_5\epsilon_1^{(d)}\lambda_2^{(d)}(z_4-z_3\rho_3^{(d)}),\\
     \{z_2,z_6\}_{{}_{\mathcal Z}}&= \lambda_1^{(d)} z_2\left( z_6\lambda_2^{(d)}(\epsilon_4^{(d)}\rho_3^{(d)}-1)+z_3 (\rho_3^{(d)}-\rho_5^{(d)})+z_4 (\epsilon_4^{(d)}\rho_5^{(d)}-1)\right)\\
     &+z_6\epsilon_1^{(d)}\lambda_2^{(d)}(z_3\rho_3-z_4), \\
     \{z_3,z_6\}_{{}_{\mathcal Z}}&=\lambda_1^{(d)}(z_2\epsilon_2^{(d)}-z_1\epsilon_1^{(d)})(z_6 \lambda_2+z_3 \rho_5^{(d)}),\\
     \{z_4,z_5\}_{{}_{\mathcal Z}}&= \lambda_1(z_2\epsilon_2^{(d)}-z_1\epsilon_1^{(d)})(z_5 \lambda_2^{(d)}\rho_3^{(d)}+z_4 \rho_6^{(d)}),\\
     \{z_4,z_6\}_{{}_{\mathcal Z}}&= -\lambda_1(z_2\epsilon_2^{(d)}-z_1\epsilon_1^{(d)})(z_6 \lambda_2^{(d)}\rho_3^{(d)}+z_4 \rho_5^{(d)}),\\
      \{z_5,z_6\}_{{}_{\mathcal Z}}&= \lambda_1^{(d)}(z_2\epsilon_2^{(d)}-z_1\epsilon_1^{(d)})(z_4-z_3\rho_3^{(d)}).
\end{split}
\end{equation}
Using these formulae, by direct comparison with \eqref{eq:PBY}, we obtain
$$
\{z_i,z_j\}_{{}_{\mathcal Z}}= \frac{\upsilon_1\upsilon_t(\upsilon_0-1)^2(\upsilon_\infty^2-1)}{\upsilon_0\upsilon_\infty(\upsilon_1-1)^2(\upsilon_t^2-1)}\{z_i(y_1,\dots,y_4),z_j(y_1,\dots,y_4)\}_{{}_{\mathcal Y}},
$$
hence proving that $\Phi^{(d)}$ is Poisson. A similar computation can be done for the inverse.
\end{proof}

\tcm{\begin{remark}
    One of the anonymous referees rightly pointed out that the Poisson structures on the surfaces arise naturally from symplectic forms determined by the geometry of their projective completions. In the case of $\Psix$, since the projective completion $\overline{\mathcal Z}_1$ is a complete intersection of two quadrics, the adjunction formula gives that its canonical divisor is $-\overline{\mathcal Z}_1\setminus{\mathcal Z}_1$. This implies the existence of a rational 2-form $\Omega_{\mathcal Z}$ on $\overline{\mathcal Z}_1$ with simple poles along  $\overline{\mathcal Z}_1\setminus{\mathcal Z}_1$ which restricts to a non-degenerate holomorphic $2$-form 
on the affine surface ${\mathcal Z}_1$. This form induces the Poisson bracket  \eqref{eq:NTZ}. An analogous construction holds for $\overline{\mathcal Y}$, which therefore admits a rational 2-form $\Omega_{\mathcal Y}$ on $\overline{\mathcal Y}$ with simple poles along  $\overline{\mathcal Y}\setminus{\mathcal Y}$ which restricts to a non-degenerate holomorphic $2$-form 
on the affine surface ${\mathcal Y}$ and gives rise to the Poisson bracket \eqref{eq:PBY}. Since the isomorphism $\Phi:\mathcal Y\to\mathcal Z_1$ is linear, it preserves the symplectic form up to a constant factor.
Similarly, the projective completion $\overline{\mathcal X}$ of the monodromy manifold ${\mathcal X}$ of $\Psix$ admits a rational $2$--form $\Omega_{\mathcal X}$ analogous to ones just discussed. The 
blow-down map $\pi:\overline{\mathcal X}\to\overline{\mathcal Y}$ contracts a line contained in the boundary and restricts to an isomorphism $ {\mathcal X}\to{\mathcal Y}$ and one has that $\Omega_{\mathcal X}=\pi*\Omega_{\mathcal Y}$. This implies the statements of Lemma \ref{lm:poisson} and Proposition \ref{prop:poisson}. 
A similar discussion could be repeated for each Painlev\'e equation, hence making these statements trivial. We nonetheless keep them as they are explicit, and therefore might be useful. 
\end{remark}}

\section{Conclusion}
In this paper, we showed that the monodromy manifold of $\tcm{q}\Psix$ gives rise to a number of new results related to monodromy manifolds and symplectic geometry of the differential Painlev\'e equations. First, the continuum limit gives rise to a Segre surface which is isomorphic to the Jimbo-Fricke cubic surface well known to be the monodromy manifold of $\Psix$. Second, we use confluence limits to obtain Segre surfaces, called $\mathcal Z$-Segre surfaces, that are isomorphic to the conventional monodromy manifolds of each Painlev\'e equation. Blow-downs of the latter give rise to other surfaces, which we call $\mathcal Y$-Segre surfaces. We show that the $\mathcal Z$-Segre surfaces and $\mathcal Y$-Segre surfaces are affinely equivalent and the linear affine transformation between them is Poisson. It is interesting to note that the blow-down maps from each Painlev\'e equation's cubic monodromy manifold to the $\mathcal Z$- and $\mathcal Y$-Segre surfaces are Poisson maps. An interesing open question is to ask whether Segre surfaces exist as monodromy manifolds for all remaining equations in Sakai's diagram.

\end{document}